%% file: main_clean.tex
\documentclass[10pt,a4paper]{amsart}
 
\usepackage{anysize}
\marginsize{20mm}{20mm}{20mm}{15mm}
\usepackage[utf8]{inputenc}
\usepackage[T1]{fontenc}
\usepackage{amsmath,amsfonts,mathrsfs,bm}

\usepackage[table]{xcolor}

\usepackage[pdftex]{graphicx}
\usepackage[title]{appendix}
\usepackage{amsthm}
\usepackage[vlined,boxruled]{algorithm2e}
\usepackage{enumitem}
\usepackage{blindtext}
\usepackage{mdframed}
\usepackage[sort,square,numbers]{natbib}
\usepackage{multicol}
\usepackage{hyperref}
\usepackage{cleveref}

\usepackage{dsfont}

\usepackage[justification=centering]{caption}
\usepackage{subcaption}

\usepackage{booktabs}

\usepackage{tikz}
\usetikzlibrary{arrows.meta}
\usetikzlibrary{positioning}
\usetikzlibrary{shapes}

\usepackage{array}
\newcolumntype{L}[1]{>{\raggedright\let\newline\\\arraybackslash\hspace{0pt}}m{#1}}
\newcolumntype{C}[1]{>{\centering\let\newline\\\arraybackslash\hspace{0pt}}m{#1}}
\newcolumntype{R}[1]{>{\raggedleft\let\newline\\\arraybackslash\hspace{0pt}}m{#1}}

\usepackage{multirow,multicol}

\usepackage{float}
\usepackage{pgffor}

\usepackage[section]{placeins}

\newmdenv{allfour}
\newmdenv[leftline=false,rightline=false]{topbot}
\newmdenv[topline=false,rightline=false]{leftbot}
\newmdenv[topline=false,leftline=false]{rightbot}
\newmdenv[topline=false,rightline=false,leftline=false]{bottom}

\newtheorem{theorem}{Theorem}[section]

\newtheorem{prop}[theorem]{Proposition}
\newtheorem{corol}[theorem]{Corollary}

\theoremstyle{remark}
\newtheorem{rem}[theorem]{Remark}

\theoremstyle{definition}

\providecommand{\keywords}[1]{\textbf{\textit{Keywords -- }} #1}

\newcounter{defcounter}
\newcommand{\R}{\mathbb{R}}
\newcommand{\peq}{\;.}
\newcommand{\veq}{\;,}

\newcommand{\ind}{\mathds{1}}

\newcommand{\bi}{[\![}
\newcommand{\ei}{]\!]}
\newcommand{\argmin}{\mathop{\mathrm{argmin}}}

\newcommand{\m}{\mathrm{m}}
\newcommand{\lgt}{\mathrm{lgt}}
\newcommand{\J}{\mathrm{Jac}}

\newcommand{\q}[1]{``#1''}

\newcommand{\dd}{\mathrm{d}}
\newcommand{\di}{\,\mathrm{d}}

\newcommand{\pv}{\,;\,}

\SetAlgoSkip{bigskip}

\definecolor{myred}{HTML}{B83232}
\definecolor{myblue}{HTML}{28568A}
\definecolor{mygreen}{HTML}{288A42}

\definecolor{darkgreen}{rgb}{0,.6,0}

\usepackage{ulem}
\begin{document}

\title[Parameter and density estimation from real-world traffic data]{Parameter and density estimation from real-world traffic data:
A kinetic compartmental approach}

	\author[M.~Pereira]{Mike Pereira} \address[Mike Pereira]{\newline Department of Mathematical Sciences \& Department of Electrical Engineering
	\newline Chalmers University of Technology \& University of Gothenburg
	\newline S--412 96 G\"oteborg, Sweden.} \email[]{mike.pereira@chalmers.se}

	\author[P.~Boyraz Baykas]{Pinar Boyraz Baykas} \address[Pinar Boyraz Baykas]{\newline Department of Mechanics and Maritime Sciences
	\newline Chalmers University of Technology
	\newline S--412 96 G\"oteborg, Sweden.}
	\author[B.~Kulcs\'ar]{Bal\'azs Kulcs\'ar} \address[Bal\'azs Kulcs\'ar]{\newline Department of Electrical Engineering
	\newline Chalmers University of Technology
	\newline S--412 96 G\"oteborg, Sweden.} \email[]{kulcsar@chalmers.se}

	\author[A.~Lang]{Annika Lang} \address[Annika Lang]{\newline Department of Mathematical Sciences
	\newline Chalmers University of Technology \& University of Gothenburg
	\newline S--412 96 G\"oteborg, Sweden.} \email[]{annika.lang@chalmers.se}

	\thanks{
	Acknowledgement. Authors thank G Szederkenyi and Gy Liptak for the inspiring discussions on TRM. 	
	The authors thank for the financial support of the Chalmers AI Research Centre (CHAIR) under the projects STONE, RITE, and SCNN and Transport Area of Advance (Chalmers University of Technology). This work has been partially supported and funded by OPNET (Swedish Energy Agency, 46365-1), the Wallenberg AI, Autonomous Systems and Software Program (WASP) funded by the Knut and Alice Wallenberg Foundation, and the Swedish Research Council (project nr.\ 2020-04170).
}


\begin{abstract}
The main motivation of this work is to assess the validity of a LWR traffic flow model to model measurements obtained from trajectory data, and propose extensions of this model to improve it. A formulation for a discrete dynamical system is proposed aiming at reproducing the evolution in time of the density of vehicles along a road, as observed in the measurements. This system is formulated as a chemical reaction network where road cells are interpreted as compartments, the transfer of vehicles from one cell to the other is seen as a chemical reaction between adjacent compartment and the density of vehicles is seen as a concentration of reactant. Several degrees of flexibility on the parameters of this system, which basically consist of the reaction rates between the compartments, can be considered: a constant value or a function depending on time and/or space. Density measurements coming from trajectory data are then interpreted as observations of the states of this system at consecutive times. Optimal reaction rates for the system are then obtained by minimizing the discrepancy between the output of the system and the state measurements. 
This approach was tested both on simulated and real data, proved successful in recreating the complexity of traffic flows despite the  assumptions on the flux-density relation.
\end{abstract}

\maketitle


\section{Introduction}

Modeling traffic flow to reflect macroscopic vehicular patterns becomes more and more important in the area of connectivity and autonomy \citep{Book_Treiber2013,Book_Kessel,Piccoli2006,Ambrosio2009}. In this regard, proposing  macroscopic traffic flow models capable of reproducing traffic flow patterns observed in real-world setting is a key problem. Such patterns are traditionally observed through data collected from sensors installed on a road (e.g.\ loop detectors) which collect vehicle counts or occupancy times, which in turn are aggregated and filtered to yield density, flux and speed estimates \citep{leduc2008road}. Ongoing progress in image capturing and processing capabilities have permitted to multiply and democratize the use of vehicle trajectory data, which arguably provide a more complete and faithful picture of traffic behaviors since the evolution of each vehicle  can be tracked along the road \citep{lu2007freeway,highDdataset}. We also expect that the penetration of connective vehicular technology and novel sensing and communication systems will further propel the above transition.   

Model-wise, macroscopic traffic flow models play a fundamental role \citep{Piccoli2006} to model network level behavior or management solutions. In particular, first order traffic flow models are predominant and reflect fundamental macroscopic properties such as conservation and flow property \citep{Book_Kessel}. Such models consider three main fundamental quantities defined over time and along the road (space): the density of vehicles at a given location and time, the flux (or flow) of vehicles passing a given location at a given time and the average speed at a given location and time. These three variables are linked to one another by the fundamental relationship of traffic, 
and are captured by first order hyperbolic partial differential equation \citep{Piccoli2006} linking flux and density of vehicles. In the well-studied Lighthill--Witham--Richards (LWR) model, the speed of vehicles is expressed as a function depending on the density only, thus yielding a partial differential equation (PDE) satisfied by this last quantity \citep{lighthill1955kinematic,richards}.

The question of assessing the legitimacy of the hypotheses made by the continuous models is natural, especially when working with traffic flow data obtained from real-word measurements \citep{fan2013data}. Take for instance \Cref{fig:fd_real}, which shows a fundamental diagram, i.e.\ a scatter plot of flow vs density measurements done at the same locations and time, obtained from the real-world trajectory data (cf. \Cref{sec:edie} for more details) of the highD dataset \citep{highDdataset}. Modeling the seemingly complex flux-density relationship observed in this diagram by a simple univariate function, as done in the LWR model, then becomes a questionable choice. Hence the following questions motivated this work: How valid is the choice of a univariate (only density-dependent) flux function when modeling traffic flows from real-world trajectory data? How can we enhance such traffic flow models for them to better reflect the patterns observed in the available data? Answering this questions and learning to mimick the spatial and temporal change of parameters in fundamental diagrams (flux functions) has two major benefits. First, it contributes to improve the acuracy in describing the traffic dynamics (e.g. propagation of jams). Second, it supports traffic management solution by enabling to reach better performance (i.e. traffic control oriented).   

\begin{figure}
\centering
\includegraphics[width=0.5\textwidth]{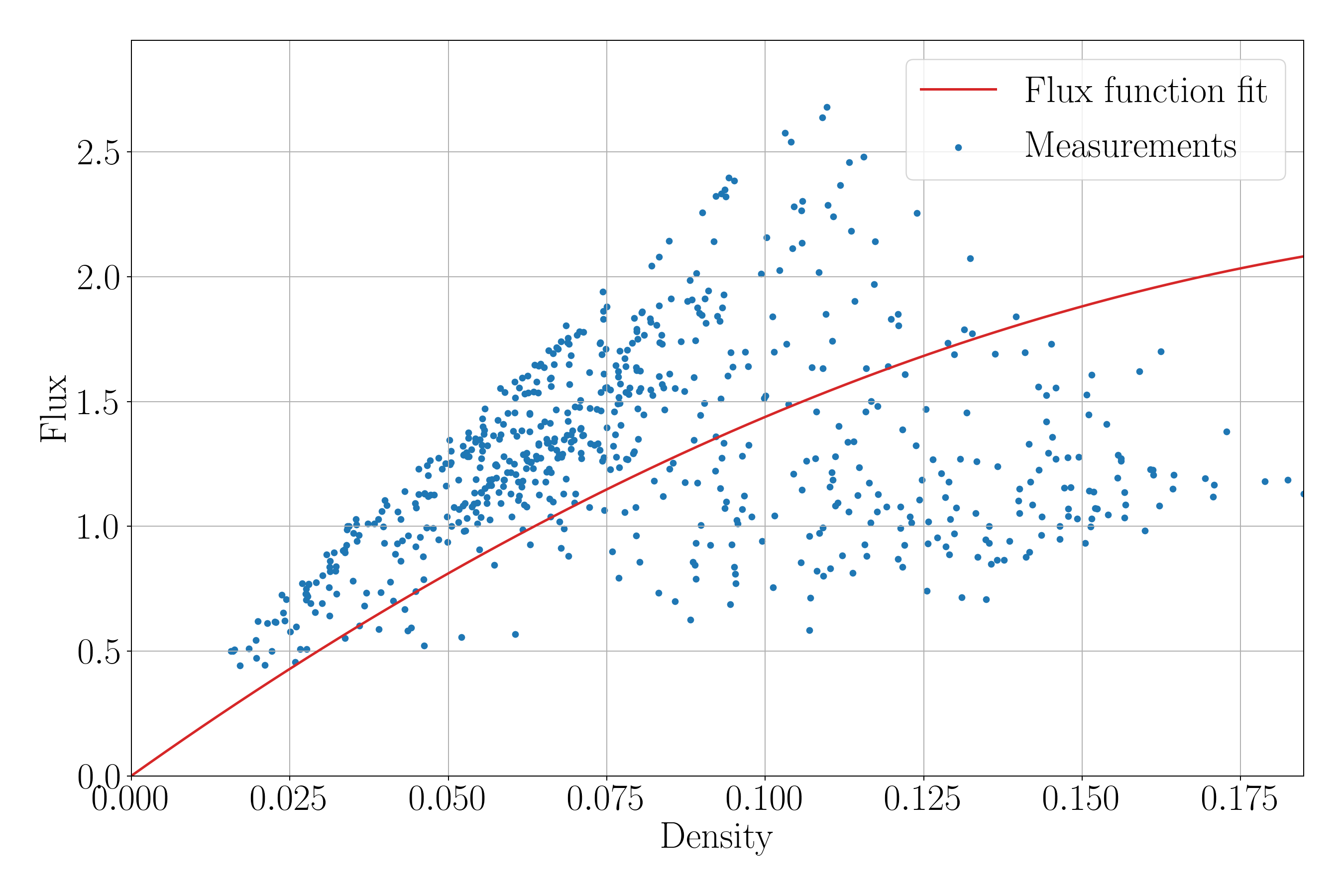}
\caption{Fundamental diagram obtained from trajectory data (highD dataset).}
\label{fig:fd_real}
\end{figure}

In this regard, we first have to bridge  the gap between continuous models and discrete measurements of traffic flows. This can be done by first fitting a flux-density function on a fundamental diagram (i.e.\ a scatter plot flux vs density measurements done at the same locations in space and time). An example of such a fit can be seen in \Cref{fig:fd_real}. Then, the resulting traffic flow PDE is solved (numerically) and finally the solution is compared to the measurements \citep{fan2013data}. In such an approach, the focus is put rather on finding an adequate family of flux functions that would best fit the fundamental diagram, and not on the underlying assumption that it depends only on the density. The first contribution of this work is to propose a different approach, which flips the two steps. In this approach, the gap between continuous models and  trajectory-based measurements is bridged by working on discretized versions of the modeling PDE obtained by finite volume methods, thus extending an approach already used to answer control-related questions linked to these traffic models \citep{karafyllis2019feedback,goatin2016speed,delle2017traffic}. The following two-step comparison is then performed: In the first step, only the density measurements are used. They are compared to the discretized PDEs through a least-square approach, thus allowing to determine the discretized models that best approximate the measurements. In the second step, the flux values resulting from the discretized PDEs are compared to the flux measurements: the fundamental diagrams and the values obtained in both cases are compared. The originality of this approach hence stems from the fact that it evaluates separately the density and flux modeling capabilities of the model.

The second contribution of this work is the interpretation of discretized PDEs as discrete dynamical systems whose parameters are directly linked to those of the original continuous PDE models, thus yielding an approach similar to the cell transmission model of \citep{ctm_daganzo}. Besides, we use in this work a particular discretization scheme, the so-called Traffic Reaction Model (TRM) \citep{liptak2020traffic}, which models flow dynamics along a discretized road as a chemical reaction network, thus allowing us to interpret the parameters of the system (and more generally of the PDE) as chemical reaction rates. Another contribution is that we propose an extension of the TRM taking the form of space and/or time dependent parameters and aiming at enhancing the modeling capabilities of the model. Finally, we validate our approach and draw conclusions on the choice of a LWR model to model real traffic flows through numerical experiments using both synthetic and real data.

The paper is organized as follows: In \Cref{sec:mdl} we recall  the derivation of (first order) continuous models for traffic flow, their discretization using finite volume methods and introduce the traffic reaction model, as well as a proposed extension. In \Cref{sec:estim}, we present the discrete dynamical system and how its parameters can be tuned with the objective of mimicking observed density measurements. Finally, we present in \Cref{sec:num} numerical experiments that were conducted on synthetic and real traffic data.


\section{Continuous traffic models  and their discretization}
\label{sec:mdl}

\subsection{First order macroscopic traffic flow model}

Let us assume that a traffic flow is studied on a one-directional stretch of road. The density of vehicles $\rho$ and the flux (or flow) of vehicles $\phi$ are two continuous quantities defined across space (i.e.\ along the road) and time routinely used to characterize the traffic flow in macroscopic models. Integrating the density function at a time~$t$, across a section $\Delta x$ of the road, gives the count of vehicles in $\Delta x$ at~$t$; and integrating the flux function at a location $x$ of the road, across a time interval $\Delta t$, gives the count of the vehicles crossing $x$ during $\Delta t$. 

For any two locations $x_1<x_2$ on the road and time~$t$, a conservation law can hence be written to express the fact that the variation of the number of vehicles between $x_1$ and $x_2$ is equal\footnote{For sake of data availability and exposition of the method, we assume in this work that no on- or off-ramps are present in the road, meaning that no additional source or sink terms need to be added to the law of conservation of vehicles.} to the difference between the number of vehicles entering this road section at $x_1$ and those leaving the section at $x_2$. This gives the integral representation
\begin{equation}
\frac{\dd}{\dd t}\left(\int_{x_1}^{x_2} \rho(t,x) \di x  \right) = \phi(t,x_1)-\phi(t,x_2), \quad t\ge 0, \quad x_1< x_2 \in\R\veq
\label{eq:cons_law}
\end{equation}
which is equivalent to the following partial differential equation
\begin{equation*}
\frac{\partial \rho}{\partial t}(t,x) + \frac{\partial \phi}{\partial x}(t,x) =0, \quad  t\ge 0, \quad x\in\R\veq
\end{equation*}
under suitable regularity conditions on the functions $\rho$ and $\phi$.

It is common to assume some additional relationship between the density $\rho$ and the flux $\phi$ based on some observed links between the two quantities. For instance, when the density is $0$ (meaning that the road is empty), so should the flux. Similarly, when the density is at its maximal value (corresponding to a bumper-to-bumper traffic), the flux should be zero as well. The so-called  Lighthill--Whitham--Richards (LWR) model \citep{lighthill1955kinematic,richards} in particular stems from these observations by expressing the flux as a function $ f$ of the density (only) as follows
\begin{equation*}
\phi(t,x)=f(\rho(t,x)), \quad t\ge 0, \quad x\in\R \veq
\end{equation*}
where $f$ is a (univariate) function satisfying $f(0)=0$ and $f(\rho_{m})=0$, for $\rho_m$ the maximal value the density can take. The simplest form $f$ can take is arguably the quadratic function defined by
\begin{equation}
f(\rho \pv v_m, \rho_m)=\rho \cdot v_m \left(1-\frac{\rho}{\rho_{\m}}\right),\quad \rho \in [0, \rho_m] \veq
\label{eq:f}
\end{equation}
where $v_m>0$ is a parameter that can be interpreted as the maximal speed achievable by vehicles on the road and the notation $f(\rho \pv v_m, \rho_m)$ is used to mark the fact that $f$ is seen as function of the density $\rho$ depending on the two parameters $v_m$ and $\rho_m$. Then, the conservation law under the LWR model becomes
\begin{equation}
\frac{\partial \rho}{\partial t}(t,x) + \frac{\partial }{\partial x}\left(f(\rho \pv v_m, \rho_m)\right)(t,x) =0, \quad  t\ge 0, \quad x\in\R\peq
\label{eq:pde_rho}
\end{equation}
We will assume that the maximal density $\rho_m>0$ is a known constant, which characterizes the capacity of the road, i.e., the maximal amount of vehicles that can fit on a road section. As such, it can be directly estimated by considering for instance the ratio between the number of lanes and the typical length of a vehicle.  We hence introduce the normalized density function $u$, which will become from now on our main variable of interest:
\begin{equation*}
u=\frac{\rho}{\rho_m} \in [0,1] \peq
\end{equation*}
Then, PDE~\eqref{eq:pde_rho} leads to the following PDE satisfied by $u$:
\begin{equation}
\frac{\partial u}{\partial t}(t,x) + \frac{\partial }{\partial x}\left(f(u \pv v_m)\right)(t,x) =0, \quad  t\ge 0, \quad x\in\R\veq
\label{eq:pde_u}
\end{equation}
where $f$ also denotes, with a slight abuse of notation, the normalized flux function defined by
\begin{equation}
f(u \pv v_m)=u \cdot v_m \left(1-u\right),\quad u \in [0, 1] \peq
\label{eq:f_u}
\end{equation}
For any given bounded initial condition $u_0 = u(0,\cdot)$, the existence and uniqueness of (entropy) solutions of PDE~\eqref{eq:pde_u} on $\R_+\times\R$, $T>0$, is guaranteed in the more general case where the maximal speed is taken to be a function $(t,x)\mapsto v_m(t,x)$ that is bounded and Lipschitz continuous \citep{karlsen2004convergence,chen2005quasilinear}.

\subsection{Finite volume approximations of PDEs}

In general, no closed-form solution of PDE~\eqref{eq:pde_u} is available, and therefore, numerical methods must be used to approximate it. In particular, finite volume schemes have been widely used to compute solutions of the (hyperbolic) PDEs  of the form~\eqref{eq:pde_u}. As we will see, the quantities computed by such schemes relate to integrals of the solution, which are more appropriate {since hyperbolic PDEs often have solutions that develop discontinuities in finite time and therefore for which point evaluations do not make sense everywhere} \citep{leveque2002finite}. 

Assume that PDE~\eqref{eq:pde_u} is {approximated on} a domain that has been discretized as follows: in time we consider equidistant time steps $t_i=i\Delta t$ for step size $\Delta t>0$ and $i\in\mathbb{N}$, in space we consider {equispaced} cells of size $\Delta x>0$ with centroids $x_j=j\Delta x$ for $j\in\mathbb{Z}$. We then introduce the cell average functions $U_j$ defined by
\begin{equation}
U_j(t) = \frac{1}{\Delta x} \int_{x_j-\Delta x/2}^{x_j+\Delta x/2} u(t, x)\di x, \quad t\ge 0, \quad j\in\mathbb{Z} \peq
\label{eq:cell_avg}
\end{equation}
Hence, $U_j$ is the cell average of the solution $u$ of the PDE on the $j$-th cell.
Then, the conservation law~\eqref{eq:cons_law} applied on the $j$-th cell yields the following differential equation satisfied by $U_j$
\begin{equation}
\frac{\dd U_j}{\dd t}(t)=\frac{1}{\Delta x}\left[f(u(t,x_{j-1/2})\pv v_m)-f(u(t,x_{j+1/2})\pv v_m) \right], \quad t\ge 0, \quad j\in\mathbb{Z} \veq
\label{eq:cons_fv}
\end{equation}
where $x_{j\pm 1/2}=x_j\pm\Delta x/2$ are the boundary points of the $j$-th cell, and $f(u(t,x_{j\pm 1/2})\pv v_m)=\phi(t,x_{j\pm 1/2})/\rho_m$ denotes the (normalized) flux of vehicles at each boundary of the cell. 

Finite volume methods propose to turn this set of equations into a system of ordinary differential equations by replacing the right-hand side of~\eqref{eq:cons_fv} by a function of the cell averages $\lbrace U_j : j\in \mathbb{Z} \rbrace$. In the particular case of (3-point) conservative schemes, the flux at a boundary point  $x_{j-1/2}$ is replaced by a so-called numerical flux $F(U_{j-1}, U_j\pv v_m)$ depending on the cell averages of the two cells $j-1$ and $j$ sharing that boundary and on the parameter $v_m$ defining the flux function. 

Hence, \Cref{eq:cons_fv} yields a system of ODEs for the finite volume approximations $\widehat{U}_j$ of the cell averages~$U_j$: 
\begin{equation}
\frac{\dd \widehat{U}_j}{\dd t}(t)=\frac{1}{\Delta x}\left[F(\widehat{U}_{j-1}, \widehat{U}_j\pv v_m)-F(\widehat{U}_{j}, \widehat{U}_{j+1}\pv v_m) \right], \quad t\ge 0, \quad j\in\mathbb{Z} \peq
\label{eq:ode_fv}
\end{equation}
We assume, for all $j$, that the initial condition $\lbrace U_j(0) : j \in \mathbb{Z}\rbrace$ is known, and use it to set the initial values of the finite volume approximation $\widehat{U}_j(t_0)=U_j(t_0)$. Then 
the numerical solution of the PDE using such schemes can be obtained at times $t_1, t_2, \dots$ by considering an Euler time discretization (of step size $\Delta t$) of the system~\eqref{eq:ode_fv}, which yields the recurrence relation
\begin{equation}
\widehat{U}_j^{i+1}=\widehat{U}_j^i+\frac{\Delta t}{\Delta x}\left[F(\widehat{U}_{j-1}^i, \widehat{U}_j^i\pv v_m)-F(\widehat{U}_{j}^i, \widehat{U}_{j+1}^i\pv v_m) \right], \quad j\in\mathbb{Z}, \quad i\in\mathbb{N} \veq
\label{eq:rec_fv}
\end{equation}
where $\widehat{U}_j^i$ is the quantity defined by $\widehat{U}_j^i\approx \widehat{U}_j(t_i)$. Different choices of numerical flux $F$ yield different schemes. Among the choices most encountered in the literature, we can cite the Lax--Friedrichs  (LxF) scheme and the Godunov scheme which are presented in~\Cref{sec:fvs} \citep{leveque2002finite,eymard2000finite,barth2018finite}.

\subsection{Traffic reaction model}
\label{sec:trm}

In this paper, we focus on a particular finite volume scheme \citep{liptak2020traffic} for solving PDE~\eqref{eq:pde_u} under the assumption that $f$ is given by~\eqref{eq:f}: the  so-called \q{Traffic Reaction Model} (TRM). The TRM is obtained by modeling the road traffic dynamic (on the discretized road) as a chemical reaction network. More precisely, the process of a vehicle passing from a cell $j$ to the next one is interpreted as a chemical reaction that \q{transforms} a unit of occupied space $O_j$ in $j$ and a unit of free space $\Phi_{j+1}$ in $j+1$ into  a unit of free space $\Phi_j$ in $j$ and a unit of occupied space $O_{j+1}$ in $j+1$ (cf. \Cref{fig:rr}). Hence, the road cells are interpreted as compartments containing two homogeneously distributed chemical reactants (Free space $\Phi_j$ and Occupied space $O_j$) and interacting with each other (through the \q{transfer} reaction). 

\begin{figure}
     \centering
\resizebox{0.8\textwidth}{!}{
	\input{fig/reaction_rates.tex}
}
\caption{Representation of the Traffic Reaction Model interpretation of traffic flow.}
\label{fig:rr}
\end{figure}
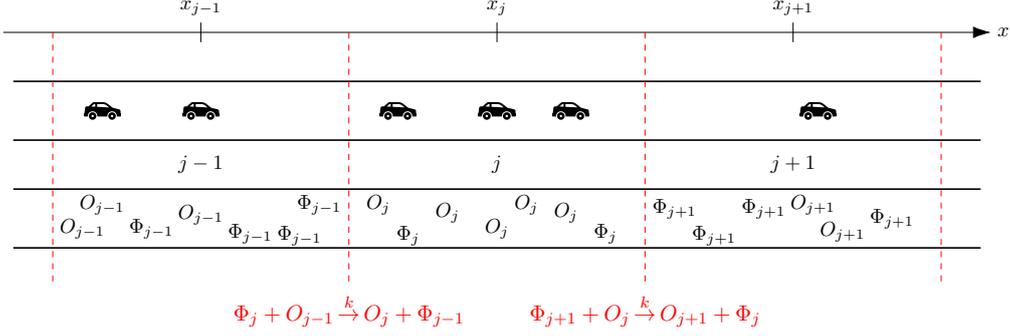

The law of mass action then allows to study the kinetics of this network of reactions \citep{Feinberg2019}. The rate at which a particular reaction happens is then modeled as being proportional to the product of the concentration of each one of its reactants (elevated to the power of the stoichiometric coefficient of the reactant, i.e., the number of \q{units} of this reactant consumed by a single reaction).  Denote by $o_j(t)$ (resp $\varphi_j(t)$) the concentration of Occupied space $O_j$ (resp. Free space $\Phi_j$) in the $j$-th compartment at time $t\ge 0$. The evolution of these concentrations can be expressed as the difference between the rate at which they are produced by the rate at which they are consumed by the reactions happening in the compartment, thus giving 
\begin{equation*}
\left\lbrace\begin{aligned}
\frac{\dd o_j}{\dd t} &=k_{j-1\rightarrow j}o_{j-1}\varphi_j-k_{j\rightarrow j+1}o_{j}\varphi_{j+1}  \\
\frac{\dd \varphi_j}{\dd t}&=-k_{j-1\rightarrow j}o_{j-1}\varphi_j+k_{j\rightarrow j+1}o_{j}\varphi_{j+1} 
\end{aligned}\right. , \quad j\in\mathbb{Z} \veq
\end{equation*}
where, for any $j$, $k_{j-1\rightarrow j}$ is the reaction rate (proportionality) constant of the reaction between the $(j-1)$-th compartment and the $j$-th compartment. Adding these two differential equations gives in particular that the quantity $\varphi_j+o_j$ is conserved through time.

From now on, we assume that all the reaction rates $k_{j-1\rightarrow j}$ are constant and equal to $k$. Identifying then the concentration of occupied space $o_j$ with the density of vehicles $\rho_j$ in the $j$-th road cell/compartment, and the identifying the concentration of free space $\varphi_j$ with a density of free space, we note that the conserved quantity $o_j+\varphi_j$ can be interpreted as the maximal capacity of the cell, which gives $o_j+\varphi_j=\rho_m$. Hence, we have
\begin{equation*}
\frac{\dd \rho_j}{\dd t} =k\rho_{j-1}(\rho_m-\rho_j)-k\rho_{j}(\rho_m-\rho_{j+1}), \quad j\in\mathbb{Z} \peq
\end{equation*}
Dividing this last expression by $\rho_m$ to get back to normalized densities  and using an Euler time discretization then gives
\begin{equation}
\widehat{U}_j^{i+1}=\widehat{U}_{j}^i+C\widehat{U}_{j-1}^i(1- \widehat{U}_j^i)-C\widehat{U}_{j}^i(1- \widehat{U}_{j+1}^i), \quad j\in\mathbb{Z}, \quad i\in\mathbb{N} \veq
\label{eq:trm_expl}
\end{equation}
where $C$ is the quantity defined by
\begin{equation}
C=\Delta t \rho_m k \peq
\end{equation}
Note in particular that by definition of the reaction rate $k$, the quantity $C\rho_m\Delta x= (k\rho_m\Delta t)(\rho_m\Delta x)$ corresponds to the maximal number of vehicles that can be transferred from the $(j-1)$-th to the $j$-th cell during a period~$\Delta t$. Indeed in the ideal case where the $j$-th cell is empty (i.e., $\varphi_j=\rho_m$), the transfer reaction between $(j-1)$ and $j$ happens at rate $k\varphi_jo_{j-1}=k\rho_mo_{j-1}$. Hence during $\Delta t$, the concentration of reactants decreases by $k\rho_mo_{j-1}\Delta t$ times. Hence, if the $(j-1)$-th cell is full ($o_{j-1}=\rho_m$), $(k\rho_m\Delta t)(\rho_m\Delta x)$ vehicles will be transferred. 

This quantity can also be expressed in terms of the maximal speed $v_m$ of the vehicles. Indeed, during $\Delta t$, vehicles can travel a distance of at most $v_m\Delta t$ meaning that at most $\rho_m(v_m\Delta t)$ can cross a cell interface during this period. By equating both expressions we get
\begin{equation*}
k=\frac{1}{\Delta x}\frac{v_m}{\rho_m}, \quad C=\frac{\Delta t}{\Delta x}v_m \veq
\end{equation*}
which in turn gives
\begin{equation}
\widehat{U}_j^{i+1}=\widehat{U}_{j}^i+\frac{\Delta t}{\Delta x}v_m\left[\tilde F_T(\widehat{U}_{j-1}^i, \widehat{U}_j^i)-\tilde F_T(\widehat{U}_{j}^i, \widehat{U}_{j+1}^i) \right], \quad j\in\mathbb{Z}, \quad i\in\mathbb{N} \veq
\label{eq:rec_fv_trm}
\end{equation}
with the normalized numerical flux $\tilde{F}_T$ given by
\begin{equation*}
\tilde F_T(u,v )=u(1-v) \peq
\end{equation*}

As a finite volume scheme, the TRM has several desirable properties. First, it is a consistent with the quadratic flux~\eqref{eq:f_u}, meaning that numerical flux $F_T$ satisfies $F_T(u,u \pv v_m)=f(u\pv v_m)$ for any $u\in [0, 1]$. Then, under the assumption that the discretization steps $\Delta t, \Delta x$ satisfy a so-called Courant–Friedrichs–Lewy (CFL) condition given by
\begin{equation}
\frac{\Delta t}{ \Delta x} \le \frac{1}{2 v_m} \veq
\label{eq:cfl}
\end{equation}
In particular (cf. \citep{liptak2020traffic}) the TRM is
\begin{itemize}
\item monotone: if the recurrence is initialized with two initial conditions $\lbrace \widehat{U}_j^0 : j\in\mathbb{Z}\rbrace$ and $\lbrace \widehat{V}_j^0 : j\in\mathbb{Z}\rbrace$ such that for any $j$, $\widehat{U}_j^0\le \widehat{V}_j^0$ then for any $i$ and any $j$, $\widehat{U}_j^i\le \widehat{V}_j^i$.
\item $L^{\infty}$-stable: if there exist $A, B \ge 0$ such that the initial condition satisfies for any $j$, $U_j^0 \in [A, B]$, then for any $i$ and $j$, $\widehat{U}_j^i\in [A, B]$.
\item convergent: 
the $L^1$ norm between the discrete cell-defined solutions and the true solution\footnote{The term true solution refers here to the notion of entropy solution of the PDE, which is the unique physically-relevant (weak) solution of the PDE \citep{leveque2002finite}.} converges to $0$ as $\Delta x \rightarrow 0$ (with $\Delta t/\Delta x$ kept constant). This convergence result is a consequence of the consistency and monotonicity of the scheme \citep{leveque2002finite}.
\end{itemize}

\begin{rem}
The CFL condition~\eqref{eq:cfl} can also be interpreted in the context of reaction kinetics. Indeed, starting from the quantity $C$ defined above, we have that the CFL condition is equivalent to imposing that 
\begin{equation}
C=\Delta t \rho_m k = \frac{\Delta t}{\Delta x}v_m <\frac{1}{2} \peq
\label{eq:cfl_reac}
\end{equation}
Let us then consider for instance the interface between the compartment $(j-1)$ and $j$. Recall that, according to the law of mass action, the rate at which the transfer reaction of this interface happens is $ko_{j-1}\varphi_j$, meaning that during $\Delta t$,  $N(\Delta t)=ko_{j-1}\varphi_j\Delta t \Delta x$ vehicles are transferred. Note that, since $o_{j-1}$ is upper-bounded by the maximal capacity $\rho_m$, we can upper-bound $N(\Delta t)$ by $N(\Delta t)\le k\rho_m\varphi_j\Delta t \Delta x$. Then, the CFL condition yields through~\eqref{eq:cfl_reac} an upper-bound for $k$, which in turn gives 
$$N(\Delta t)< \varphi_j\frac{\Delta x}{2} \peq$$
This means in particular that, during $\Delta t$, the number of vehicles transferred from cell $(j-1)$ to cell $j$ is lower than the number of free slots in the left-half of cell $j$. Similarly, we can prove that this same number is lower than the number of vehicles slots in the right-half of cell $j-1$ (by starting by upper-bounding $\varphi_j$). 

Hence, the CFL condition allows to decouple what is happening at the different interfaces during a time-lapse~$\Delta t$: indeed, at each interface between two cells $(j-1)$ and $j$, the reaction dynamics come down to a transfer of vehicles from the right-half of cell~$(j-1)$ to the left-half of cell~$j$, and in this sense are independent of what is happening at other interfaces (or within other half cells).

\label{rem:CFL}
\end{rem}

\subsection{Extension of the TRM to spatial and temporal variation in parameter}
\label{sec:ext}

Reverting to the traffic interpretation of traffic flow dynamics used by the TRM and presented in \Cref{sec:trm}, working with a constant parameter $v_m$ in the traffic model implies that the traffic flows without perturbation along a homogeneous road: indeed all the reactions between consecutive road compartments happen with the same rate.  A direct generalization of this model consists in considering that these reaction rates can now vary in time or across space (i.e., two pairs of consecutive compartments can have different reaction rates).

Let us first assume that the road is infinite and discretized into cells of same size. The evolution of the normalized density in a given compartment would then take the form
\begin{equation*}
\frac{\dd \widehat{U}_j}{\dd t}(t) =\rho_m k_{j-1\rightarrow j}(t)\widehat{U}_{j-1}(t)\left(1-\widehat{U}_j(t)\right)-\rho_m k_{j\rightarrow j+1}(t)\widehat{U}_{j}(t)\left(1-\widehat{U}_{j+1}(t)\right), \quad t\ge 0, \quad j \in\mathbb{Z} \veq
\end{equation*}
where, for any $j$, $k_{j-1\rightarrow j}(t)$ is the reaction rate of the reaction between the $(j-1)$-th compartment and the $j$-th compartment, at time $t$. An explicit Euler discretization of this expression then gives the recurrence
\begin{equation}
\widehat{U}_j^{n+1} =\widehat{U}_j^n+ C_j^n\widehat{U}_{j-1}^n\left(1-\widehat{U}_j^n\right)-C_{j+1}^n\widehat{U}_{j}^n\left(1-\widehat{U}_{j+1}^n\right), \quad j\in\mathbb{Z}, \quad n \in\mathbb{N} \veq
\label{eq:trm_var}
\end{equation}
where for any $j, n$, the quantity $C_j^n$ is defined, at the $n$-th time step $t_n$, by
\begin{equation*}
C_j^n=\rho_m k_{j-1\rightarrow j}(t_n)\Delta t \peq
\end{equation*}
Note that, following the same approach as in \Cref{sec:trm}, the quantity $C_j^n\rho_m\Delta x$ amounts to the  maximal amount of vehicle transfers that can happen between the $(j-1)$-th cell and the $j$-th cell  during a period $\Delta t$ starting at time $t_n$. Assuming that, around the interface $(j-1)/j$, the maximal speed of the vehicles is now a time dependent function $v_m(\cdot, x_{j-1/2}) : t\mapsto v_m(t, x_{j-1/2})$, this quantity can be equated to
\begin{equation*}
C_j^n\rho_m\Delta x=\rho_m\int_{t_n}^{t_n+\Delta t} v_m(t, x_{j-1/2})\di t \veq
\end{equation*}
by seeing the integral on the right-hand side as the limit of a Riemann sum, and using the similar result derived in \Cref{sec:trm} for the constant case. Hence, we have
\begin{equation*}
k_{j-1\rightarrow j}(t_n)=\frac{1}{\rho_m\Delta x} \left(\frac{1}{\Delta t}\int_{t_n}^{t_n+\Delta t} v_m(t, x_{j-1/2})\di t\right), \quad C_j^n=\frac{\Delta t}{\Delta x} \left(\frac{1}{\Delta t}\int_{t_n}^{t_n+\Delta t} v_m(t, x_{j-1/2})\di t\right) \peq
\end{equation*}
In particular, following the reasoning of \Cref{rem:CFL}, we will assume that the quantities $C_j^n$ satisfy the same condition as in the constant case, namely $C_j^n\in (0, 1/2)$.

The scheme defined by~\eqref{eq:trm_var} can be seen as finite volume scheme with a numerical flux consistent with the space-time dependent flux function $f$ defined by
\begin{equation}
f(t,x, \rho(t,x))=v_m(t,x)\rho\left(1-\frac{\rho}{\rho_m}\right) \peq
\end{equation}
This type of finite volume scheme was studied in the context of approximation of non-homogeneous scalar conservation laws \citep{chainais2001finite}. Under some regularity assumption on the speed parameter $v_m$ (that are not stricter than those described earlier for the existence and uniqueness of an entropy solution), this finite volume scheme converges to the entropy solution of PDE~\eqref{eq:pde_u}, hence corresponding to a LWR model with space-time varying parameter $v_m$ \citep[Theorem 1]{chainais2001finite}.


\section{A discrete dynamical system to bridge the gap between models and measurements}
\label{sec:estim}

\subsection{Constant parameter case}

Let us assume that measurements of the density and flux of vehicles along a road are available.  In particular, we assume that these measurements are made along a road discretized into $N_x$ cells of size $\Delta x$ (centered at locations $x_j=j\Delta x$, $j\in\bi 0, N_x-1\ei$) and at $N_t$ time steps spaced by $\Delta t$ (and denoted by $t_i=i\Delta t$, $i\in\bi 0, N_t-1\ei$). These measurements are collected into a density matrix $\bm D=\lbrace D_j^i : j\in\bi 0, N_x-1\ei, i\in\bi 0, N_t-1\ei\rbrace \in\R^{N_t\times N_x}$ and a flux matrix $\bm F=\lbrace F_j^i : j\in\bi 0, N_x-1\ei, i\in\bi 0, N_t-1\ei\rbrace\in\R^{N_t\times N_x}$, whose entries $D_j^i$ and $F_j^i$ are  the measurements made at time $t_i$ and location $x_j$. We will also assume that the maximal density $\rho_m$ of the road is known, and that therefore a normalized density matrix $\bm U=\lbrace U_j^i : j\in\bi 0, N_x-1\ei, i\in\bi 0, N_t-1\ei\rbrace\in\R^{N_t\times N_x}$ can be obtained from $\bm D$ by dividing its entries by $\rho_m$.  We aim at assessing whether the continuous LWR model introduced in the previous section adequately represents the traffic flow as observed through $\bm D$ (or $\bm U$) and $\bm F$.

In order to bridge the gap between continuous models of density and the discrete measurements at hand, and therefore to be able to compare them, we think of the density observations in $\bm U$ as arising from a particular solution of the (continuous) PDE~\eqref{eq:pde_u} for some unknown (but constant) value  of the parameter $\bar v_m$. We then propose to leverage the fact that finite volume schemes would naturally provide estimates for the entries of $\bm U$ (assuming $\bar v_m$ is known). The idea is then to look for the PDE parameter value $v_m^*$ which gives finite volume estimates $\widehat{\bm U}^*$ closest to the observed data $\bm U$. Then, the quality of the continuous model is assessed by comparing $\widehat{\bm U}^*$ to $\bm U$, and comparing the flux matrix $\bm F$ to flux estimates obtained by applying the flux-density relationship~\eqref{eq:f} of the model to the density estimates $\widehat{\bm U}^*$.

The optimal parameter value $v_m^*$ is obtained as follows: For any choice of parameter $v_m$, we can compute the finite volume discretization of PDE~\eqref{eq:pde_u} by applying the recurrence relation~\eqref{eq:rec_fv} $N_t-1$ times, thus giving a matrix of estimates $\widehat{\bm U}=\lbrace \widehat{U}_j^i : j\in\bi 0, N_x-1\ei, i\in\bi 0, N_t-1\ei\rbrace$. The initial state of this recurrence is set up using $\bm U$ as
\begin{equation}
\widehat{U}_j^0=U_j^0, \quad j\in\bi 0, N_x-1\ei \peq
\label{eq:ic}
\end{equation}
Observed data is only available on a road of finite length. As has been seen earlier, the scheme only takes values from the neighboring compartments into account. Therefore, we choose boundary conditions using once again $\bm U$  by imposing
\begin{equation}
\widehat{U}_0^i=U_0^i, \quad \widehat{U}_{N_x-1}^i=U_{N_x-1}^i, \quad i\in\bi 0, N_t-1\ei \peq
\label{eq:bc}
\end{equation}

This overall process is seen as computing the output of a discrete dynamical system at times $t_1, \dots, t_{N_t-1}$. Indeed, note that for the finite volume schemes~\eqref{eq:rec_fv_trm}, \eqref{eq:rec_fv_god} and~\eqref{eq:rec_fv_lxf} considered in this paper, we can introduce the (unit-free) scaling parameter $C$ by
\begin{equation}
C=\frac{\Delta t}{\Delta x}v_m
\label{eq:cvm}
\end{equation}
and then write the finite volume recurrence relation~\eqref{eq:rec_fv} as
\begin{equation}
\widehat{\bm U}^{i+1}(C)=\mathcal{H}^i(\widehat{\bm U}^i(C),  C \pv \bm U), \quad i\in\bi 0,N_t-2\ei \veq
\label{eq:rec}
\end{equation}
where  $\widehat{\bm U}^i\in\R^{N_x}$ is the vector defined by $\widehat{\bm U}^i=[\widehat U_0^i, \dots, \widehat U_{N_x-1}^i]^T$, and $\mathcal{H}^i = (\mathcal{H}_0^i, \dots,\mathcal{H}_{N_x-1}^i) : \R^{N_x} \times \R^{N_C}\rightarrow \R^{N_x}$ is the transformation defined in part by the boundary conditions~\eqref{eq:bc} as
\begin{equation}
\mathcal{H}_j^i\left(\widehat{\bm U}^i, C \pv \bm U\right)
=
\begin{cases}
 U_0^{i+1} & \text{if } j=0 \\
h(\widehat{U}_{j-1}^i,\widehat{U}_{j}^i,\widehat{U}_{j+1}^i)+C\left[\tilde F(\widehat{U}_{j-1}^i, \widehat{U}_j^i)-\tilde F(\widehat{U}_{j}^i, \widehat{U}_{j+1}^i ) \right] & \text{if } j\in\bi 1, N_x-2\ei \\
 U_{N_x-1}^{i+1} & \text{if } j=N_x-1
\end{cases}\veq
\label{eq:transfo}
\end{equation}
with $h$ and $\tilde{F}$ depending on the choice of numerical scheme (see \Cref{eq:rec_fv_trm,eq:rec_fv_god,eq:rec_fv_lxf}). The discrete dynamical system is then defined as follows:
\begin{itemize}
\item the state vector of the system contains the finite volume approximations of the PDE across the discretized road, at a given time step;
\item the initial state of the system is the vector $\widehat{\bm U}^0$, defined by the initial condition~\eqref{eq:ic} of the scheme;
\item the recurrence relation~\eqref{eq:rec} defines the successive  state updates;
\item the scaling parameter $C$ acts like a control parameter of the system.
\end{itemize}

\begin{rem}
Note that the CFL condition~\eqref{eq:cfl} actually imposes a restriction on the domain of definition of the control parameter $C$: for the schemes considered in this paper to yield approximations of the cell averages of the solution (and ensure that the recurrence does not diverge), we should only consider $C\in (0,1/2)$. 
\end{rem}

The optimal PDE parameter $v_m^*$ is then obtained by finding the value of the control parameter $C$ of the discrete dynamical system that minimizes a cost function measuring the discrepancy between the output of the system  and the data $\bm U$. In particular, we consider a least-square approach, meaning that the optimal control parameter $C^*$ will be the solution of the problem
\begin{equation}
C^*=\argmin\limits_{C\in (0, 1/2)} \sum_{i=0}^{N_t-1}\sum_{j=0}^{N_x-1} \left( \widehat{U}_j^i(C) - U_j^i\right)^2 \peq
\label{eq:min_pbm2}
\end{equation}
Following~\eqref{eq:cvm}, this gives in turn an optimal PDE parameter $v_m^*$ given by
\begin{equation*}
v_m^*=\frac{\Delta x}{\Delta t} C^* \veq
\end{equation*}
and optimal finite volume estimates given by $\widehat{\bm U}^*=\widehat{\bm U}(C^*)$.

Finally, we turn the minimization problem~\eqref{eq:min_pbm2} into an unconstrained minimization problem by introducing the parameter $\theta$ defined by
\begin{equation}
C(\theta)=\frac{1}{2}\lgt(\theta), \quad \theta \in \R \veq
\label{eq:corresp}
\end{equation}
where $\lgt : \R \rightarrow (0,1)$ is the logit function\footnote{The logit function is defined by $\lgt(\theta)=(1+e^{-\theta})^{-1}$, $\theta \in\R$, and has an inverse defined by $\lgt^{-1}(y)=-\log(y^{-1}-1)$, $y\in (0,1)$.} which defines a strictly increasing bijection between $\R$ and $(0,1)$. The strict monotonicity and smoothness of $\lgt$ then allow to cast the minimization problem~\eqref{eq:min_pbm2} into the following equivalent minimization problem
\begin{equation}
\theta^*=\argmin\limits_{\theta\in \R} \sum_{i=0}^{N_t-1}\sum_{j=0}^{N_x-1} \left( \widehat{U}_j^i(C(\theta)) - U_j^i\right)^2 \veq
\label{eq:min_pbm3}
\end{equation}
where, for any $\theta\in\R$, $C(\theta)\in(0,1/2)$ is defined by~\eqref{eq:corresp}. Then, the optimal control parameter $C^*$ of \Cref{eq:min_pbm2} is simply obtained by taking $C^*=C(\theta^*)$. 

Following from the boundary conditions~\eqref{eq:bc} and initial conditions~\eqref{eq:ic}, the minimization problem~\eqref{eq:min_pbm3} then boils down to the unconstrained minimization of the cost function $L$ defined by 
\begin{equation}
L(\theta)=\frac{1}{2}\sum_{i=1}^{N_t-1}\sum_{j=1}^{N_x-2} \left( \widehat{U}_j^i(C(\theta)) - U_j^i\right)^2, \quad \theta \in\R \peq
\label{eq:costf}
\end{equation}
This minimization task can be in particular tackled using gradient-based optimization problems since the structure of the recurrence relation~\eqref{eq:rec} can be leveraged to derive analytic expressions for the gradient of $L$ (cf. \Cref{sec:grad_comp}). That is of course if we assume that we can take the derivative of the maps $\mathcal{H}^i$ in~\eqref{eq:rec}, which in our case forces us to work with either the LxF scheme or the TRM scheme. In particular, in the applications presented in this paper, only these two schemes are considered and the conjugate gradient algorithm is used to perform the minimization \citep{nocedal2006numerical}.

\subsection{Varying parameter case}

The approach presented in the previous section naturally extends to the assumption where parameters varying in space or time are considered (as described in \Cref{sec:ext}). The discrete dynamical system is defined using the recurrence relation \eqref{eq:trm_var}, and its output is once again compared to the density data $\bm U$ to derive optimal control parameter values through  a minimization approach. In particular,
\begin{itemize}
\item the boundary and initial conditions are set in the same way as in the constant case;
\item the recurrence relation defining the system takes the form
\begin{equation*}
\widehat{U}_j^{n+1} =\widehat{U}_j^n+  C_j^n\widehat{U}_{j-1}^n\left(1-\widehat{U}_j^n\right)-C_{j+1}^n\widehat{U}_{j}^n\left(1-\widehat{U}_{j+1}^n\right), \quad j\in\bi 1, N_x-2\ei, \quad n \in \bi 0, N_t-2\ei \; ;
\end{equation*}
\item the control parameters of the system are the coefficients $\bm C=\lbrace C_j^n : j\in \bi 0, N_x\ei, n \in\bi 0, N_t-1\ei\rbrace$, and are determined by minimizing (without constraints) a cost function $L$ given as the sum of a least-square cost and a regularization term $R(\bm C)$ (clarified below):
\begin{equation*}
L(\bm \theta)= 
\frac{1}{2}\sum_{i=1}^{N_t-1}\sum_{j=1}^{N_x-2} \left( \widehat{U}_j^i(\bm C(\bm \theta)) - U_j^i\right)^2+\lambda R(\bm C(\bm\theta)), \quad \bm \theta \in \R^{(N_x+1)N_t} \veq
\end{equation*}
where $\bm C(\bm \theta)\in (0,1/2)^{(N_x+1)N_t}$ is obtained by applying the function~\eqref{eq:corresp} to each entry of $\bm \theta \in \R^{(N_x+1)N_t}$ and $\lambda>0$ is a hyperparameter balancing the importance of the least-square minimization of the regularization;
\item the following regularization term $R(\bm C)$ is considered:
\begin{equation*}
R(\bm C)=\frac{1}{2} \left( \sum_{j=0}^{N_x} \sum_{i=0}^{N_t-2}(C_j^i-C_{j}^{i+1})^2+ \sum_{i=0}^{N_t-1}\sum_{j=0}^{N_x-1} (C_j^i-C_{j+1}^{i})^2 \right), \quad \bm C \in (0,1/2)^{(N_x+1)N_t} \peq
\end{equation*}

\end{itemize}

Note that the regularization term introduced above plays two roles. On the one hand it allows to reduce the risk of overfitting: indeed the number of parameters now amounts to $N_t(N_x+1)$ which is larger than the number of terms in the least-square term, and thus increases the risk of overfitting. On the other hand it allows to ensure some kind of smoothness in space and time of the parameters, as sharp changes between consecutive coefficients in space or time are penalized. This kind of smoothness assumption of the space-time varying parameter $v_m$ is usually required to prove the existence and uniqueness of solutions of PDE~\eqref{eq:pde_u}  and explains why we try to enforce it in the estimation approach \citep{karlsen2004convergence,chen2005quasilinear}.

Finally, the minimization of the cost function is performed once again using the Conjugate gradient algorithm, while using the explicit formula of the gradient given in \Cref{sec:grad_comp}.

\subsection{Extension to a multilevel approach}
\label{sec:f_disc}

So far, the finite volume approximations $\widehat{\bm U}$ were computed on the same discretization grid as the one used to create the density matrix $\bm U$. This has two consequences. First, if the discretization steps are large with respect to the size of the domain, we can fear that the finite volume scheme will not yield satisfactory approximations of the cell averages of the solution. Second, this choice implicitly imposes a restriction on either the range of admissible parameters $v_m$ or the discretization pattern $\Delta t, \Delta x$ since we also impose that for the discretization steps used to compute $\widehat{\bm U}$, the CFL condition~\eqref{eq:cfl} should be satisfied. For instance, assuming that we upper-bound the admissible values of the critical speed $v_m$ by a speed of $180 \text{ km}\cdot\text{h}^{-1}=50 \text{ m}\cdot\text{s}^{-1}$, the CFL condition~\eqref{eq:cfl} gives that the discretization steps $\Delta t, \Delta x$ should satisfy $\Delta t / \Delta x < 0.01 \text{ s}\cdot\text{m}^{-1}$, or equivalently $\Delta t < 0.01 \Delta x \text{ s}$. Such a coupling may lead to consider very overly small time step sizes or large space step sizes, which can be limiting if one is not able to change the discretization pattern of the data.

To circumvent these limitations, we propose to use a multilevel approach where a different (and finer) discretization pattern, denoted by $\widehat{\Delta t}$, $\widehat{\Delta x}$, is used for the finite volume computations. In particular, we will take for some $P_t, P_x \in \mathbb{N}^*$, 
\begin{equation}
\widehat{\Delta t}=\frac{\Delta t}{P_t}, \quad \widehat{\Delta x}=\frac{\Delta x}{P_x} \veq
\label{eq:discr_b}
\end{equation}
meaning that the discretization pattern of the finite volume scheme will be a subdivision of the discretization pattern $\Delta t, \Delta x$ of the data (cf. \Cref{fig:mll}). Therefore, the CFL condition will become
\begin{equation}
\frac{\widehat{\Delta t}}{\widehat{\Delta x}}=\frac{{\Delta t}}{{\Delta x}}\frac{P_x}{P_t}\le \frac{1}{2v_m} \peq
\label{eq:cfl_f}
\end{equation}
Thus introducing two additional parameters $P_t$, $P_x$ can be used to enforce the CFL condition without having to impose anything on the discretization pattern of the data.

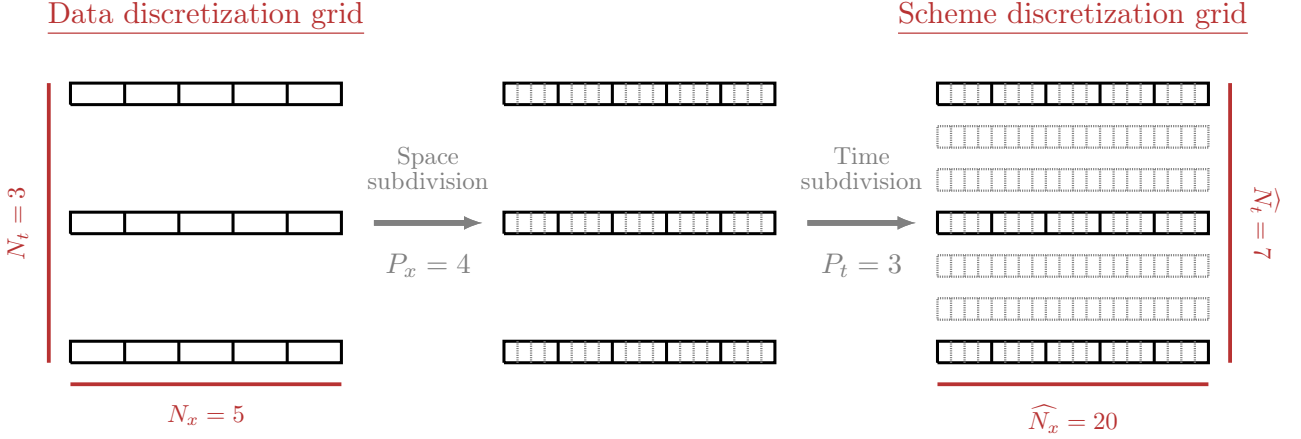
\begin{figure}
     \centering
\resizebox{\textwidth}{!}{
	\input{fig/multilevel.tex}
}
\caption{Space and time subdivisions of the multilevel approach. The discretization grid on which the data are defined (left, in black) is subdivided in space (each cell is subdivided into $P_x=4$ subcells) and in time ($P_t-1=2$ intermediate time steps are added between each consecutive time steps of the initial discretization). The resulting grid (right, in grey) is used to compute the finite volume density approximations.  }
\label{fig:mll}
\end{figure}
The choice~\eqref{eq:discr_b} yields that each road cell used to compute the density matrix~$\bm U$ is subdivided into $P_x$ subcells and each time step is subdivided into $P_t$ steps. Hence, the finite volume matrix $\widehat{\bm U}$ will have size $(P_t(N_t-1)+1)\times (P_xN_x)$. We then note that, for $ i\in\bi 0, N_x-1\ei, j\in\bi 0, N_t-1\ei$ and by definition, each entry $U_j^i$ of the density matrix $\bm U$ can be expressed as function of the solution $u$ of the PDE~\eqref{eq:pde_u} (with parameter $\bar v_m$) as
\begin{equation*}
U_j^i= \frac{1}{\Delta x} \int_{x_{j-1/2}}^{x_{j-1/2}+\Delta x} u(i\Delta t, x)\di x
=\frac{1}{P_x}\sum_{k=0}^{P_x-1}\frac{1}{\widehat{\Delta x}}\int_{x_{j-1/2}+k\widehat{\Delta x}}^{x_{j-1/2}+(k+1)\widehat{\Delta x}} u(iP_t \widehat{\Delta t}, x)\di x \veq
\end{equation*}
where the integral on the right-hand side corresponds to the cell average of the solution over the $k$-th subcell of the $j$-th road cell, at time $iP_t\widehat{\Delta t}$. Since this last quantity is approximated by the entry $(iP_t, k+jP_x)$ of the finite volume matrix $\widehat{\bm U}$, we deduce that an approximation of $U_j^i$ is obtained by taking the average of the quantities $\lbrace \widehat{U}_{k+jP_x}^{iP_t} : k\in\bi 0, P_x-1\ei\rbrace$. 

In order to use the recurrence relation~\eqref{eq:rec_fv}, initial and boundary conditions defined on the finer discretization grid of the finite volume scheme are needed. We deduce them from the data $\bm U$ by imposing an initial condition constant across all subcells of a given road cell, i.e.
\begin{equation}
\widehat{U}_{k+jP_x}^0=U_j^0, \quad j\in\bi 0, N_x-1\ei, \quad k\in\bi 0, P_x-1\ei \veq
\label{eq:ic_b}
\end{equation}
and for the boundary conditions, by considering a linear interpolation of the boundary conditions obtained from the data, i.e.
\begin{equation}
\left\lbrace\begin{aligned}
& \widehat{U}_{k+0P_x}^{l+iP_t}=U_0^{i} + \frac{l}{P_t}(U_0^{i+1}-U_0^{i}) \\
& \widehat{U}_{k+(N_x-1)P_x}^{l+iP_t}=U_{N_x-1}^i+\frac{l}{P_t}(U_{N_x-1}^{i+1}-U_{N_x-1}^i)
\end{aligned}\right.,
\quad i\in\bi 0, N_t-2\ei, \quad k\in\bi 0, P_x-1\ei, \quad l\in\bi 0, P_t\ei \peq
\label{eq:bc_b}
\end{equation}
The minimization problems introduced in the previous sections can then be readily reformulated to account for the difference in discretization steps between the data matrix and the finite volume estimates, as presented in detail in \Cref{sec:adpt_pbm}.

\section{Numerical experiments}
\label{sec:num}

\subsection{Application to parameter identification}
\label{sec:ident}

In this case study, we numerically solve PDE~\eqref{eq:pde_u} using the Godunov finite volume scheme over a space domain $[-1.5, 1.5]$ and a time frame $[0,1]$, with a parameter value $\bar v_m =1$ and a maximal density $\rho_m=1$. In particular, the space discretization step is chosen small compared to the domain extension, namely $\overline{\Delta x}=10^{-4}$, in order to guarantee that the numerical solution is close to the true solution. As for the time discretization step, it is set to $\overline{\Delta t}=0.25 \overline{\Delta x}/\bar{v}_m$ in accordance with the CFL condition~\eqref{eq:cfl}. Hence $\bar N_x=30000$ space cells and $\bar{N}_t=40001$ time steps are considered. The initial condition is taken as
\begin{equation*}
u_0(x)=0.5e^{-10x^2}+0.2\left(1+\cos(10\pi x)e^{-(3x^2+x)}\right), \quad x\in [-1.5.1.5] \veq
\end{equation*}
and therefore has a profile that is non-symmetric and has oscillations across space (cf. \Cref{fig:ic_sol}). Note that reflexive boundary conditions are considered, meaning that we set $U_{-1}^i=U_0^i$ and $U_{N_x}^i=U_{N_x-1}^i$ for any $i\ge 0$ in the recurrence~\eqref{eq:rec_fv}. The resulting numerical solution is cropped in space into the section $[-1,1]$ in order to avoid any possible boundary effect coming from the boundary conditions: this solution, represented in \Cref{fig:ic_sol}, is from now on considered as the ground truth solution.

\begin{figure}
\hfill
\begin{subfigure}{0.4\textwidth}
\includegraphics[width=\textwidth]{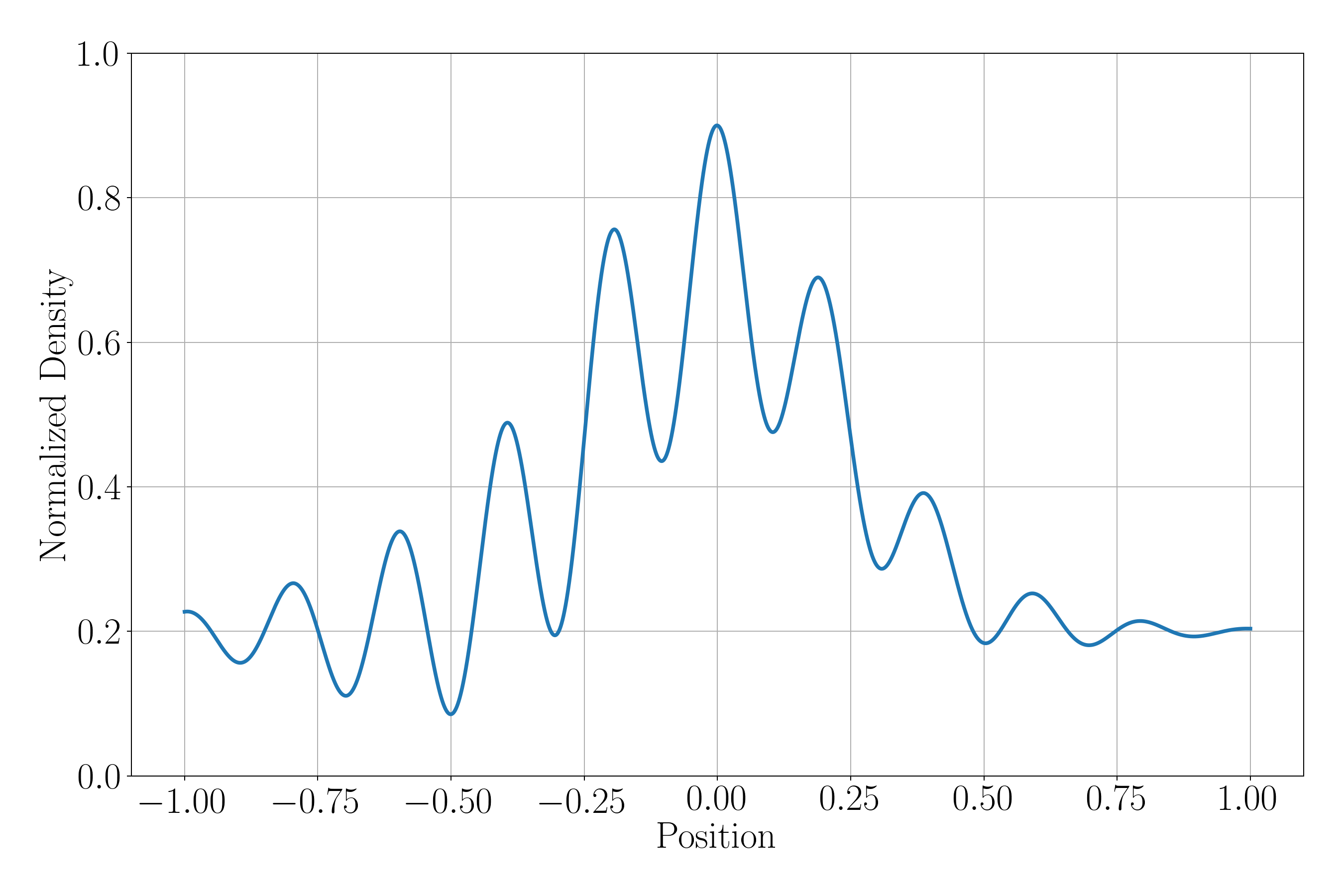}
\caption{Initial condition.}
\end{subfigure}
\hfill
\begin{subfigure}{0.4\textwidth}
\includegraphics[width=\textwidth]{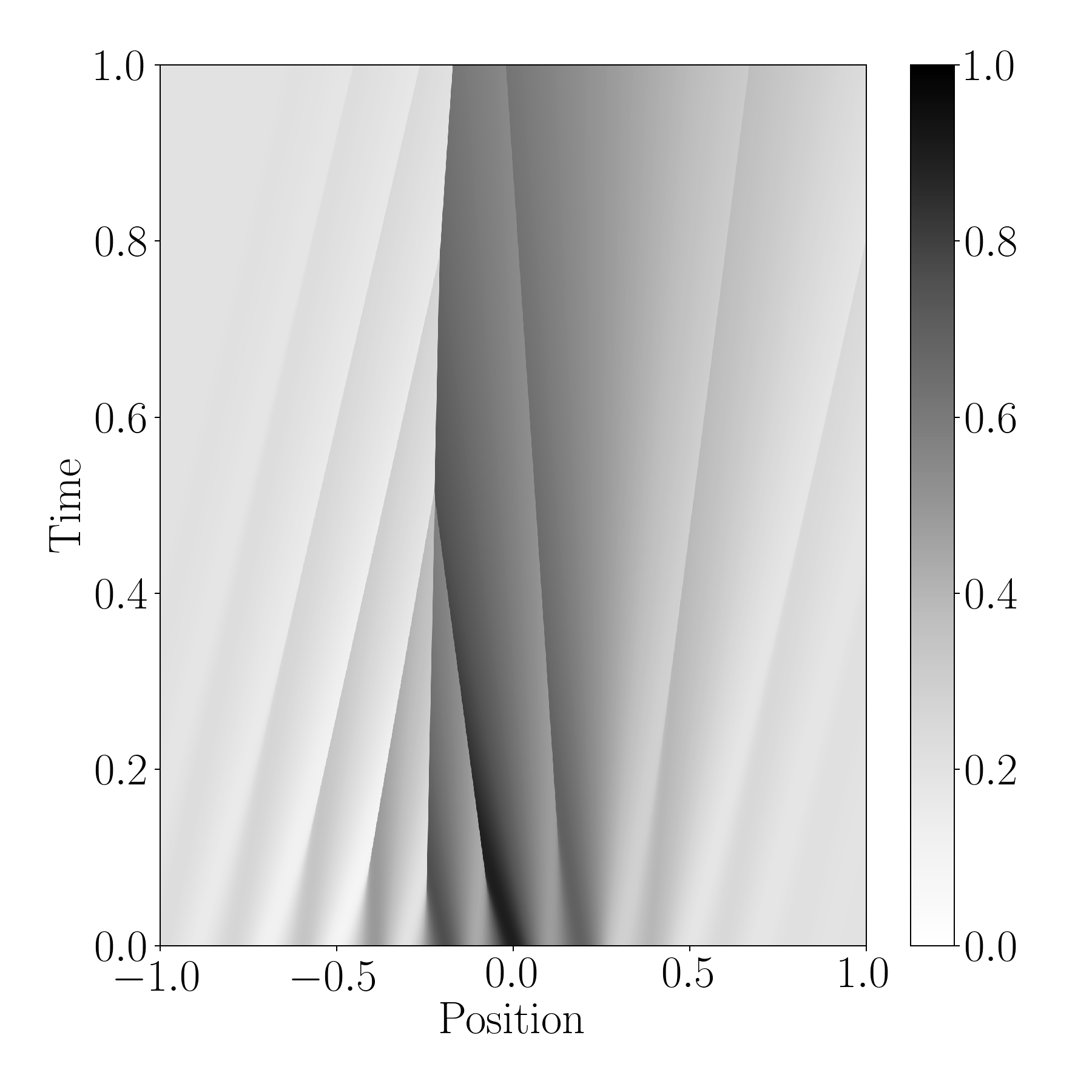}
\caption{Solution.}
\label{fig:ic_sol_mat}
\end{subfigure}
\hfill
\caption{Initial condition and associated solution of the PDE used in the synthetic case.}
\label{fig:ic_sol}
\end{figure}

Starting from the PDE solution described above, we build density data matrices $\bm U$ corresponding to different choices of discretization steps $\Delta t, \Delta x$. Examples of such density matrices are represented in \Cref{fig:ex_discr_dat}. We then estimate for each density matrix the value of the PDE parameter $v_m^*$ by solving the minimization problem~\eqref{eq:min_pbm3}, using both the TRM and the LxF scheme.  We also compute the Root Mean Square Error (RMSE) between the considered density matrix $\bm U$ and the corresponding finite volume approximations $\widehat{\bm U}^*=\widehat{\bm U}(v_m^*)$. \
The discretization steps used by these schemes are set according to \Cref{sec:f_disc}:  the finite volume approximations are computed on a grid obtained by applying $P_x\in\lbrace 1, 3, 5\rbrace$ space subdivisions and $P_t$ time subdivisions of the discretization grid of $\bm U$, where for each value $P_x$, $P_t$ is the smallest integer so that the CFL condition~\eqref{eq:cfl_f} is satisfied. 

\begin{figure}
\hfill
\begin{subfigure}{0.24\textwidth}
\includegraphics[width=\textwidth]{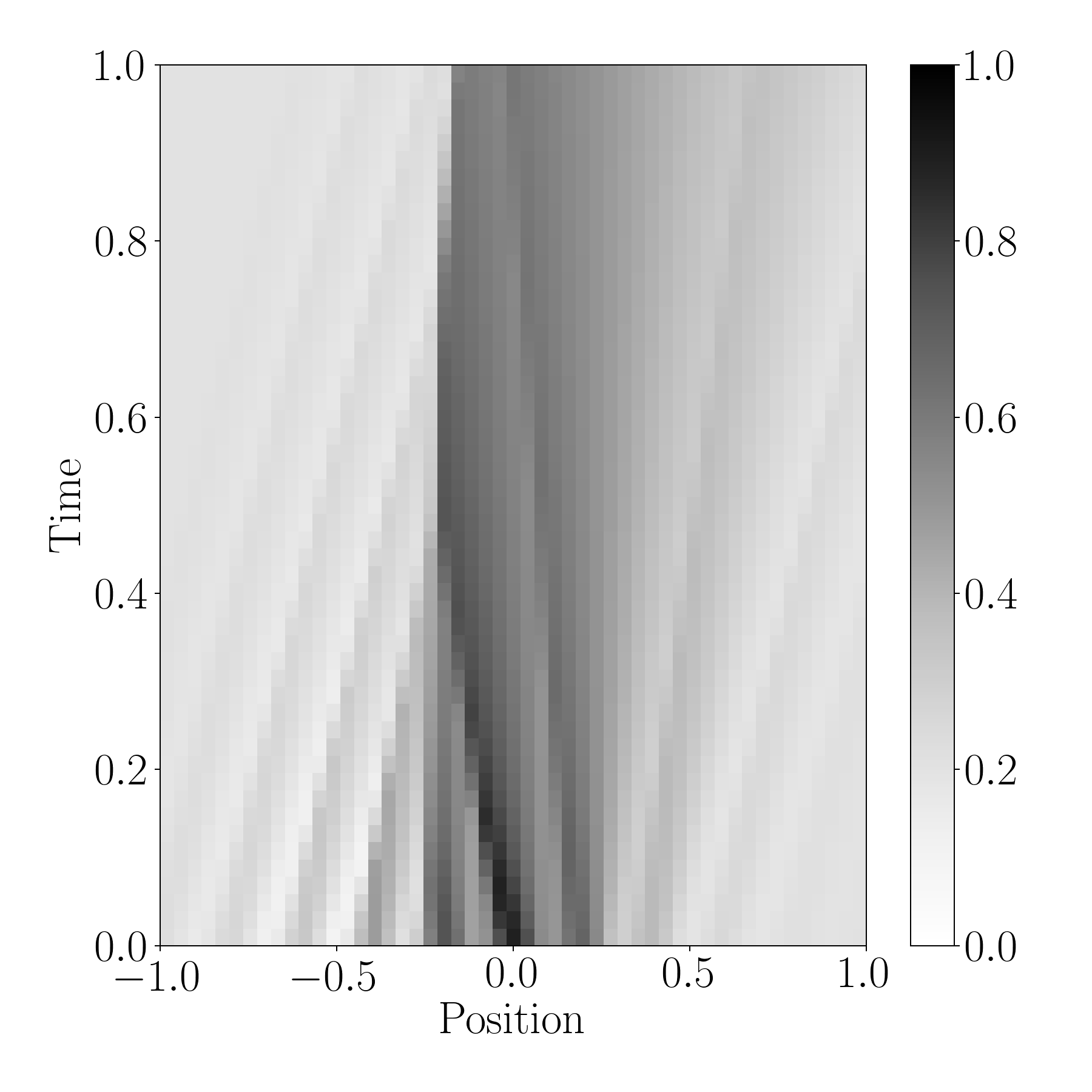}
\caption{$N_t=51, N_x=51$.}
\end{subfigure}
\hfill
\begin{subfigure}{0.24\textwidth}
\includegraphics[width=\textwidth]{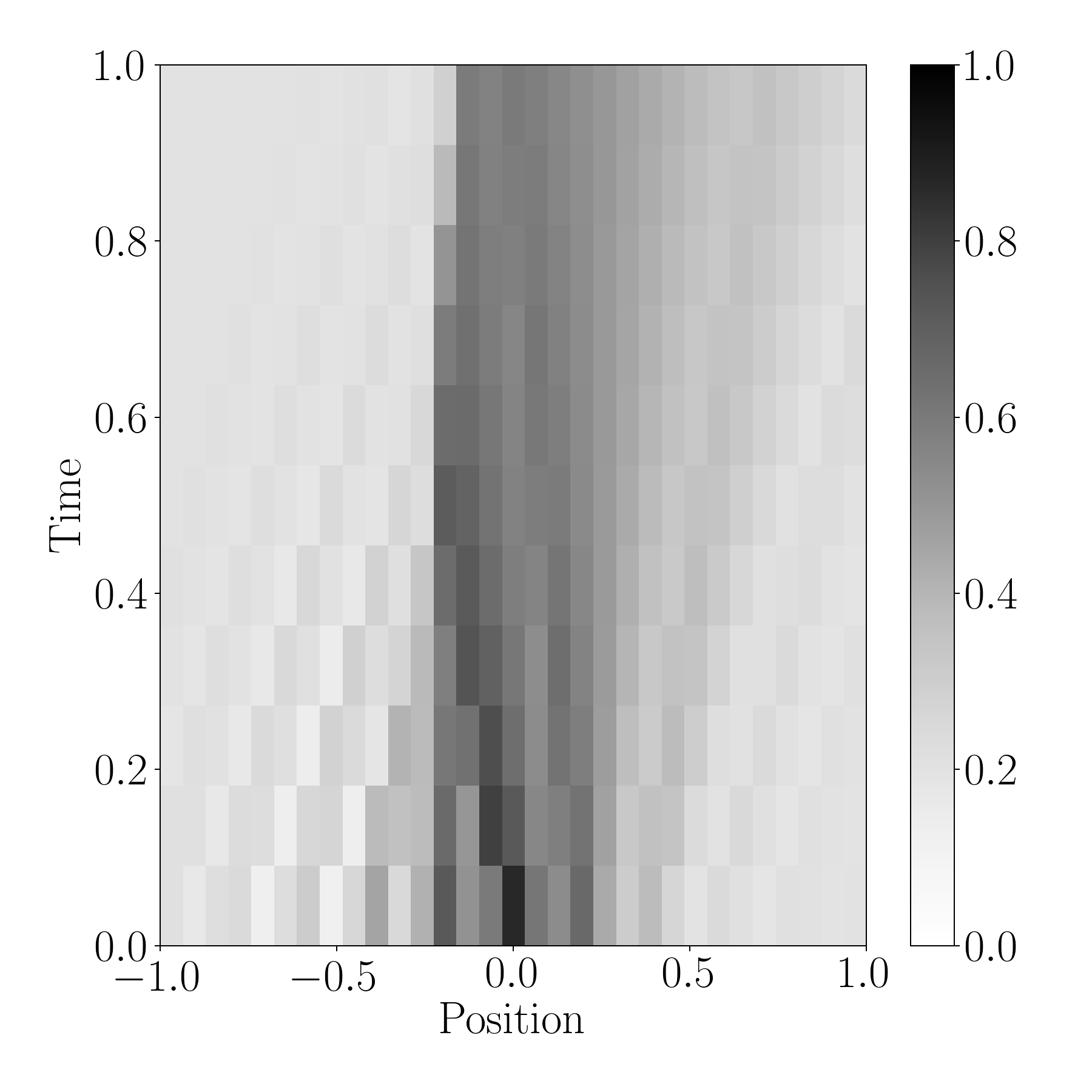}
\caption{$N_t=11, N_x=31$.}
\end{subfigure}
\hfill
\begin{subfigure}{0.24\textwidth}
\includegraphics[width=\textwidth]{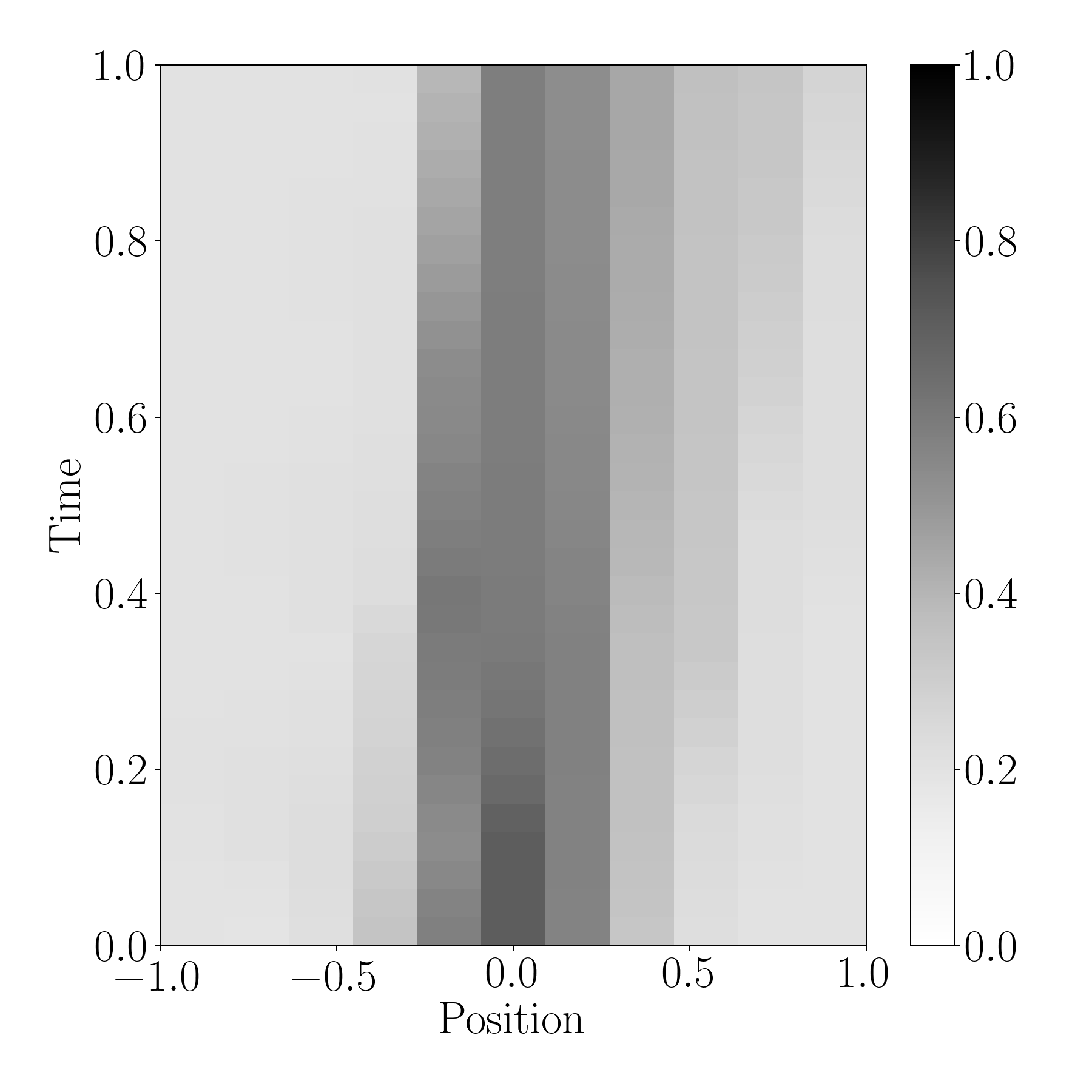}
\caption{$N_t=31, N_x=11$.}
\end{subfigure}
\hfill
\begin{subfigure}{0.24\textwidth}
\includegraphics[width=\textwidth]{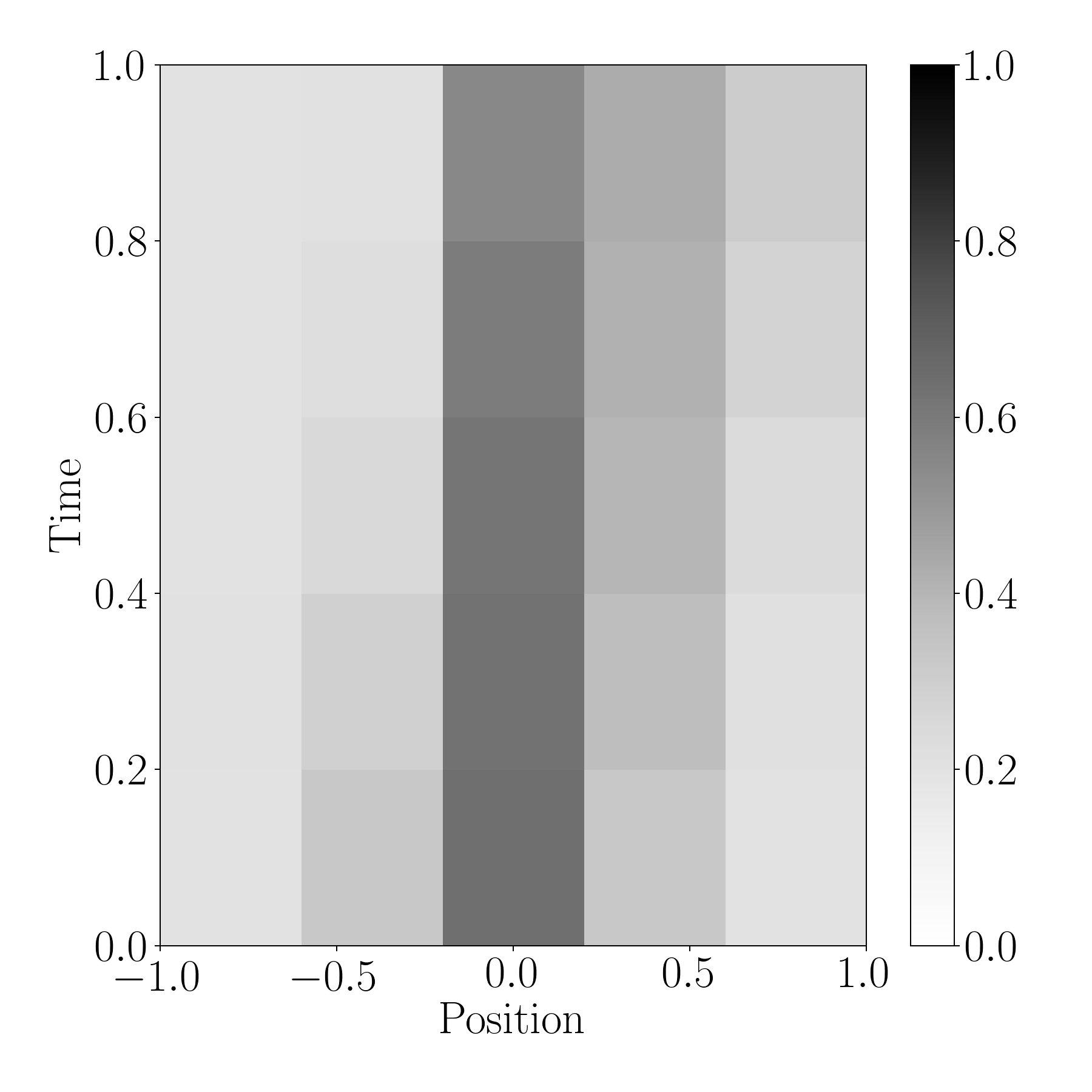}
\caption{$N_t=5, N_x=5$.}
\label{fig:coarse_dat}
\end{subfigure}
\hfill
\caption{Density matrices obtained after various choices of discretizations (into $N_x$ cells in space and $N_t$ steps in time) of the PDE solution.}
\label{fig:ex_discr_dat}
\end{figure}

The relative errors between the estimated parameters and the true value $\bar v_m =1$ are given in \Cref{tab:param_ct_all} and the RMSE values between the associated finite volume approximations and the density matrices are given in \Cref{tab:rmse_ct_all}. First, one can note that for density matrices with more than $11$ space cells, errors on the parameter estimation lower than $15\%$ (and even in some case lower than $5\%$) can be obtained. This shows that the parameter can indeed be identified from the discretized density matrices. As for the residual error on the parameter, it can be explained by the nature of the density matrices used here: indeed, comparing the PDE solution in \Cref{fig:ic_sol} and its discretizations in \Cref{fig:ex_discr_dat}, suggests that considering too coarse discretizations might smear the solution to a point where identification is no longer possible. In such cases, the features  of the original solution which could help to better identify the parameter are no longer visible in the discretized data: for instance, in \Cref{fig:coarse_dat}, the time and position where sharp changes in the PDE solution  occurred (as seen in \Cref{fig:ic_sol_mat}) are no longer identifiable.

Then, for both the TRM and the LxF scheme, the errors on the parameter and density estimations seem to only depend on the number of space discretization steps $N_x$, and not on the number of time discretization steps $N_t$. Besides, the higher the number of space subdivisions used in the scheme, the better the parameter and density estimates are. A takeaway from these results is that the quality of the parameter and density estimations can be improved independently of the time discretization of the data, by working with fine space discretization steps and by subdividing the cells in space when using the schemes.

\begin{figure}
\centering 

\begin{table}[H]
\input{fig/tab_param_all.tex}
\caption{Relative error $\vert \bar{v}_m - v_m^*\vert/\bar v_m$ on the parameter estimation for various choices of discretization steps and schemes, in case where all data is used.}
\label{tab:param_ct_all}
\end{table}
\vspace{-0.5ex}
\begin{table}[H]
\input{fig/tab_err_all.tex}
\caption{RMSE between the density data $\bm U$ and the approximated densities $\widehat{\bm U}^*$ for various choices of discretization steps and schemes, in case where all data is used.}
\label{tab:rmse_ct_all}
\end{table}

\end{figure}

Besides, when comparing the schemes, one can note that the TRM systematically and significantly outperforms the LxF scheme in terms of RMSE and generally yields better parameter estimates. To understand why, we represent in \Cref{fig:dens_ct_1151,fig:dens_ct_5151} the finite volume approximations associated with two density matrices (respectively obtained by taking $N_x=11$ and  $N_x=51$) when using both the TRM and the LxF scheme (with 5 space subdivisions). It can be observed, especially in \Cref{fig:dens_ct_5151}, that the density estimates are smoother than the original data. The LxF scheme seems to smear the solution more than the TRM, which explains the higher RMSE on the density estimates.

\begin{figure}
\centering
\begin{figure}[H]
\hfill
\begin{subfigure}{0.28\textwidth}
\includegraphics[width=\textwidth]{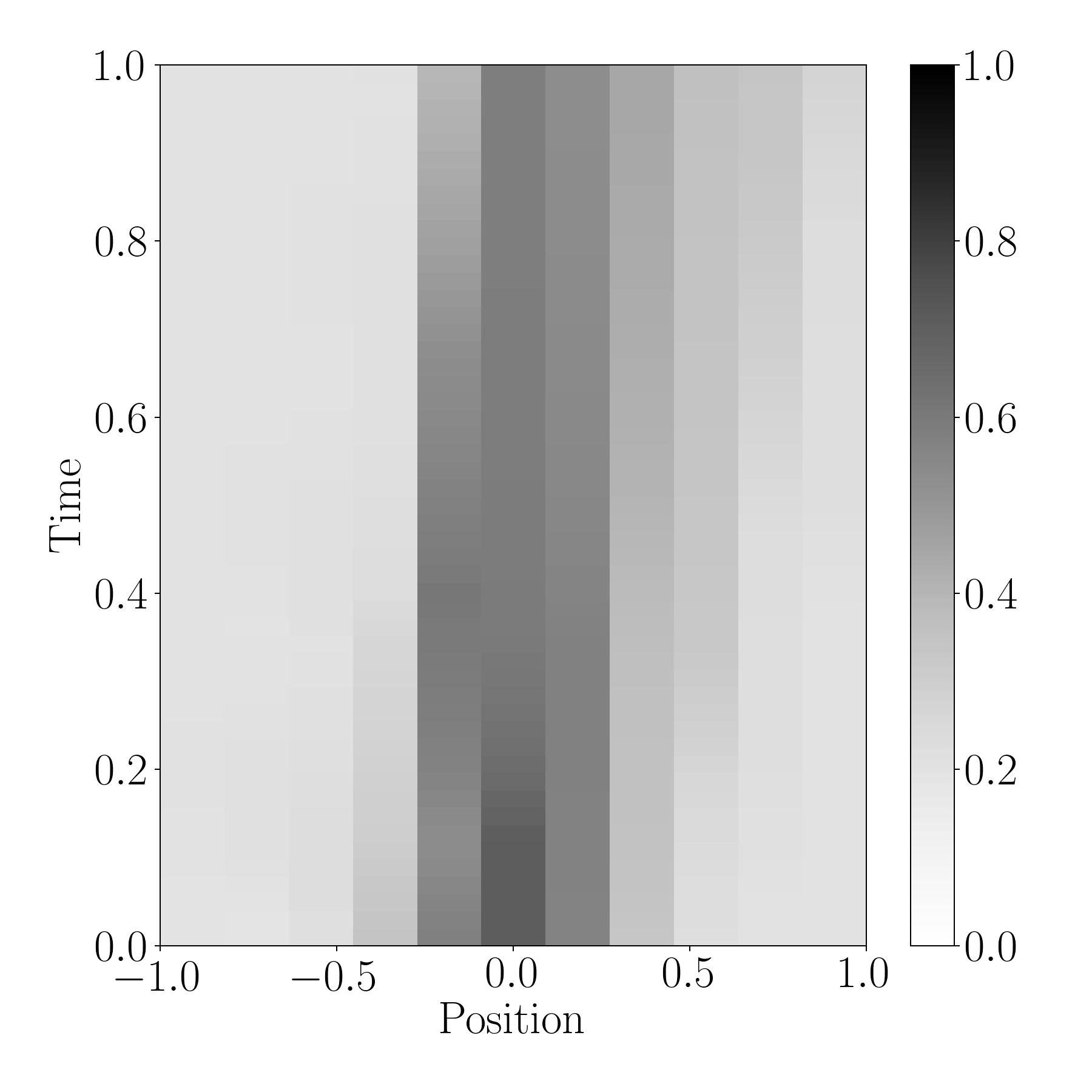}
\caption{Density data.}
\end{subfigure}
\hfill
\begin{subfigure}{0.28\textwidth}
\includegraphics[width=\textwidth]{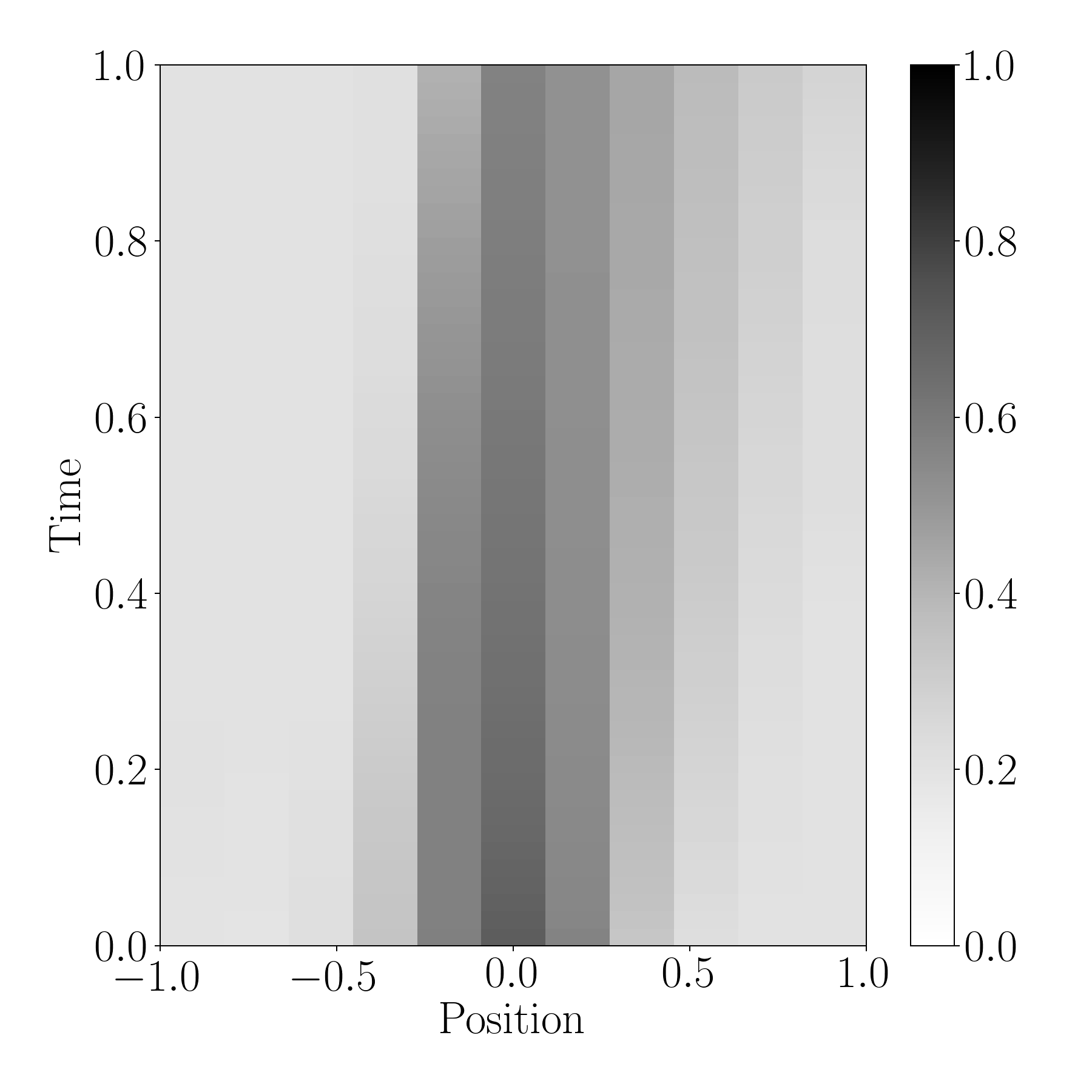}
\caption{Estimation from TRM.}
\end{subfigure}
\hfill
\begin{subfigure}{0.28\textwidth}
\includegraphics[width=\textwidth]{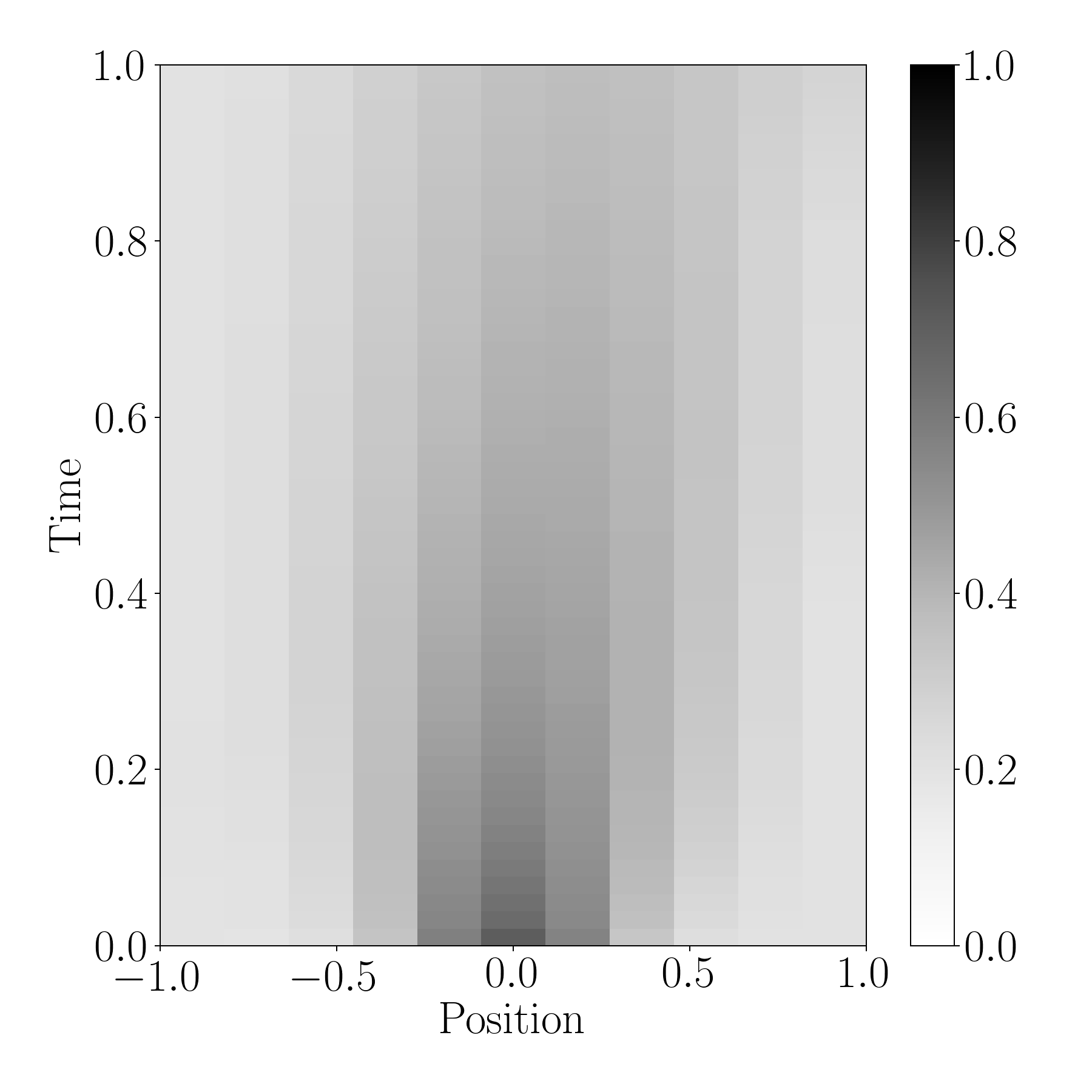}
\caption{Estimation from LxF.}
\end{subfigure}
\hfill
\caption{Estimated densities from data (a) discretized into $11$ space cells and $51$ time steps, using the TRM (b) and LxF (c) schemes with $5$ subdivisions.}
\label{fig:dens_ct_1151}
\end{figure}

\vspace{-2ex}

\begin{figure}[H]
\hfill
\begin{subfigure}{0.28\textwidth}
\includegraphics[width=\textwidth]{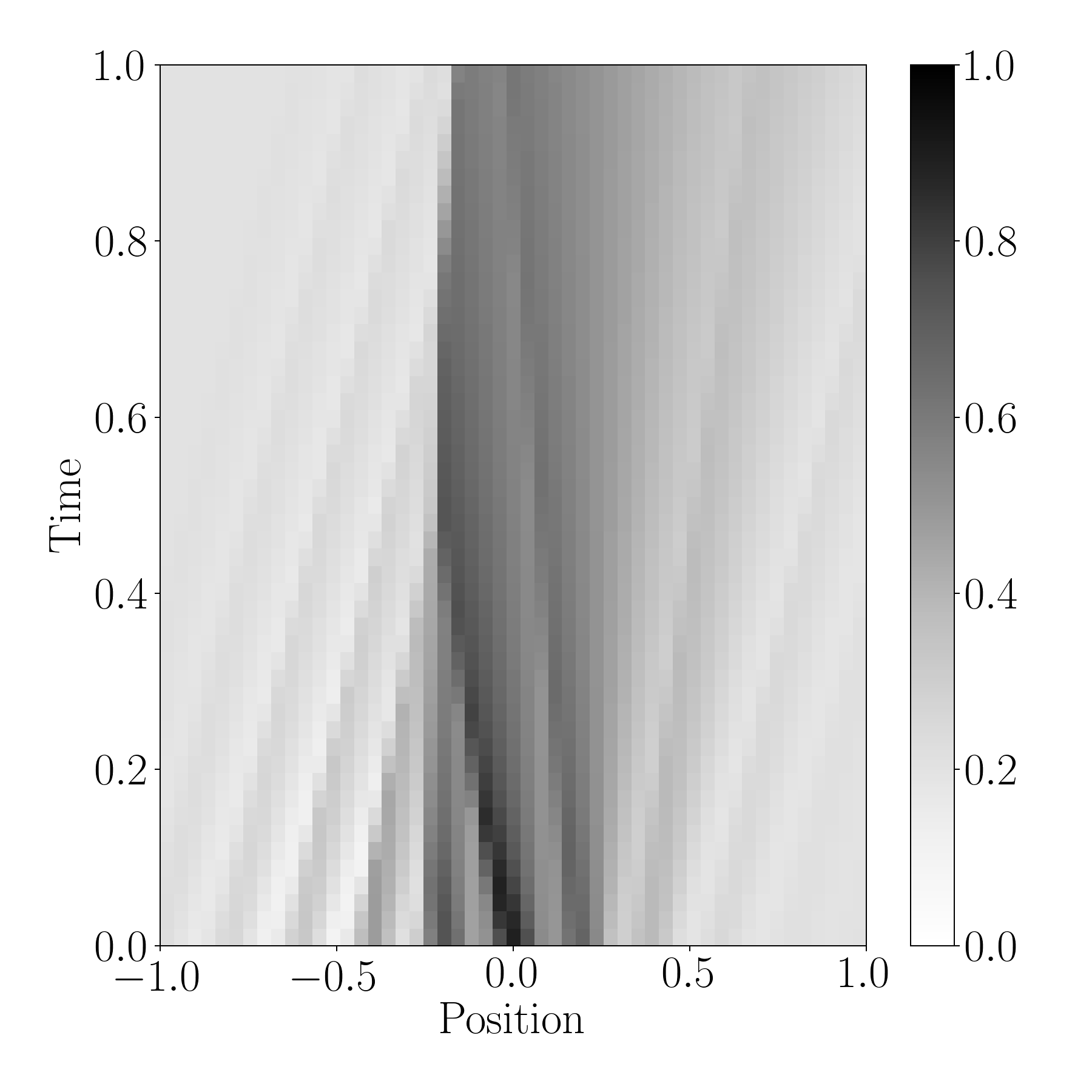}
\caption{Density data.}
\end{subfigure}
\hfill
\begin{subfigure}{0.28\textwidth}
\includegraphics[width=\textwidth]{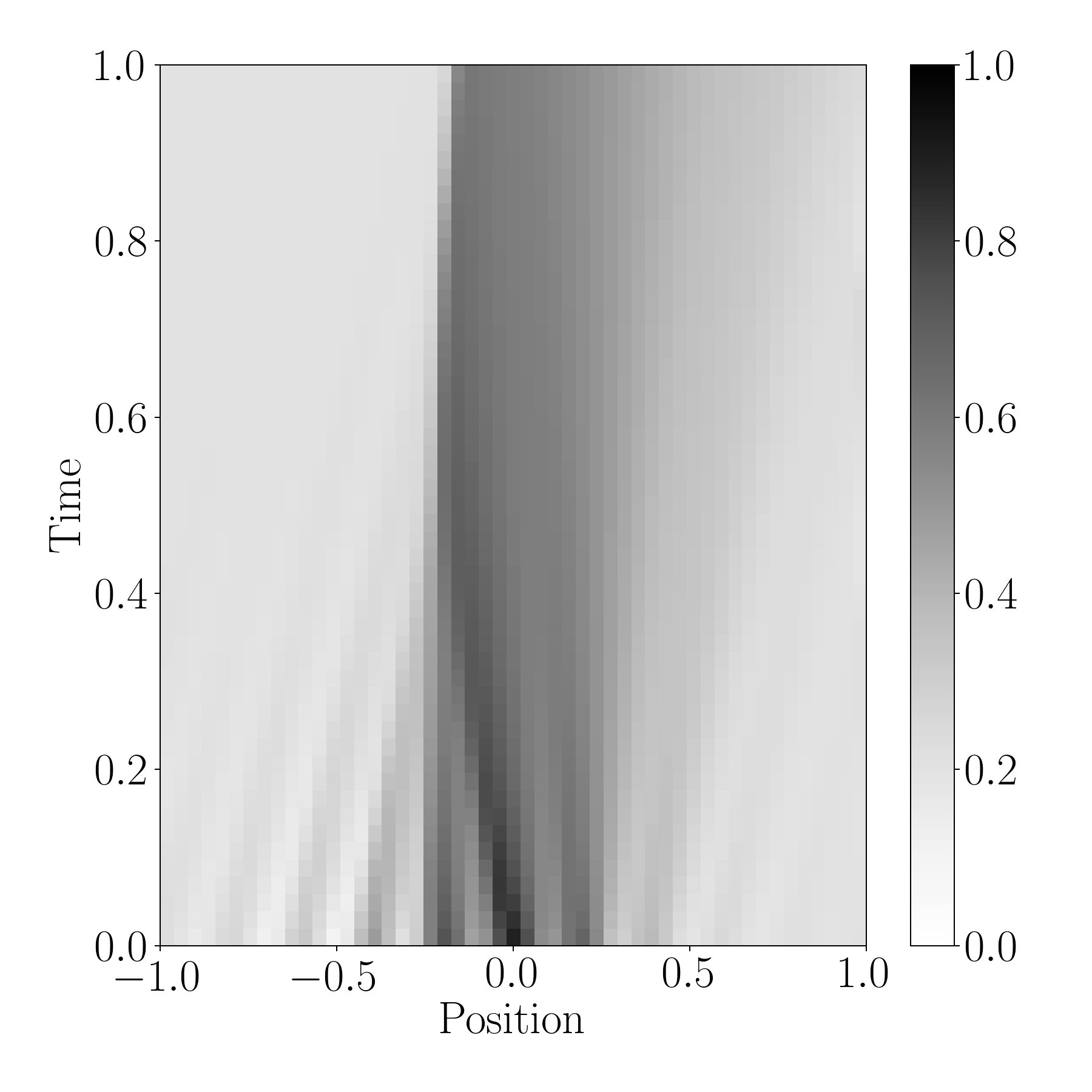}
\caption{Estimation from TRM.}
\end{subfigure}
\hfill
\begin{subfigure}{0.28\textwidth}
\includegraphics[width=\textwidth]{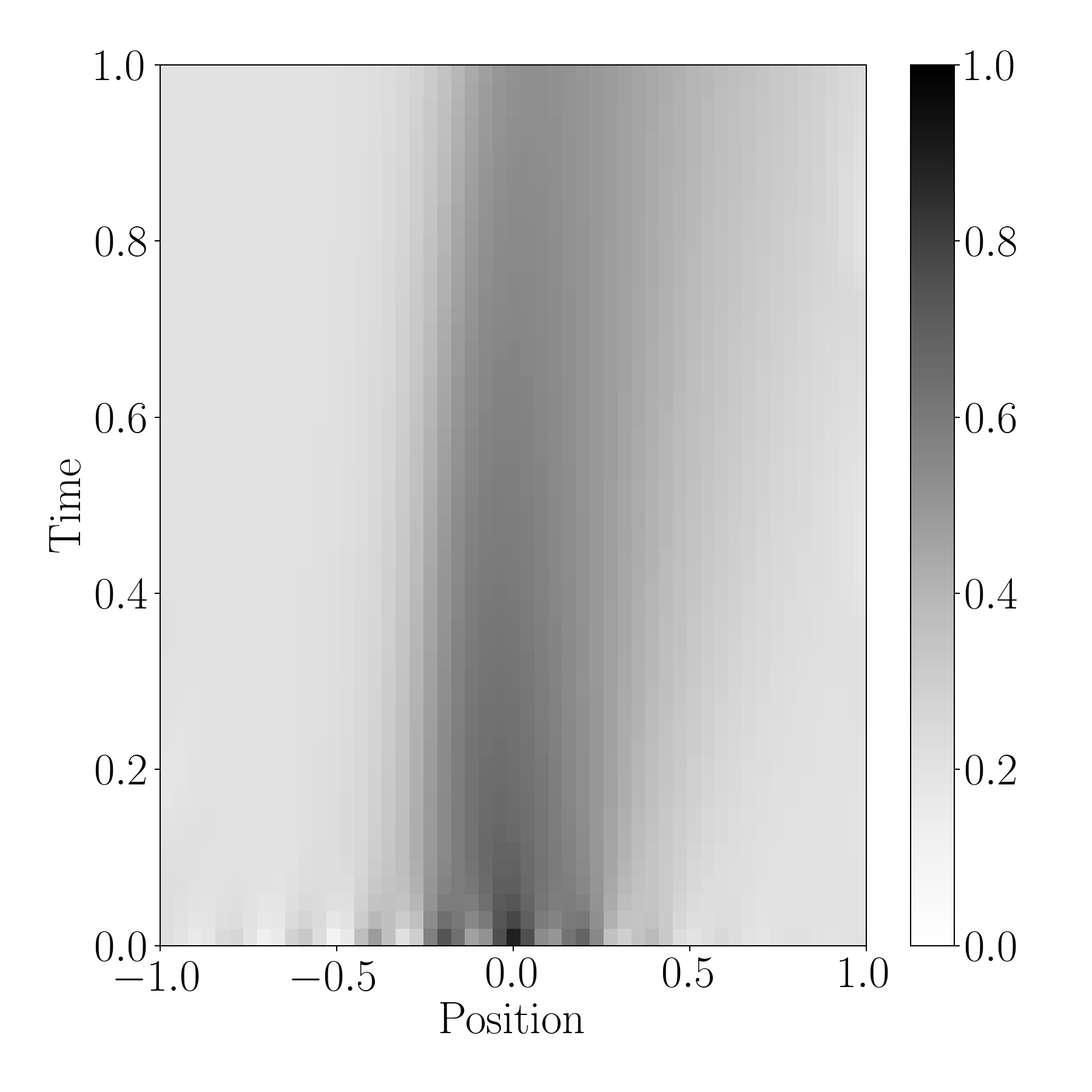}
\caption{Estimation from LxF.}
\end{subfigure}
\hfill
\caption{Estimated densities from data (a) discretized into $51$ space cells and $51$ time steps, using the TRM (b) and LxF (c) schemes with $5$ subdivisions.}
\label{fig:dens_ct_5151}
\end{figure}

\end{figure}

Finally, in order to test the robustness of the approximation, we consider the following approach. Starting from one of the previously formed density matrices, we \q{hide} some of its columns during the estimation procedure. More precisely, we carry out the parameter estimation (and density approximations) while assuming that only $3$ of the density columns are known: the first and the last one (which are used to define boundary conditions) and the center column (i.e., the $((N_x-1)/2)$-th column). However, we will always assume that the $0$-th row of $\bm U$ is observed (as it is used to define the initial state of the finite volume recurrence). 
Following \Cref{sec:adpt_pbm}, this means in particular that the cost function of the minimization problem now takes the form  
\begin{equation}
\tilde L_{[P_t,P_x]}(\theta)=\frac{1}{2}\sum_{i=1}^{N_t-1}\sum_{j_c\in I_c} \left( \frac{1}{P_x}\sum_{k=0}^{P_x-1}\widehat{U}_{k+j_cP_x}^{iP_t}(C(\theta)) - U_{j_c}^i\right)^2, \quad \theta \in\R \veq
\label{eq:costf_t}
\end{equation}
where $I_c=\left\lbrace {(N_x-1)}/{2}\right\rbrace$ is the set of observed road cells (excluding the boundary cells). Following \Cref{prop:grad_t}, the minimization of this cost function can once again be tackled using a gradient-based optimization algorithm and in particular the conjugate gradient algorithm.

The relative error between the estimated parameters and the true value $\bar v_m =1$ are then given in \Cref{tab:param_ct_mid} and the RMSE values between the associated finite volume approximations and the density matrices are given in \Cref{tab:rmse_ct_mid}.
One observes that the TRM is still able to yield good estimates of the parameter and RMSE on the density estimates that are similar to when considering the whole density matrix. However, the LxF scheme now gives poor estimates of the parameter and high-RMSE density estimates. Hence, the TRM proves to be more robust to missing input data than the LxF scheme. From now on, only the TRM will be used as finite volume scheme. In the next section, we apply the same approach as the one used in this case study to real-world measurements of density.

\begin{figure}

\begin{table}[H]
\input{fig/tab_param_mid.tex}
\caption{Relative error $\vert \bar{v}_m - v_m^*\vert/\bar v_m$ on the parameter estimation for various choices of discretization steps and schemes, in case where only the center road cell is observed through time.}
\label{tab:param_ct_mid}
\end{table}
\vspace{-0.5ex}
\begin{table}[H]
\input{fig/tab_err_mid.tex}
\caption{RMSE between the density data $\bm U$ and the approximated densities $\widehat{\bm U}^*$ for various choices of discretization steps and schemes, in case where only the center road cell is observed through time.}
\label{tab:rmse_ct_mid}
\end{table}

\end{figure}


\subsection{Mimicking real traffic dynamics}
\label{sec:app_real}

\subsubsection{Generalized density data}
\label{sec:edie}

Most tools designed to measure traffic flows and densities are based on counting the number of vehicles passing a given point of the road or present in a given section of the road. The resulting density measurements are then essentially  discrete since they depend on these discrete count variables. However, when trajectory data is available, Edie \cite{edie1963discussion} suggests a generalization of this idea that leverages the continuity of trajectories (in space and time) to yield a continuous estimation of the density. Take a region $A$ of the space-time domain on which the vehicle trajectories lie. The density of vehicles in $A$ is defined as the ratio between the time spent by all the vehicles in $A$ by the area of $A$. Similarly, the flow of vehicles in $A$ is defined as the ratio between the distance traveled by all the vehicles in $A$ by the area of $A$. Both quantities can be computed for each vehicle whose trajectory intersects $A$, using the definition represented in \Cref{fig:edie}.

\begin{figure}
    \centering
    \includegraphics[width=0.6\textwidth]{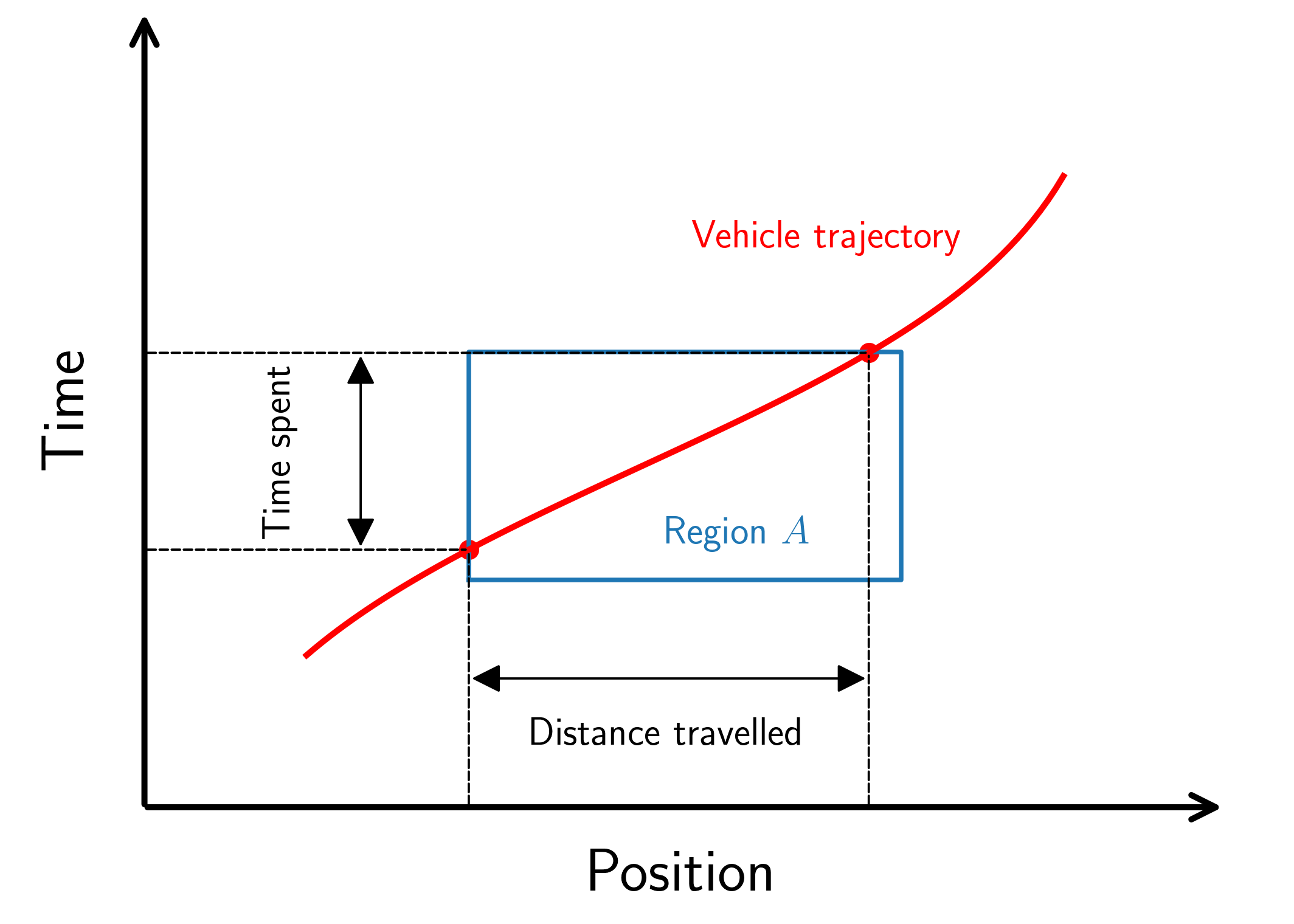}
    \caption{Quantities used in Edie's generalized definition of density and flow. }
    \label{fig:edie}
\end{figure}

Using this definition of a density measurement, it is possible to build a matrix of density measurements $\bm D$ containing the densities associated with a space-time discretization of the domain on which the trajectories lie. Indeed, we discretize this domain into a grid composed of $N_x$ cells in the space dimension and $N_t$ cells in the time dimension (see \Cref{fig:traj_edie}). Then, $\bm D$ is built as the matrix of size $N_t\times N_x$ for which the $(i, j)$-th entry, denoted by $D_j^i$, is the estimated density of the $(i,j)$-th grid cell, as obtained by applying Edie's definition on the space-time region defined by the cell.

\begin{figure}
    \centering
    \includegraphics[width=0.6\textwidth]{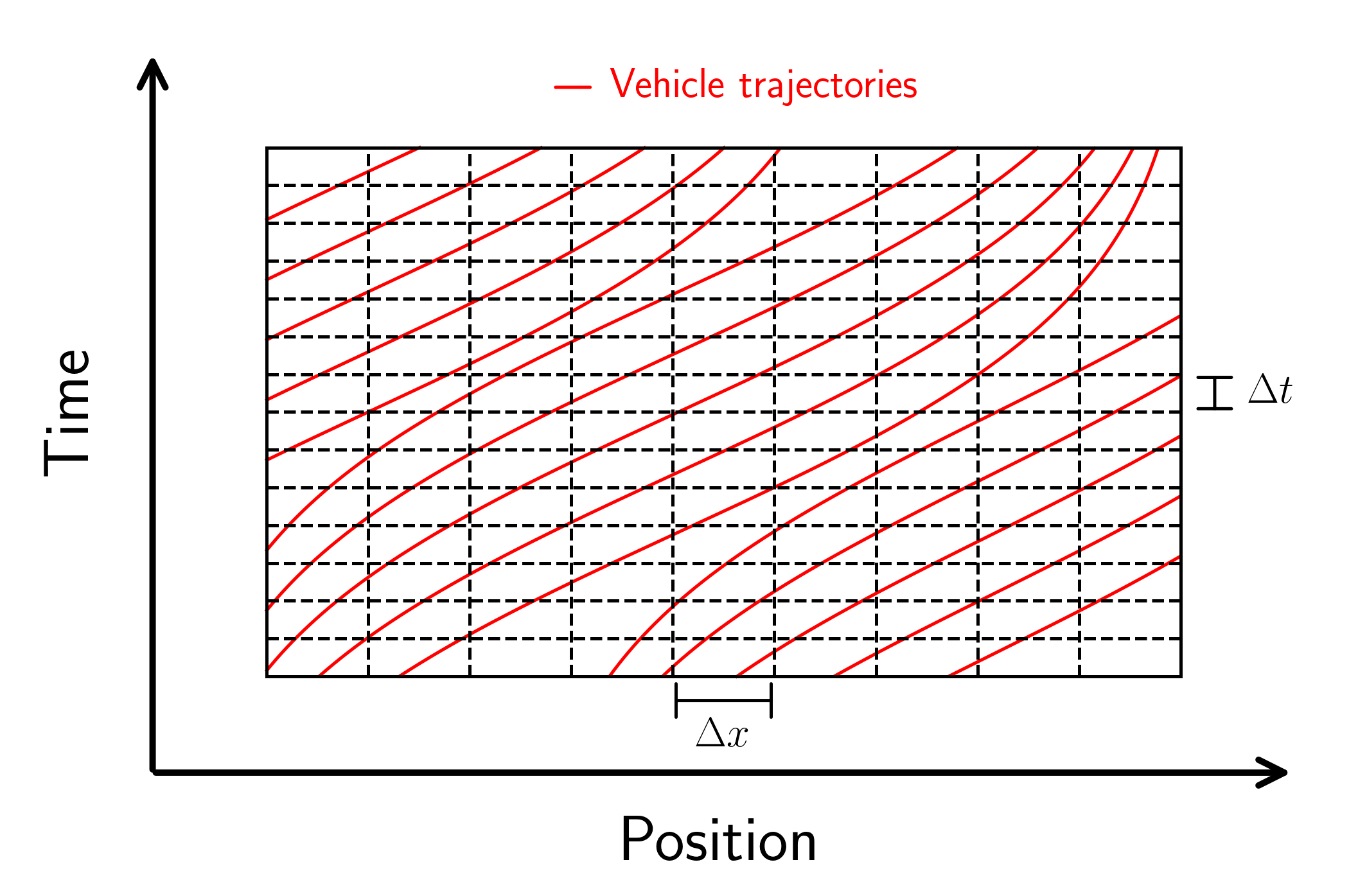}
    \caption{Space-time discretization grid used to compute the density matrix from Edie's definition, which is applied to each cell of the grid. }
    \label{fig:traj_edie}
\end{figure}

The resulting densities $D_j^i$ are assimilated to a ratio between some kind of continuous count of vehicles (given by the ratio between the total time spent by all vehicles the region $A$ and the time width of $A$) and the size of the road cells. As such, they can be considered as approximations of the cell average $\rho_j^i$ of the density function $\rho$, which is expressed as
\begin{equation*}
 D_j^i\approx \rho_j^i=\frac{1}{\Delta x} \int_{x_j-\Delta x/2}^{x_j+\Delta x/2} \rho(t_i, x) \di x \peq
\end{equation*} 
where $(t_i,x_j)$ are the coordinates of the center the $(i,j)$-th cell of the grid introduced above. The right hand side of this last equality links the cell averages of the normalized density $u$ to the density data $D_j^i$: indeed, dividing both sides of this equation by $\rho_m$ yields that the cell average $U_j^i$ of the normalized density is approximated by the ratio $D_j^i/\rho_m$.

\subsubsection{Constant parameter case}
\label{sec:real_dat_ct}

We test our approach for parameter identification (and density estimation) on real-world data. We consider in particular two density matrices, both obtained from the same trajectory data, using the approach presented in \Cref{sec:edie}. We work here with trajectory data extracted from the highD dataset, which comes from video recordings along sections of German highways \citep{highDdataset}. We consider a particular section with length of about $400\,\mathrm{m}$, which we discretize into $11$ road cells. Based on this, we build two density matrices, which we call Dataset~1 and Dataset~2, corresponding to observations of the section over time-lapses of $2\,\mathrm{min}$ taken at different times, so that Dataset~1 reflects free flow conditions only, and Dataset~2 reflects a transition between free flow conditions and a congested state. The time step used to build these density matrices is $\Delta t=2\,\mathrm{s}$. The resulting density matrices have $N_t=60$ rows and $N_x=11$ columns and can be observed in \Cref{fig:dat_real}. Finally we estimate the maximal density of the considered road by dividing the number of lanes by the mean length of the observed vehicles, which gives $\rho_m\approx 0.49\,\mathrm{m}^{-1}$.

\begin{figure}[H]
\hfill
\begin{subfigure}{0.4\textwidth}
\includegraphics[width=\textwidth]{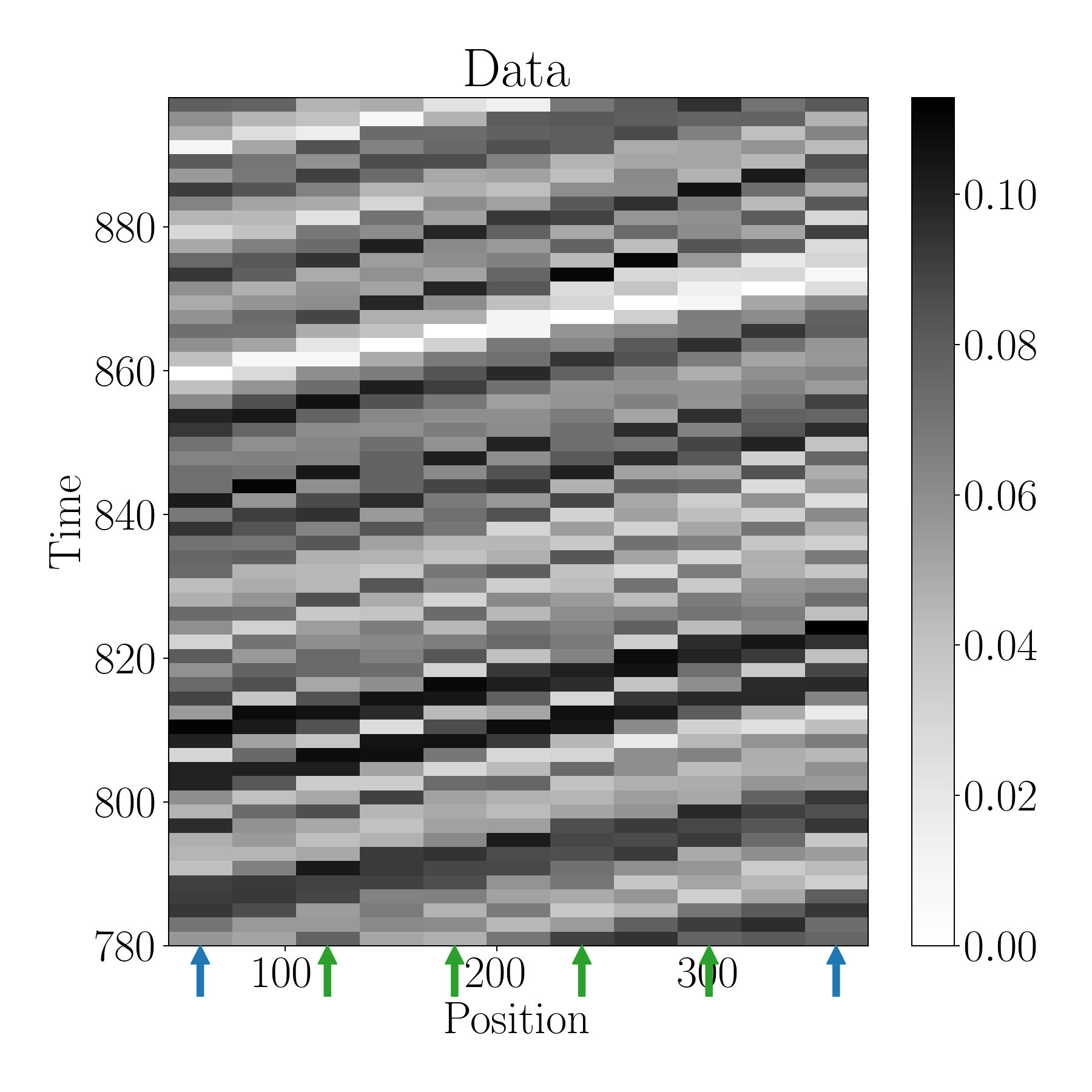}
\caption{Dataset~1.}
\end{subfigure}%
\hfill%
\begin{subfigure}{0.4\textwidth}
\includegraphics[width=\textwidth]{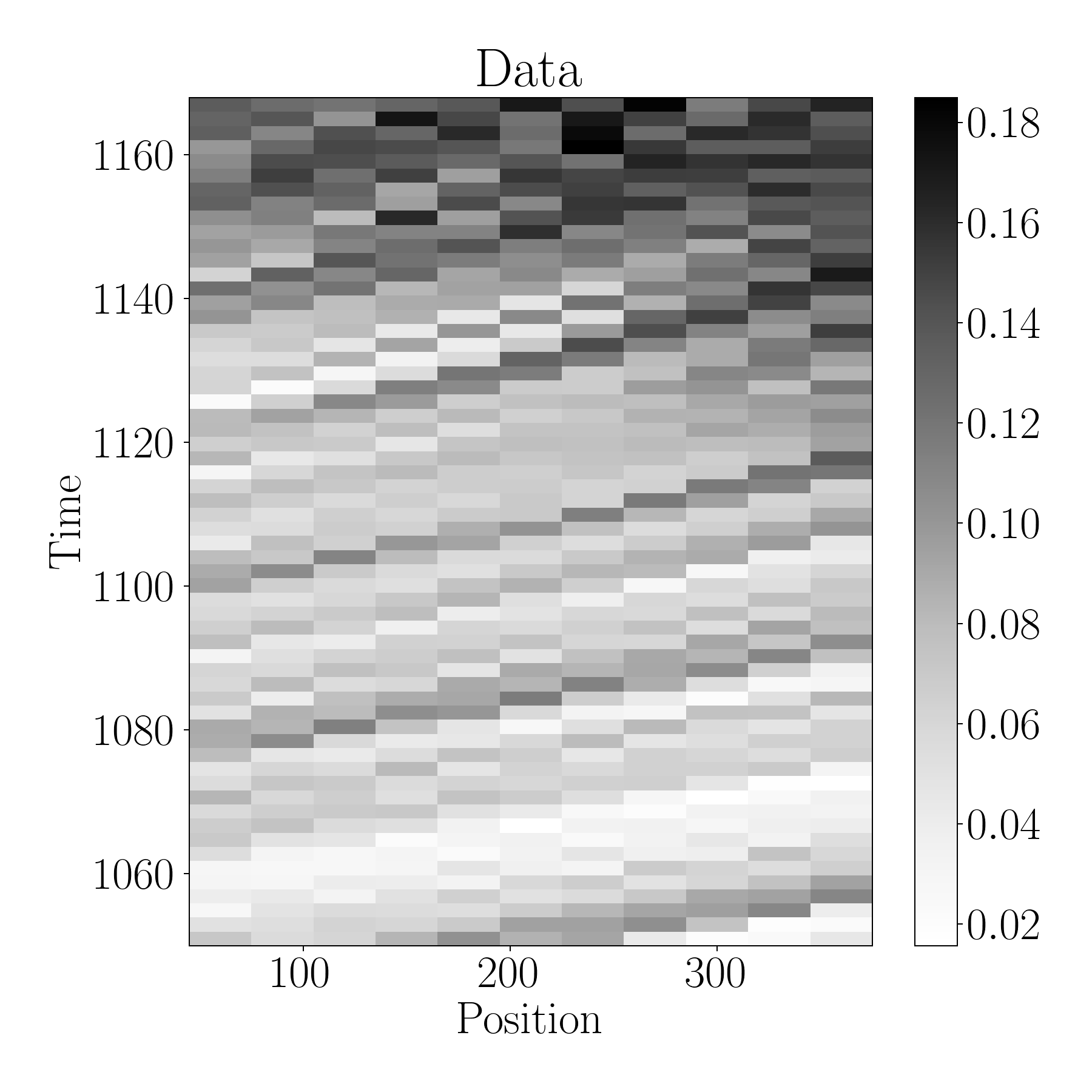}
\caption{Dataset~2.}
\end{subfigure}%
\hfill
\caption{Representation of the two datasets used in the real-data case.}
\label{fig:dat_real}
\end{figure}

For each density matrix, we compute the parameter and density estimate using the TRM scheme with the following discretization choice: the space cells are subdivided into $P_x=3$ subcells and the number of time subdivisions $P_t$ is taken to be the smallest integer so that the CFL condition~\eqref{eq:cfl_f} is satisfied for some rough estimate of the maximal speed $v_m=130 \text{ km}\cdot\text{h}^{-1}$, thus giving $P_t=15$. Once again, the first and last columns as well as the first row of the density matrix are used as boundary and initial conditions for the finite volume recurrences. Then, three cases are considered: either the whole density matrix is used and hence the cost function~\eqref{eq:costf_t} with $I_c=\lbrace 1,\dots,9\rbrace$ is minimized, or half of the columns are used and hence the cost function~\eqref{eq:costf_t} is minimized but with $I_c=\lbrace 2,4,6,8\rbrace$, or only one column is used and hence the cost function~\eqref{eq:costf_t} is minimized with $I_c=\lbrace 5\rbrace$. 

The results of these estimations are represented in \Cref{fig:real_ct_1} for Dataset~1 and in \Cref{fig:real_ct_2} for Dataset~2. In both cases we can once again notice the robustness of the estimation since removing some columns from the dataset does not affect significantly the value of the estimated parameter or the RMSE of the estimated densities. Besides, one can note that TRM seems to smooth the true evolution of the densities. When comparing the results obtained for both datasets,  the RMSE for Dataset~1 is significantly lower than that of Dataset~2, which can be explained by comparing visually the estimated densities in both cases.

For Dataset~1, the TRM was able to recreate the linear trends of density values appearing in the dataset and that are characteristic of free flow conditions: indeed, in this case, the vehicles are able to travel freely across the road and hence the vehicles can transfer from one cell to the next undisturbed. Therefore, modeling this vehicle transfer with a unique and constant reaction rate, as the TRM does, seems appropriate. For Dataset~2, however, congestion appears in the dataset and hence there is a change in the conditions with which vehicles can transfer from one cell to the other. A single reaction rate becomes now a more controversial choice, which is confirmed by the fact that the estimated densities do not depict the same congestion as in the data.

\begin{figure}

\begin{figure}[H]
\hfill
\begin{subfigure}{0.249\textwidth}
\includegraphics[width=\textwidth]{fig/Data.png}
\caption{Dataset~1.}
\end{subfigure}%
\hfill%
\begin{subfigure}{0.249\textwidth}
\includegraphics[width=\textwidth]{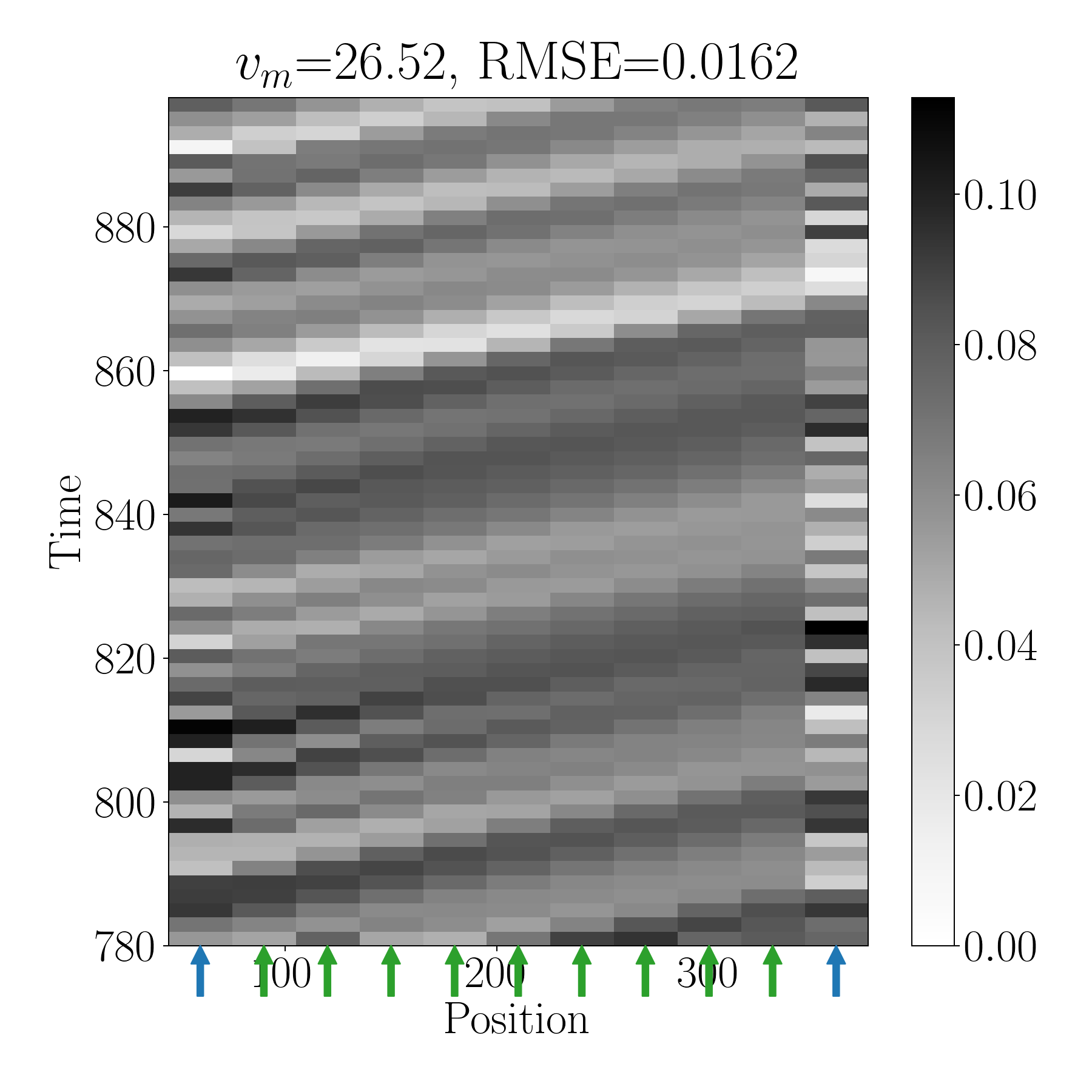}
\caption{Using all the dataset.}
\end{subfigure}%
\hfill%
\begin{subfigure}{0.249\textwidth}
\includegraphics[width=\textwidth]{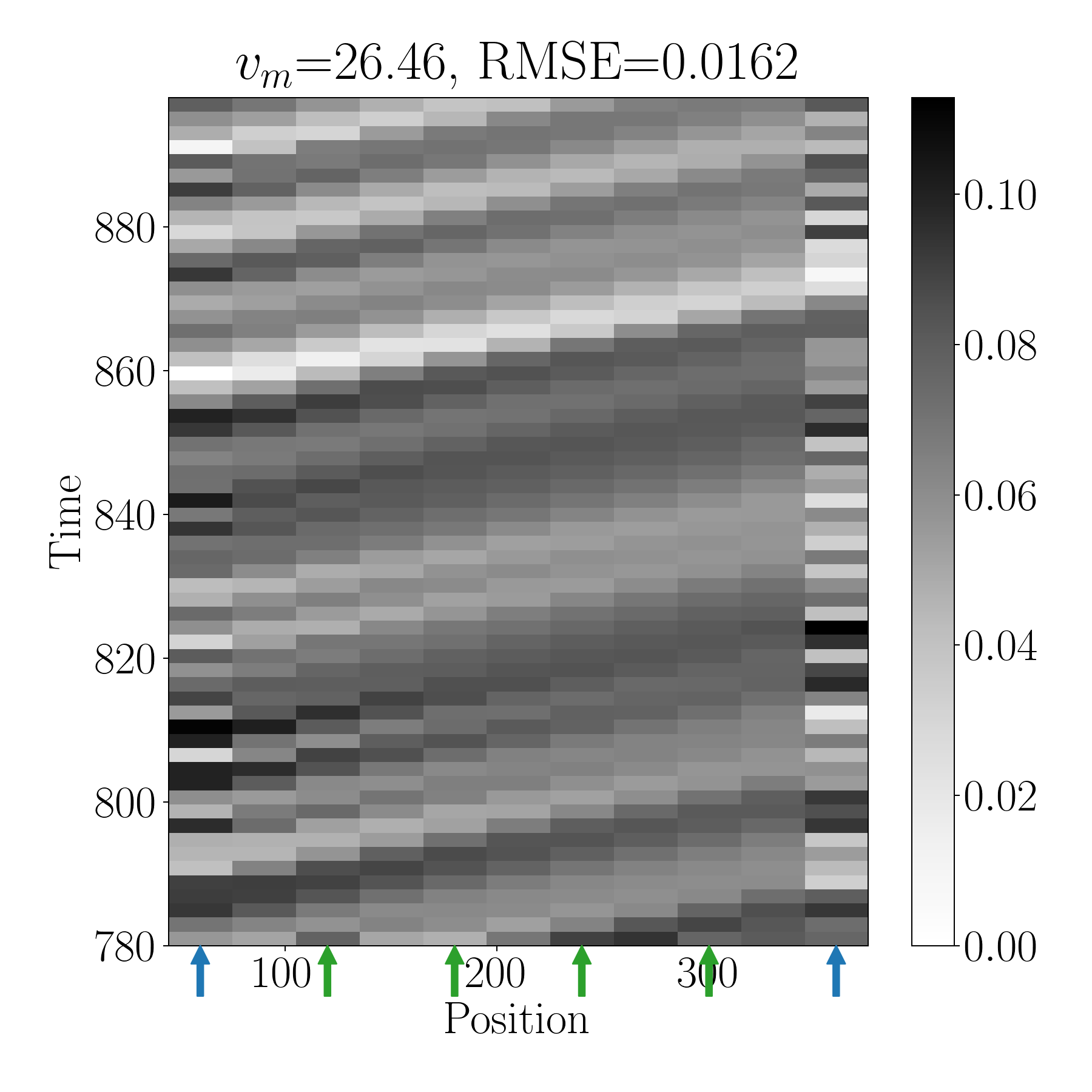}
\caption{Using 4 columns.}
\end{subfigure}%
\hfill%
\begin{subfigure}{0.249\textwidth}
\includegraphics[width=\textwidth]{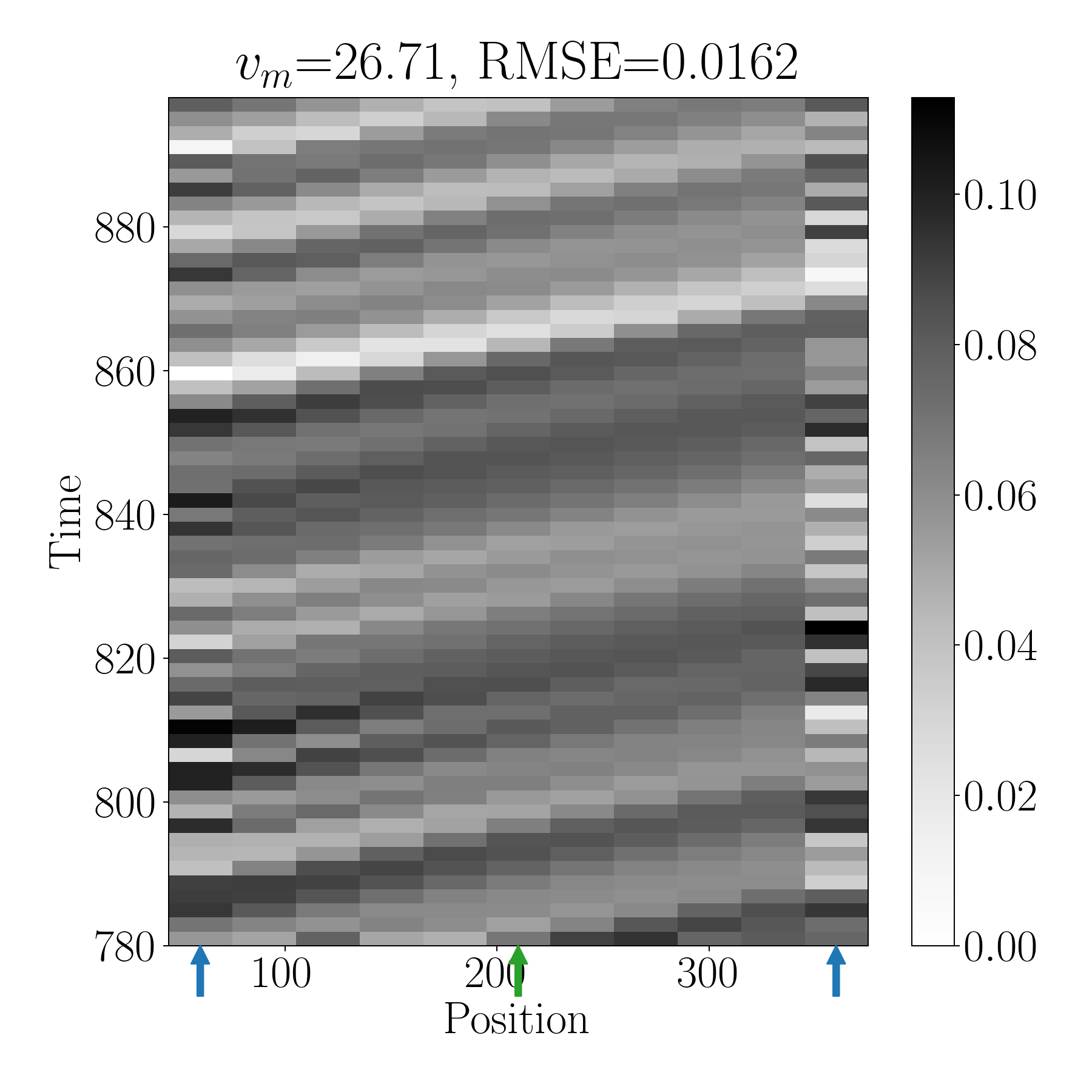}
\caption{Using 1 column.}
\end{subfigure}%
\hfill
\caption{Estimated densities using different subsets of the Dataset~1 (and the TRM), and associated estimated parameter $v_m$ and RMSE. The blue arrows point to the columns used for the boundary conditions and the green arrows to the ones used in the minimization problem.}
\label{fig:real_ct_1}
\end{figure}

\vspace{-2ex}

\begin{figure}[H]
\hfill
\begin{subfigure}{0.249\textwidth}
\includegraphics[width=\textwidth]{fig/Data_both.png}
\caption{Dataset~2.}
\end{subfigure}%
\hfill%
\begin{subfigure}{0.249\textwidth}
\includegraphics[width=\textwidth]{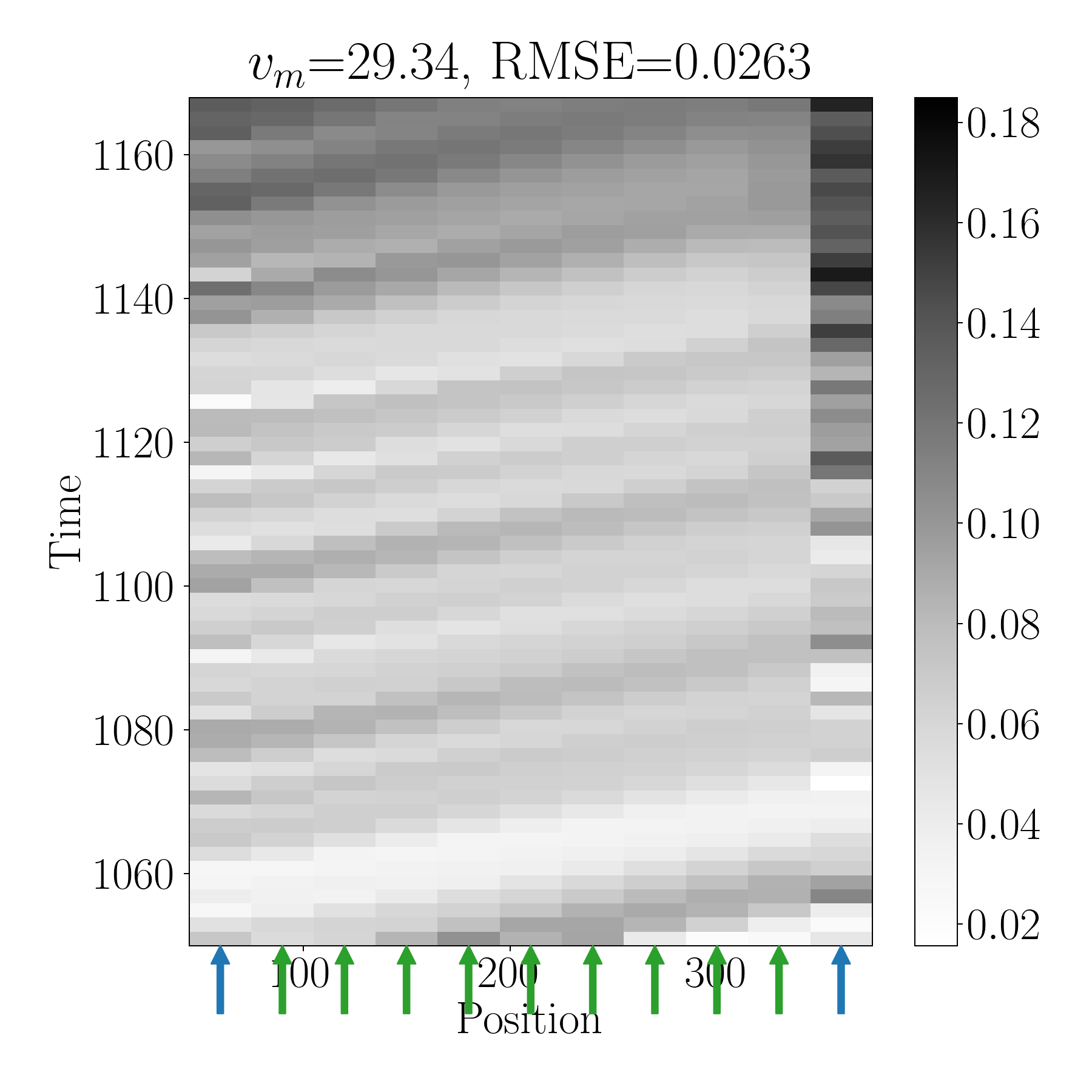}
\caption{Using all the dataset.}
\end{subfigure}%
\hfill%
\begin{subfigure}{0.249\textwidth}
\includegraphics[width=\textwidth]{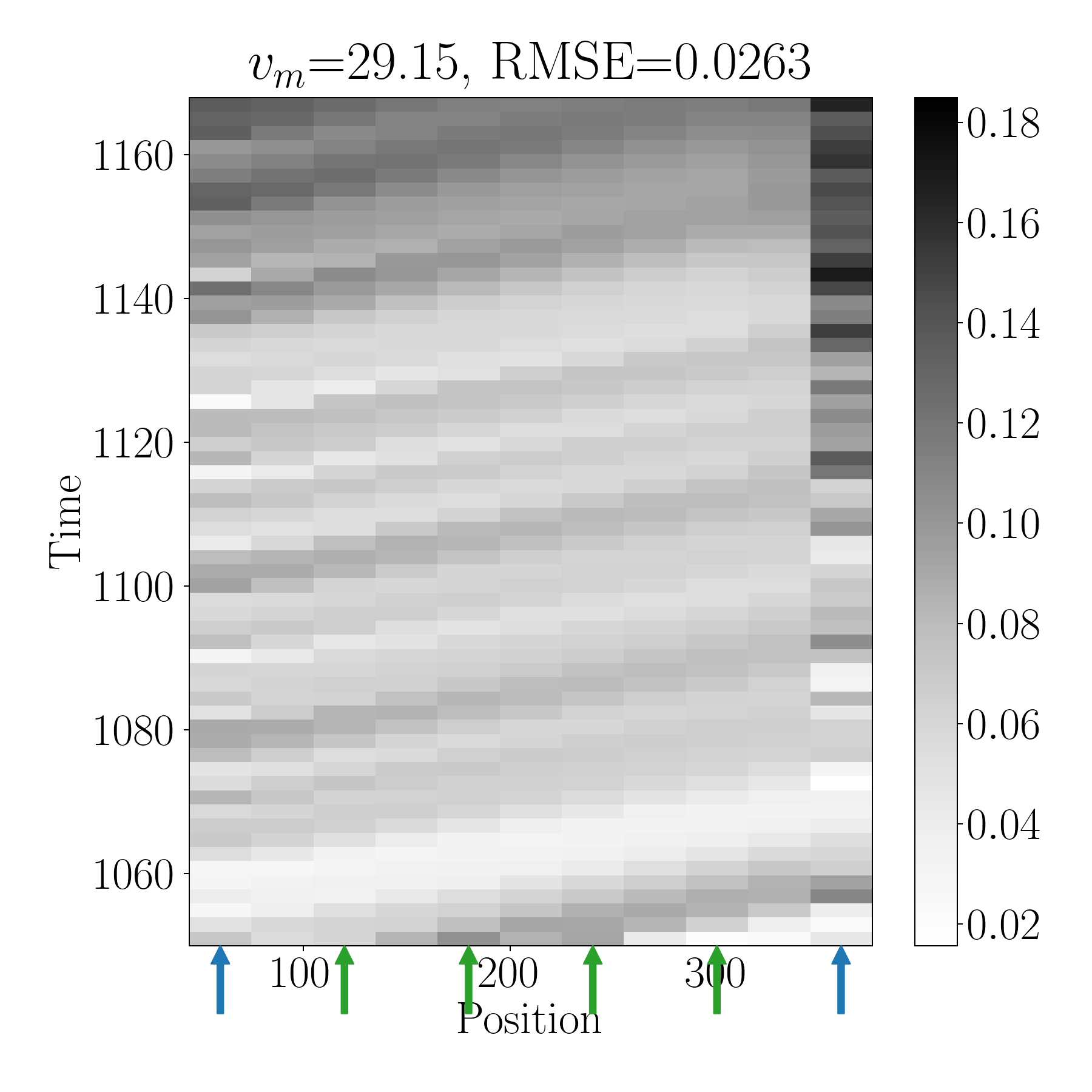}
\caption{Using 4 columns.}
\end{subfigure}%
\hfill%
\begin{subfigure}{0.249\textwidth}
\includegraphics[width=\textwidth]{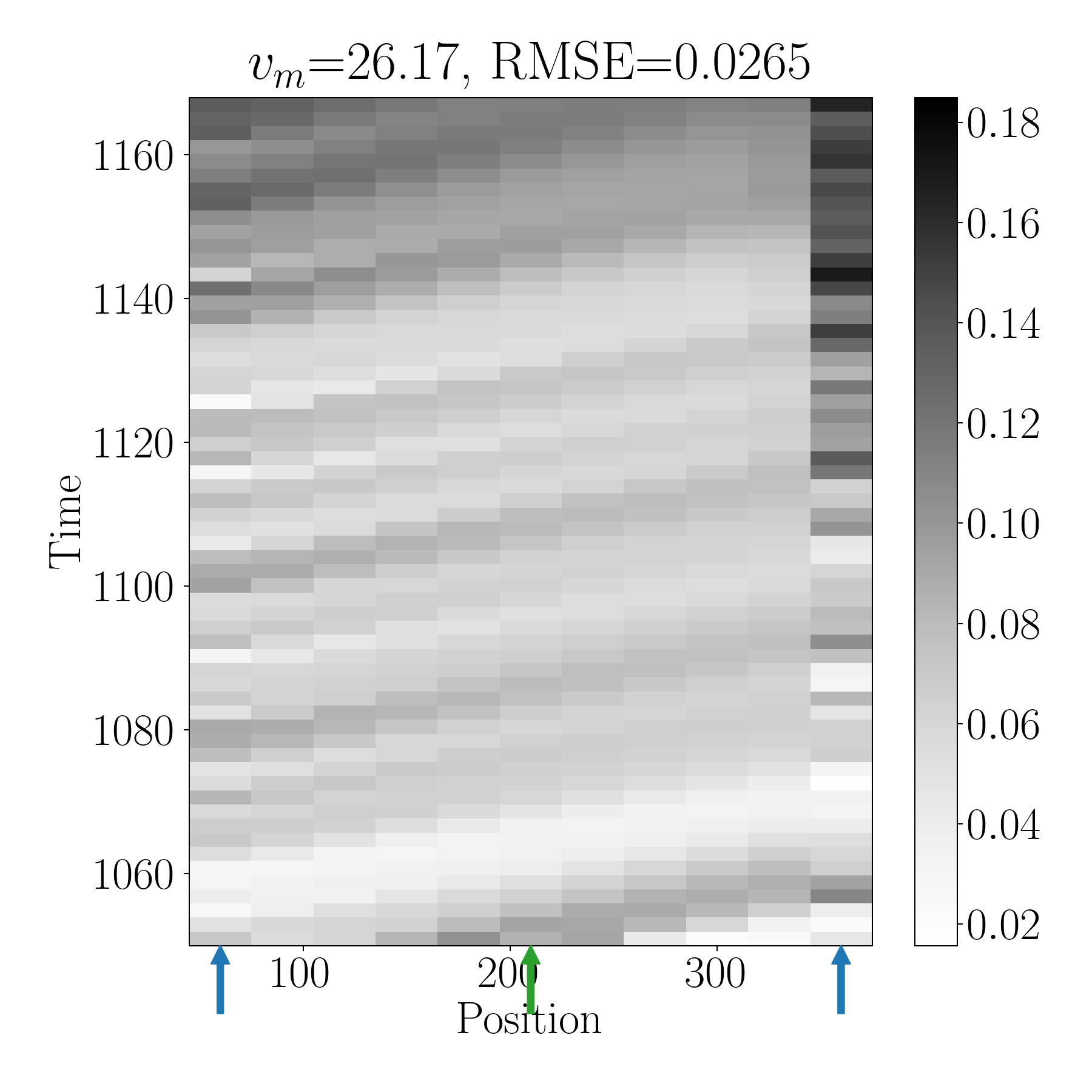}
\caption{Using 1 column.}
\end{subfigure}%
\hfill
\caption{Estimated densities using different subsets of the Dataset~2 (and the TRM), and associated estimated parameter $v_m$ and RMSE. The blue arrows point to the columns used for the boundary conditions and the green arrows to the ones used in the minimization problem.}
\label{fig:real_ct_2}
\end{figure}

\end{figure}

Another way to understand the difference of approximation quality between both datasets is to compare their fundamental diagrams. A fundamental diagram is a scatter-plot representing density measurements against flux measurements done at the same time and space locations. In our case, flux measurements associated to our datasets can be computed from the trajectory data using once gain the approach in \Cref{sec:edie}. As for the flux \q{measurements} associated with the estimated densities, we use the quadratic flux-density relationship~\eqref{eq:f} assumed by the LWR model, and plug in the estimated parameter $v_m^*$. The resulting fundamental diagrams are shown in \Cref{fig:fd_ct}. For Dataset~1, the true fundamental diagram looks quite linear, as expected for free flow conditions, and the quadratic flux of the estimation process then yields an adequate approximation. However, for Dataset~2, the quadratic flux fails to give a good approximation of the  true fundamental diagram, which now shows a mix of linear trend and more diffuse point pattern. In order to improve these estimations, we propose to offer more flexibility to the models by adding new (and physically meaningful) parameters. This is the purpose of the next section.

\begin{figure}
\hfill
\begin{subfigure}{0.4\textwidth}
\includegraphics[width=\textwidth]{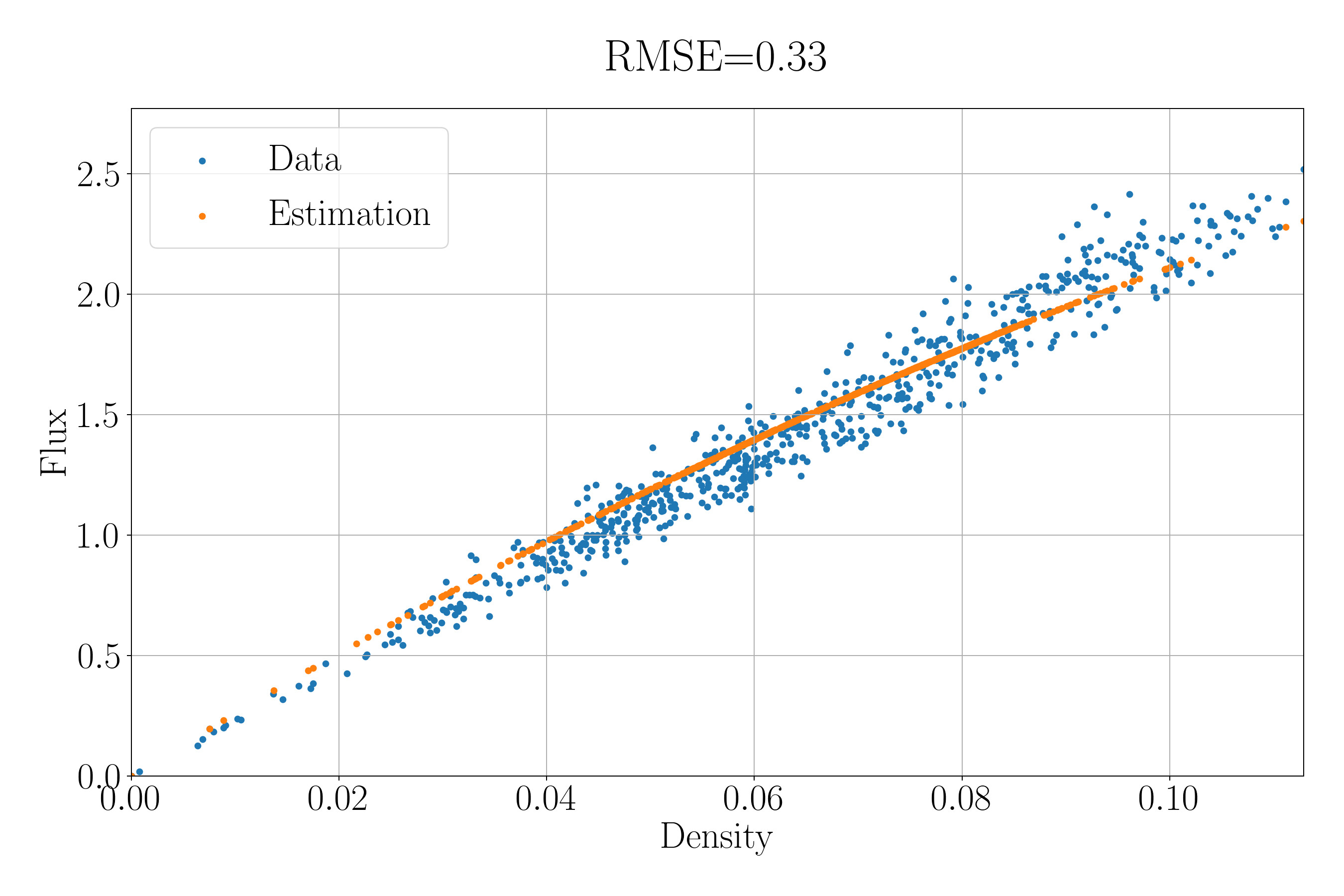}
\caption{Dataset~1.}
\end{subfigure}%
\hfill%
\begin{subfigure}{0.4\textwidth}
\includegraphics[width=\textwidth]{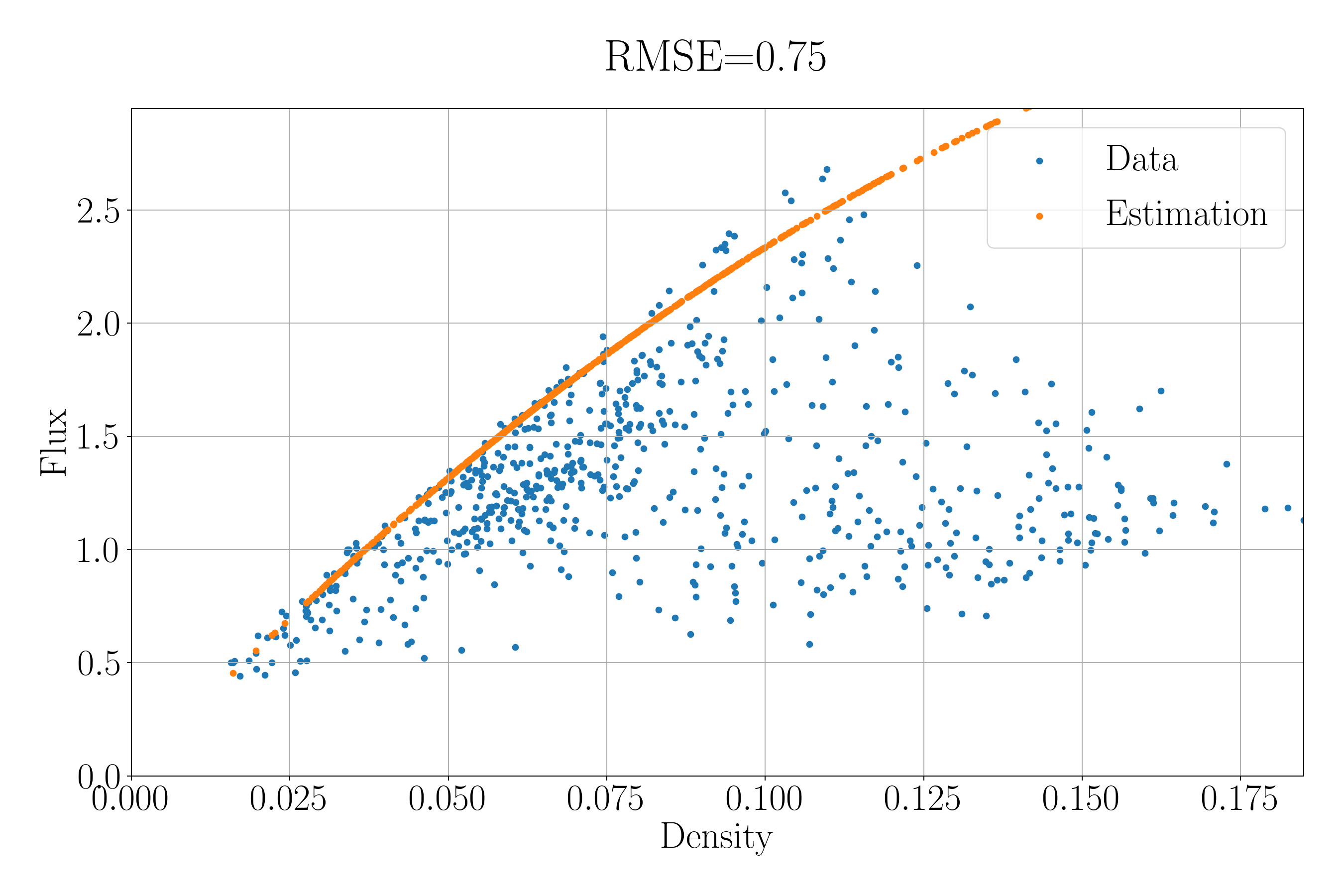}
\caption{Dataset~2.}
\end{subfigure}%
\hfill
\caption{Comparison between the fundamental diagram of the two datasets and their estimations through the TRM.}
\label{fig:fd_ct}
\end{figure}


\subsubsection{Varying parameter case}

Starting from the two datasets introduced in \Cref{sec:real_dat_ct}, we use the same least-square minimization approach to derive the values of the now varying parameter $v_m$. We consider three cases:
\begin{itemize}
\item the parameter $v_m$ is time-dependent only, meaning that the scaling parameters of the finite volume recurrence satisfy that for any $n \in\bi 0, N_t -1\ei$, there exists $C^n\in (0, 1/2)$ such that for any $j\in \bi 0, N_x\ei$, $C_{j}^n=C^n$. Hence, the actual number of parameters to be estimated in this case is $N_t$.
\item the parameter $v_m$ is space-dependent only, meaning that the scaling parameters of the finite volume recurrence satisfy that for any  $j\in \bi 0, N_x\ei$, there exists $C^n\in (0, 1/2)$ such that for any $n \in\bi 0, N_t-1 \ei$, $C_{j}^n=C_j$. Hence, the actual number of parameters to be estimated in this case is $N_x$.
\item the parameter $v_m$ is space-time-dependent, and hence, the actual number of parameters to be estimated in this case is $(N_x+1)N_t$.
\end{itemize}
The estimations are carried out while considering half of the columns of the density matrices (hence $I_c=\lbrace 2,4,6,8\rbrace$) and the parameter $\lambda$ are set so that the overall RMSE between the estimated densities and the whole density matrix is minimized. A TRM with 3 space subdivisions (and 15 time subdivisions) is used as a finite volume scheme, which is the scheme used for the robustness study in the constant case (cf. \Cref{sec:real_dat_ct}, results in \Cref{fig:real_ct_1,fig:real_ct_2}). The results are presented in \Cref{fig:real_t} for the time-dependent case, \Cref{fig:real_s} for the space-dependent case and \Cref{fig:real_st} for the space-time-dependent case.

\begin{figure}
\begin{minipage}{0.24\textwidth}
\begin{subfigure}[c]{\textwidth}
\includegraphics[width=\textwidth]{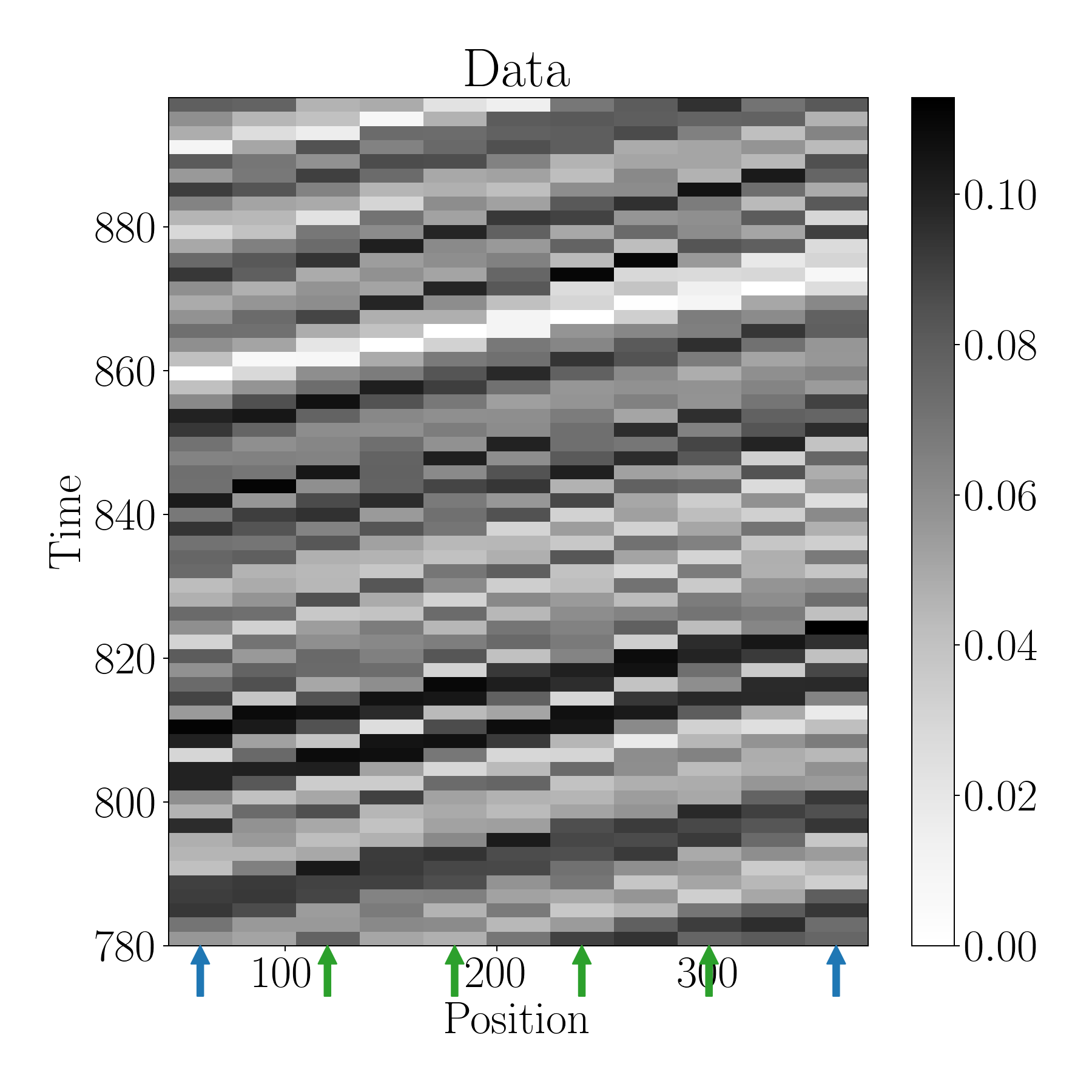}
\caption{Dataset~1.}
\end{subfigure}
\begin{subfigure}[c]{\textwidth}
\includegraphics[width=\textwidth]{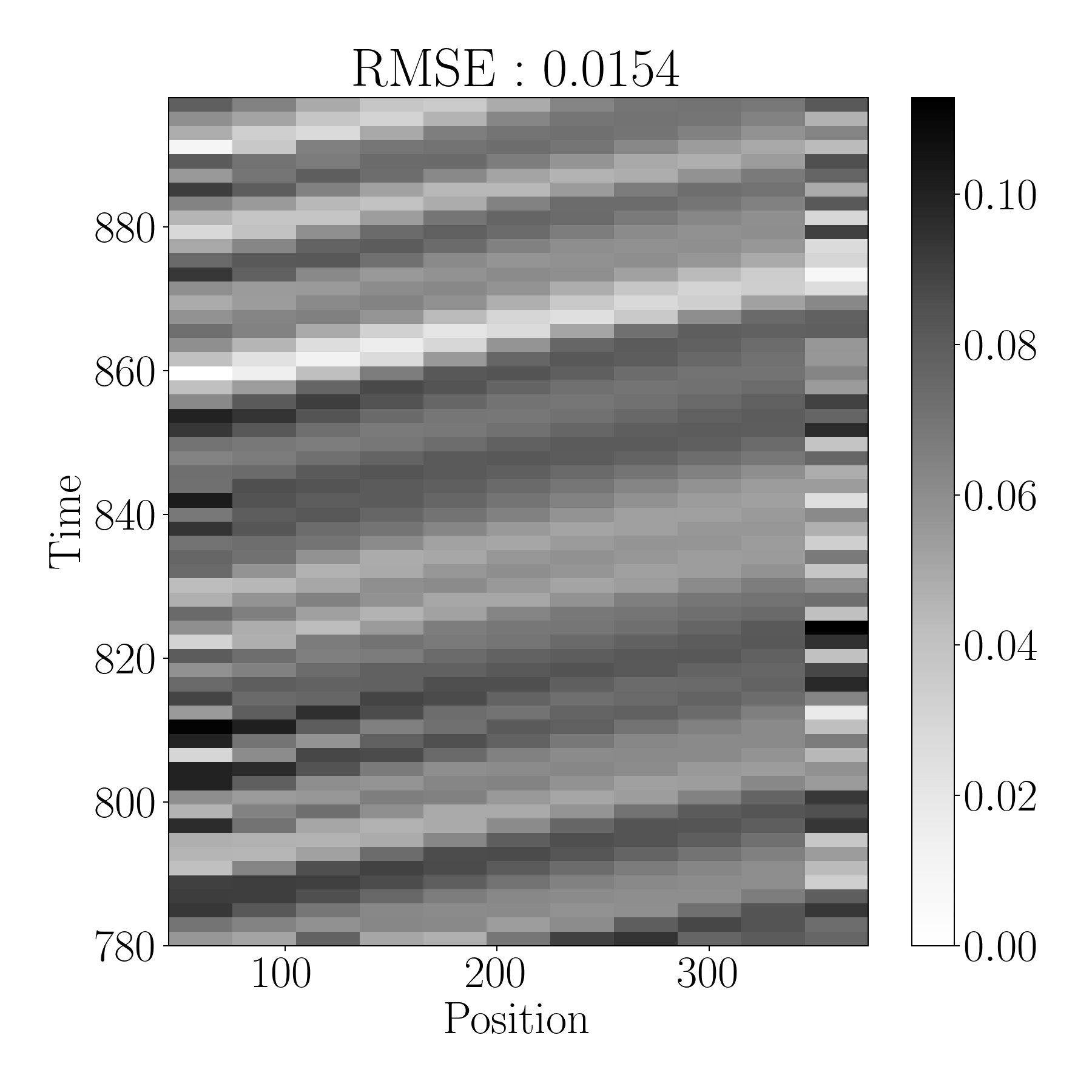}
\caption{Estimated densities for Dataset~1.}
\end{subfigure}
\end{minipage}%
\begin{minipage}{0.249\textwidth}
\begin{subfigure}[c]{\textwidth}
\includegraphics[width=\textwidth]{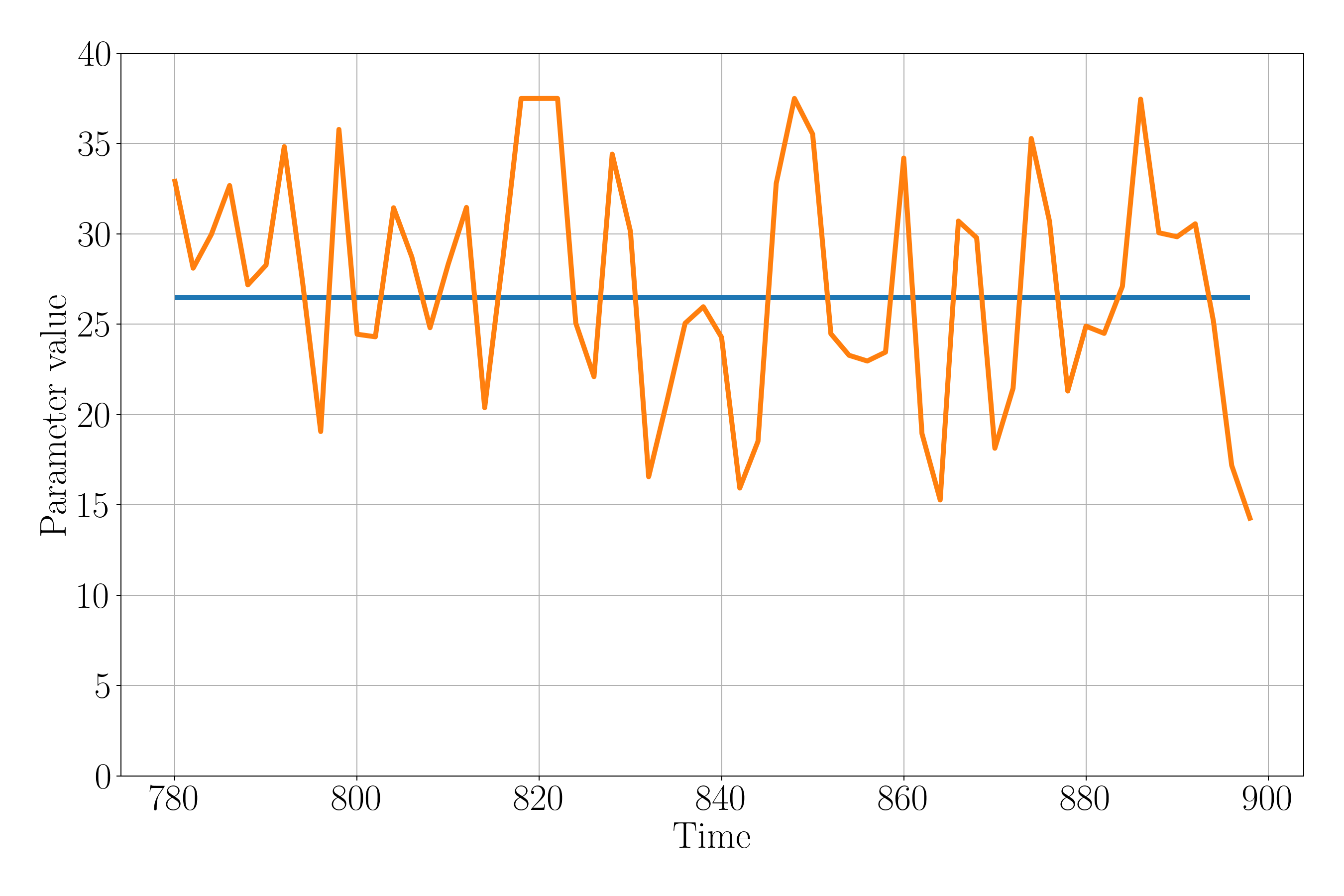}
\caption{Estimated parameters for Dataset~1.}
\end{subfigure}
\end{minipage}%
\hfill
\begin{minipage}{0.24\textwidth}
\begin{subfigure}[c]{\textwidth}
\includegraphics[width=\textwidth]{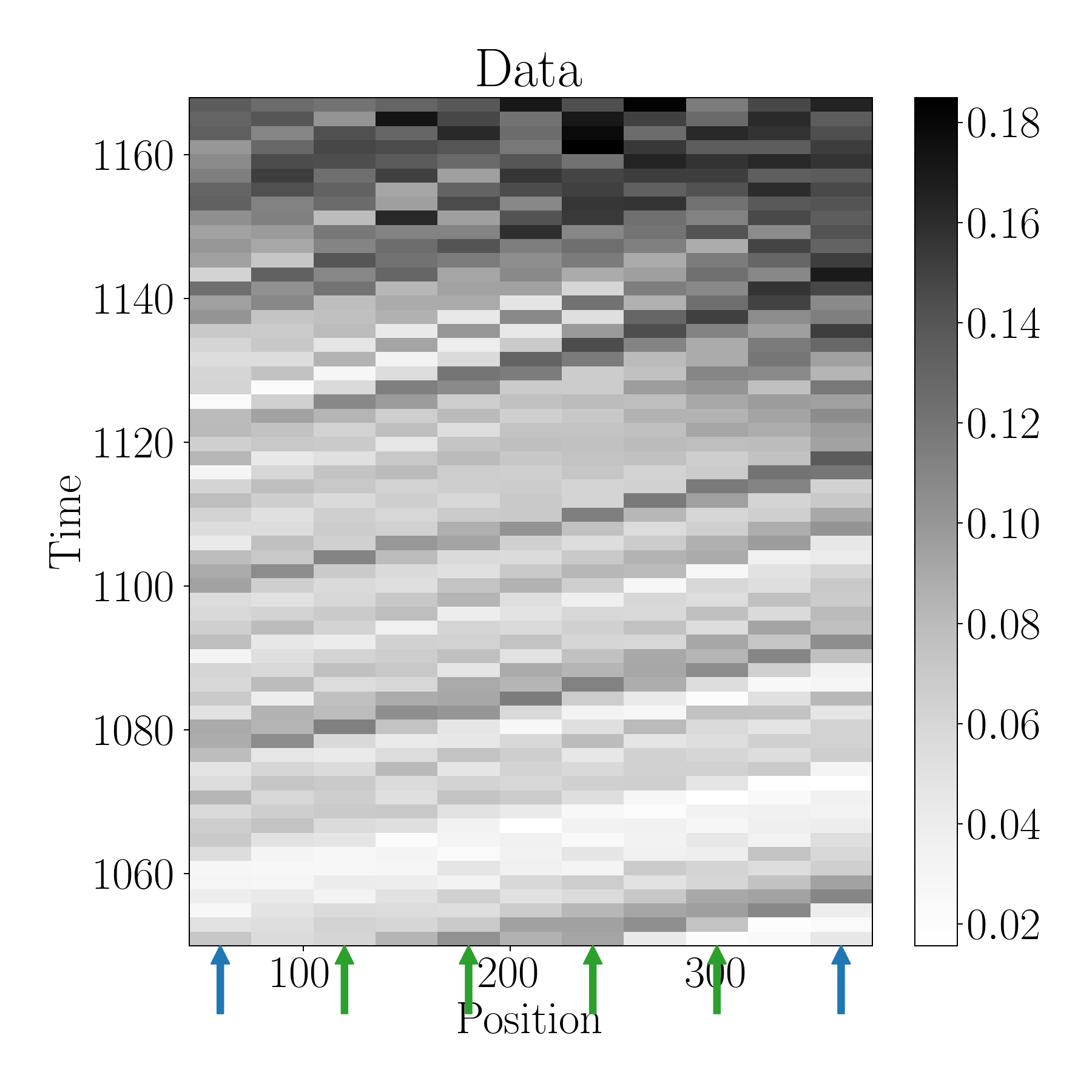}
\caption{Dataset~2.}
\end{subfigure}
\begin{subfigure}[c]{\textwidth}
\includegraphics[width=\textwidth]{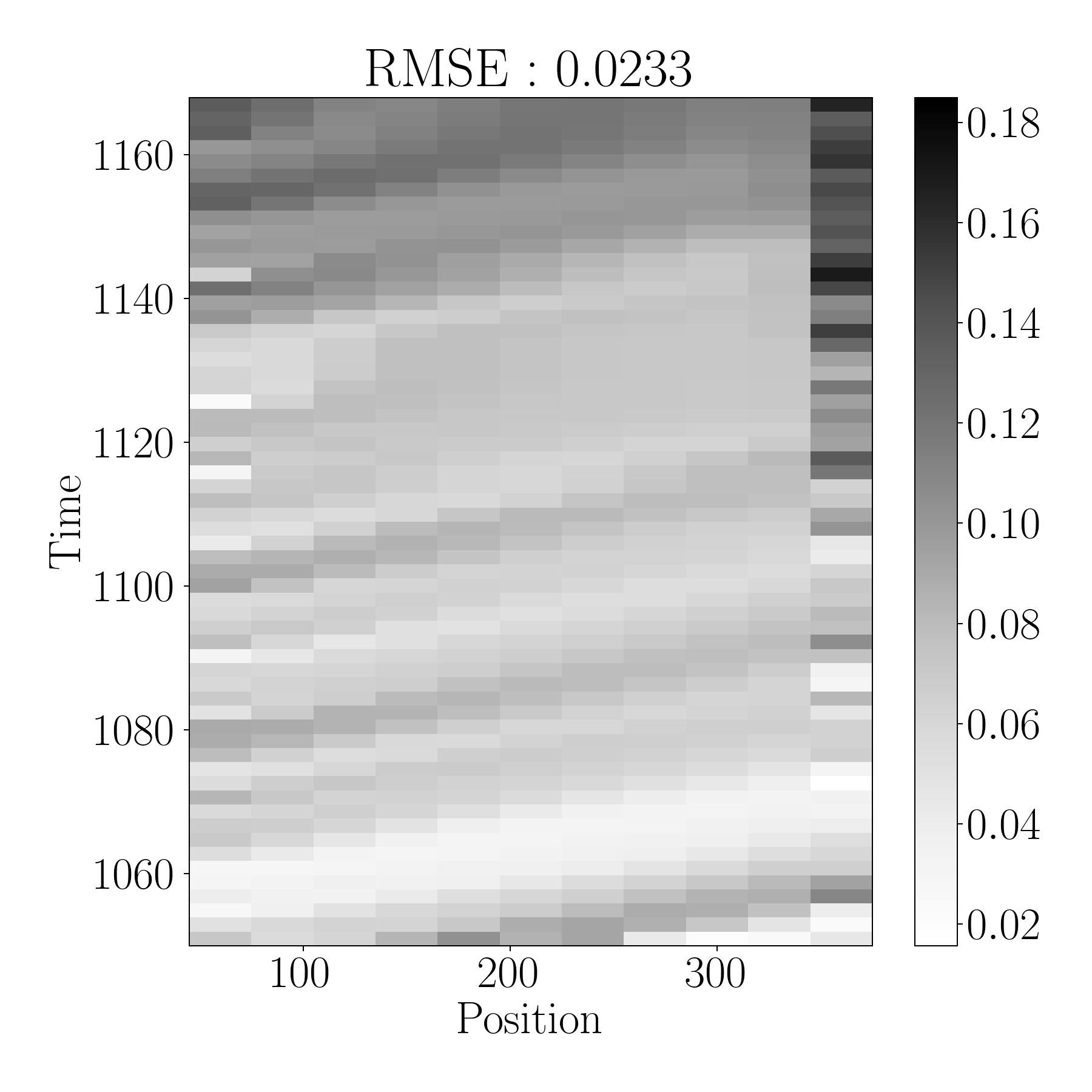}
\caption{Estimated densities for Dataset~2.}
\end{subfigure}
\end{minipage}%
\begin{minipage}{0.249\textwidth}
\begin{subfigure}[c]{\textwidth}
\includegraphics[width=\textwidth]{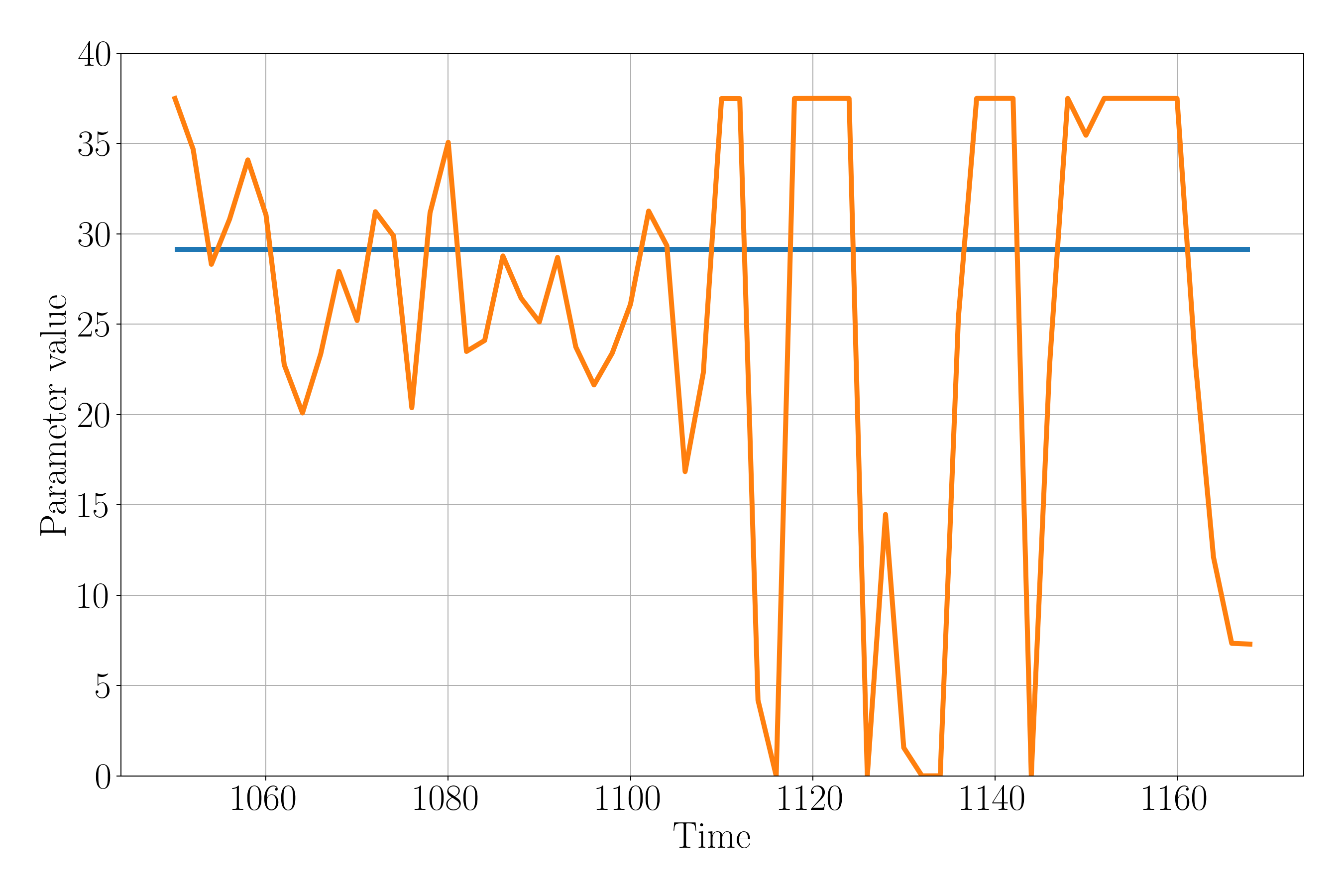}
\caption{Estimated parameters for Dataset~2.}
\end{subfigure}
\end{minipage}
\caption{Estimated densities and parameters for each dataset when time dependent parameters are considered. The plots (c) and (f) represent the evolution through time of the estimated parameter.The blue line represents the estimated parameter in the constant case.}
\label{fig:real_t}
\end{figure}

\begin{figure}
\begin{minipage}{0.24\textwidth}
\begin{subfigure}[c]{\textwidth}
\includegraphics[width=\textwidth]{fig/Dmat_Lin.png}
\caption{Dataset~1.}
\end{subfigure}
\begin{subfigure}[c]{\textwidth}
\includegraphics[width=\textwidth]{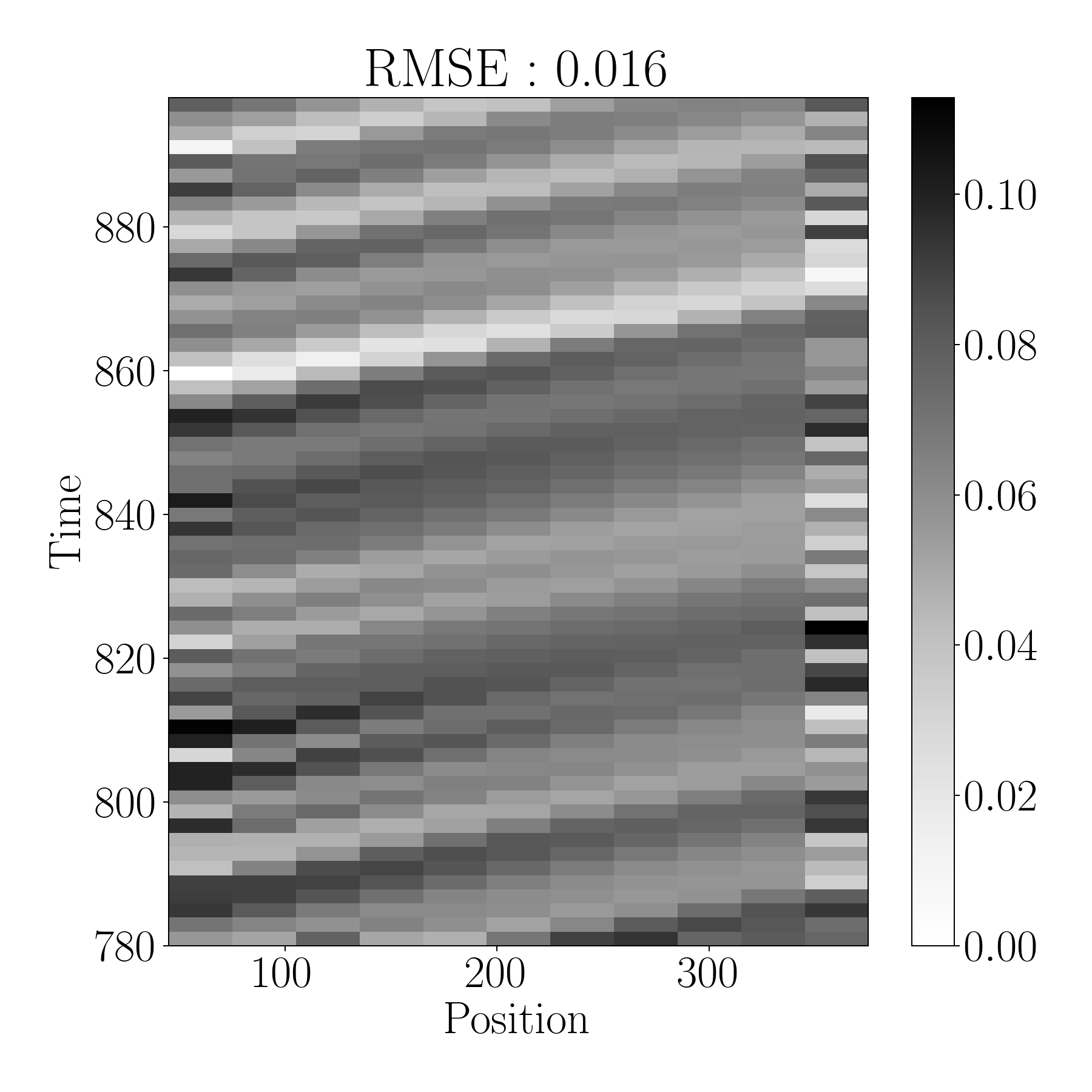}
\caption{Estimated densities for Dataset~1.}
\end{subfigure}
\end{minipage}%
\begin{minipage}{0.249\textwidth}
\begin{subfigure}[c]{\textwidth}
\includegraphics[width=\textwidth]{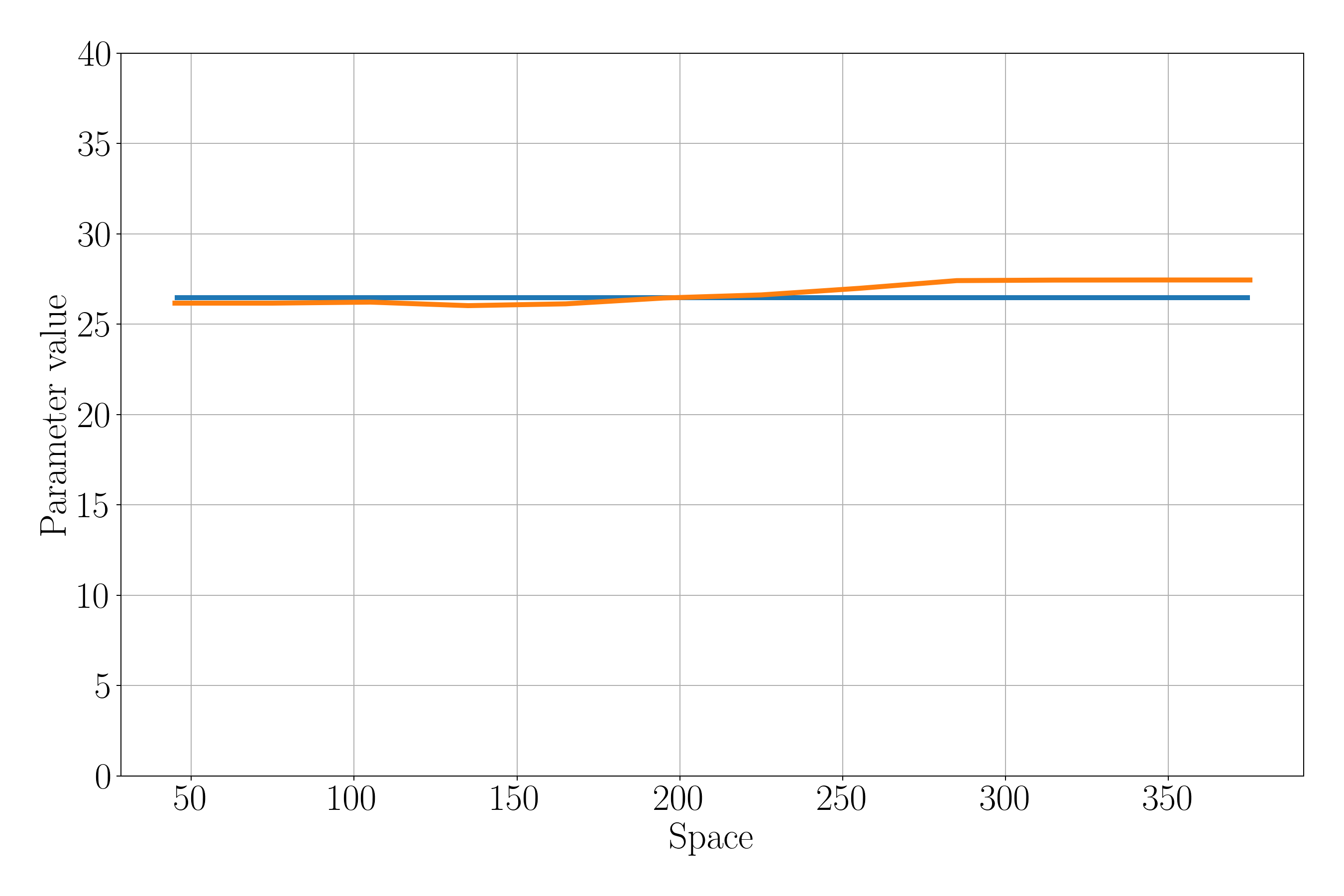}
\caption{Estimated parameters for Dataset~1.}
\end{subfigure}
\end{minipage}%
\hfill
\begin{minipage}{0.24\textwidth}
\begin{subfigure}[c]{\textwidth}
\includegraphics[width=\textwidth]{fig/Dmat_both.png}
\caption{Dataset~2.}
\end{subfigure}
\begin{subfigure}[c]{\textwidth}
\includegraphics[width=\textwidth]{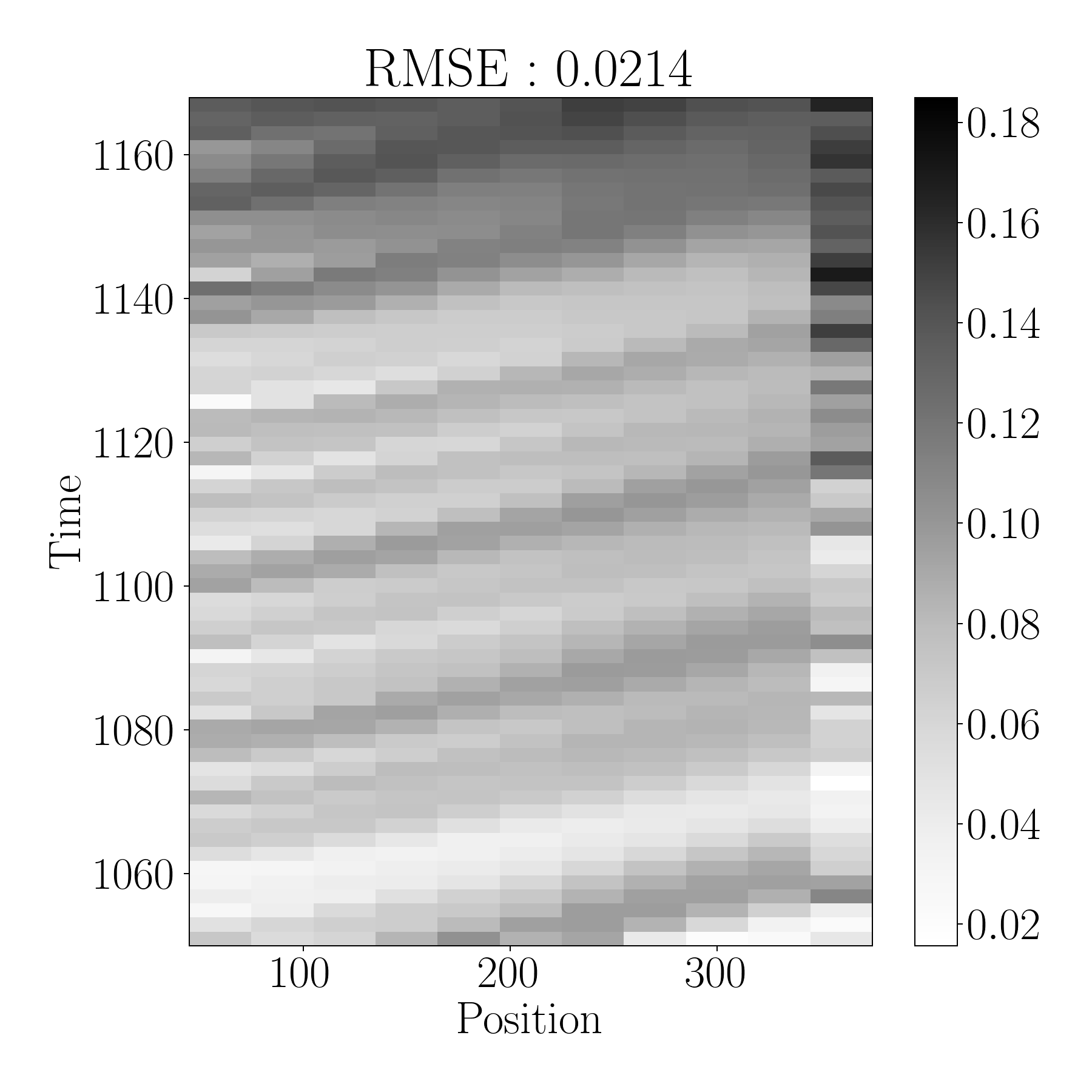}
\caption{Estimated densities for Dataset~2.}
\end{subfigure}
\end{minipage}%
\begin{minipage}{0.249\textwidth}
\begin{subfigure}[c]{\textwidth}
\includegraphics[width=\textwidth]{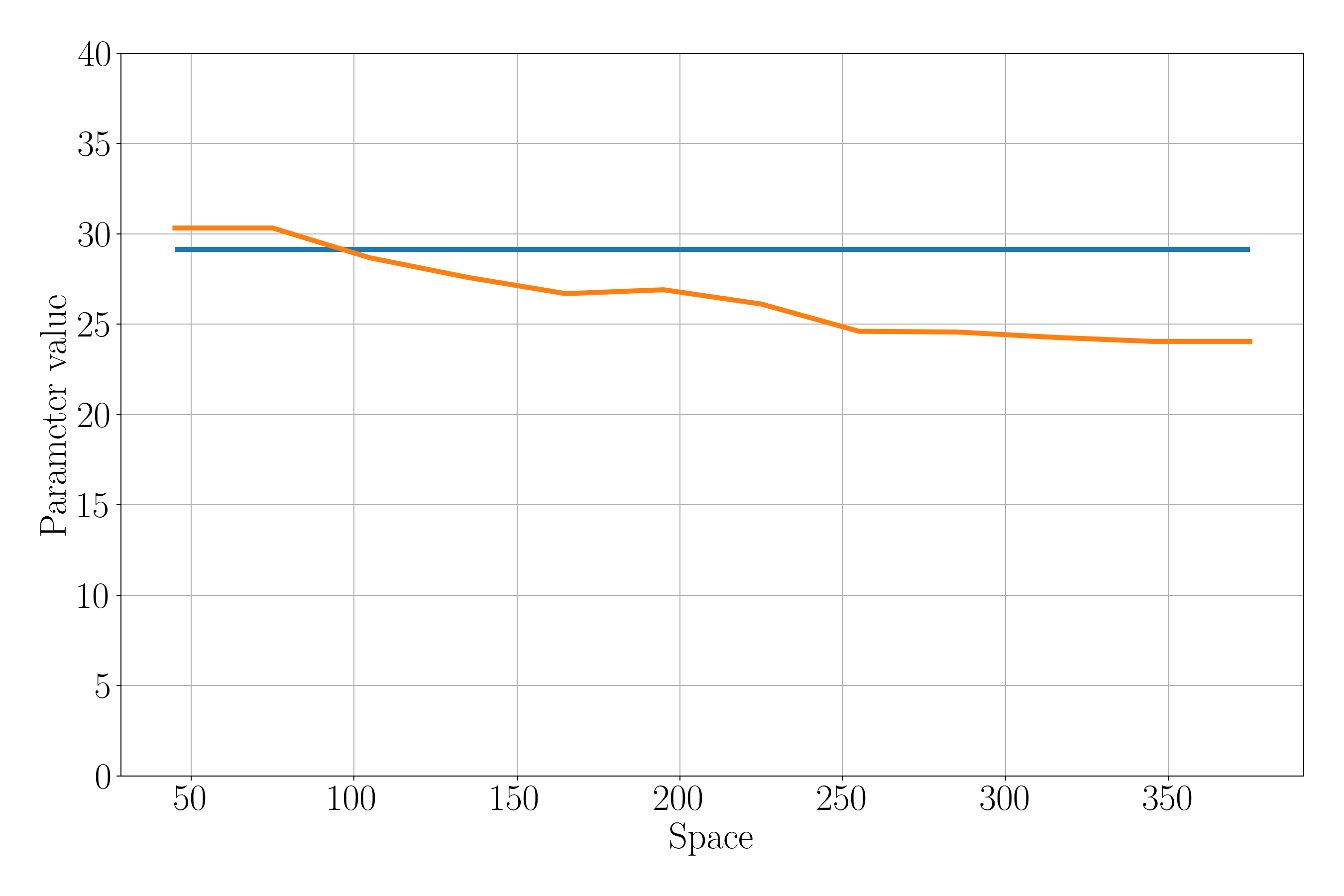}
\caption{Estimated parameters for Dataset~2.}
\end{subfigure}
\end{minipage}
\caption{Estimated densities and parameters for each dataset when space dependent parameters are considered. The plots (c) and (f) represent the evolution in space of the estimated parameter.The blue line represents the estimated parameter in the constant case.}
\label{fig:real_s}
\end{figure}

\begin{figure}
\begin{minipage}{0.24\textwidth}
\begin{subfigure}[c]{\textwidth}
\includegraphics[width=\textwidth]{fig/Dmat_Lin.png}
\caption{Dataset~1.}
\end{subfigure}
\begin{subfigure}[c]{\textwidth}
\includegraphics[width=\textwidth]{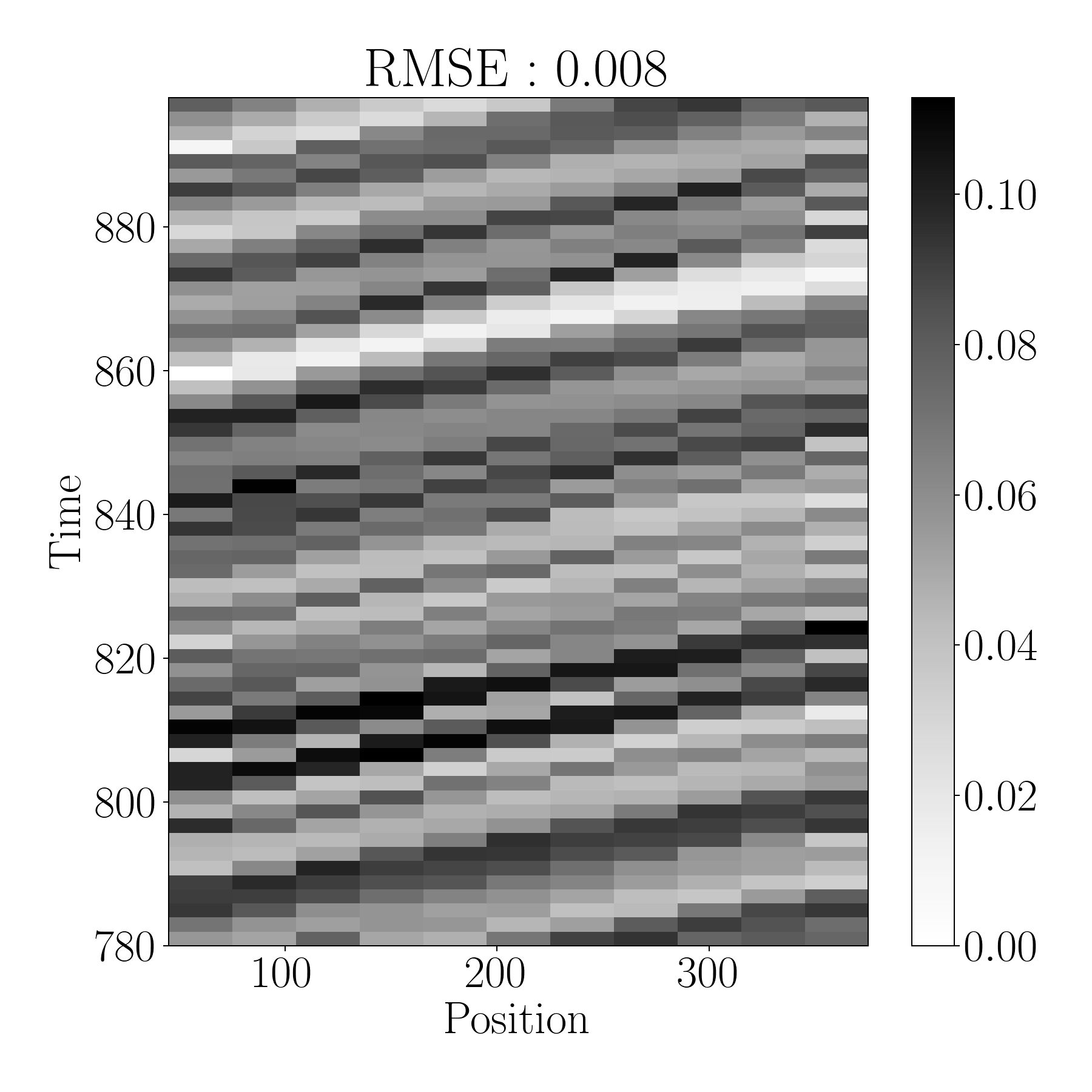}
\caption{Estimated densities for Dataset~1.}
\end{subfigure}
\end{minipage}%
\begin{minipage}{0.249\textwidth}
\begin{subfigure}[c]{\textwidth}
\centering
\foreach \n in {0,...,11}{
	\includegraphics[height=0.047\textheight]{fig/Dmat_Lin_flux_st_\n.png}%
	\vspace{-0.6ex}\\
}
\vspace{-1ex}
\includegraphics[width=\textwidth]{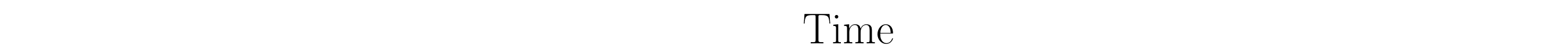}
\caption{Estimated parameters for Dataset~1.}
\end{subfigure}
\end{minipage}%
\hfill
\begin{minipage}{0.24\textwidth}
\begin{subfigure}[c]{\textwidth}
\includegraphics[width=\textwidth]{fig/Dmat_both.png}
\caption{Dataset~2.}
\end{subfigure}
\begin{subfigure}[c]{\textwidth}
\includegraphics[width=\textwidth]{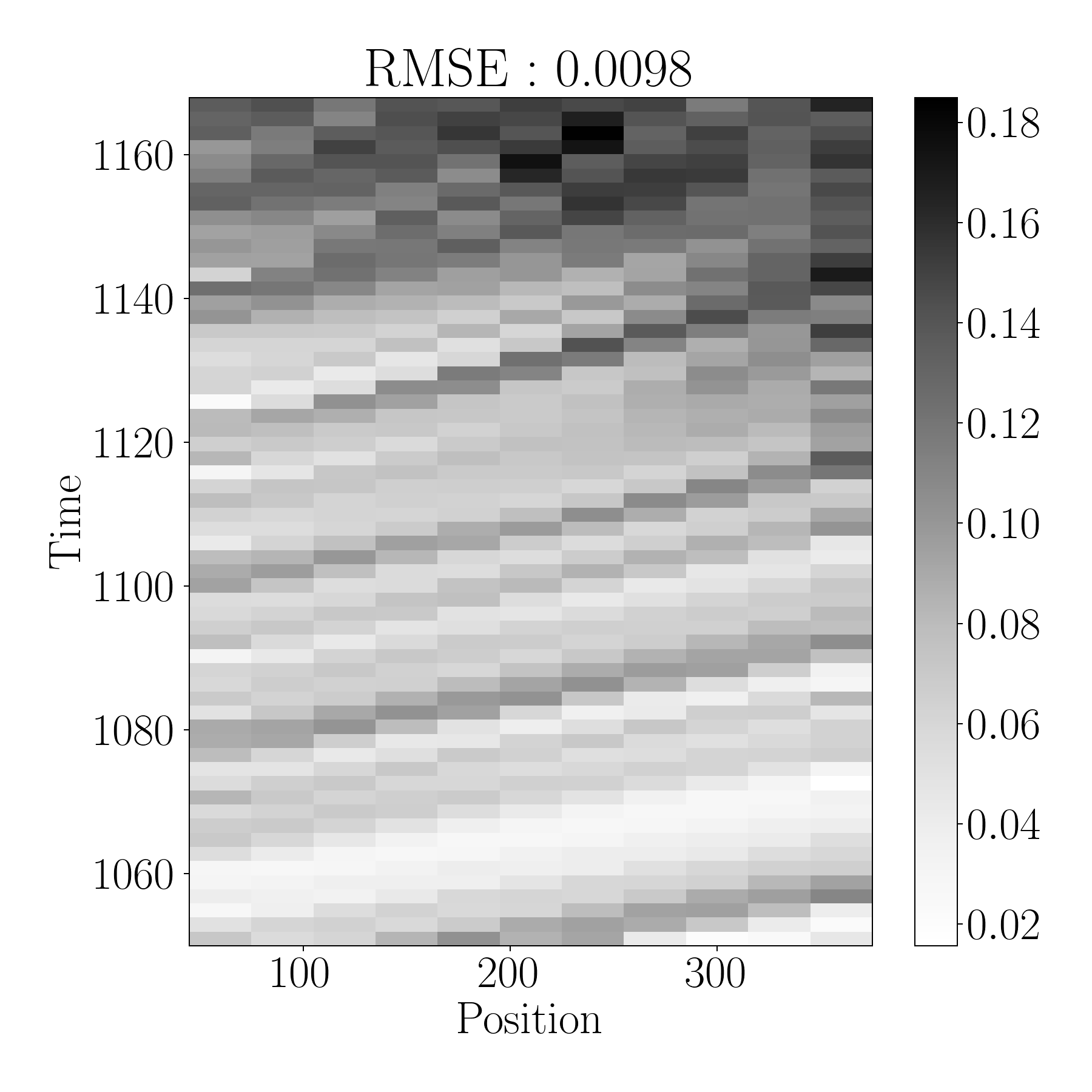}
\caption{Estimated densities for Dataset~2.}
\end{subfigure}
\end{minipage}%
\begin{minipage}{0.249\textwidth}
\begin{subfigure}[c]{\textwidth}
\centering
\foreach \n in {0,...,11}{
	\includegraphics[height=0.047\textheight]{fig/Dmat_both_flux_st_\n.png}%
	\vspace{-0.6ex}\\
}
\vspace{-1ex}
\includegraphics[width=\textwidth]{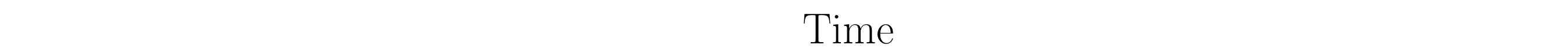}
\caption{Estimated parameters for Dataset~2.}
\end{subfigure}
\end{minipage}
\caption{Estimated densities and parameters for each dataset when space-time dependent parameters are considered. The plots (c) and (f) represent, for each location, the evolution through time of the estimated parameter. The blue line represents the estimated parameter in the constant case.}
\label{fig:real_st}
\end{figure}

Considering first the results for Dataset~1 (free flow conditions), we can see that the varying parameters estimated in each case stay close and vary around the value estimated under the assumption that the parameter is constant. This is coherent with the conclusions drawn in \Cref{sec:real_dat_ct}: in free flow conditions, the LWR model with a constant parameter  is an adequate choice of model. Comparing now the RMSE of the estimated densities, we see that the time-dependent and space-dependent parameters yield very similar values and density profiles as in the constant case (cf. \Cref{fig:real_ct_1}). However, a significant decrease of the RMSE is observed when a space-time-dependent parameter is considered. Hence, adding small perturbations of the parameters in space and time seems to yield more realistic density profiles and in particular the small scale variations of the density that are not observed in the constant case (due to the smoothness of the estimation). 

Considering now the results for Dataset~2 (free flow and congested conditions), we can see that the varying parameters estimated in each case do not stick around the value estimated in the constant case anymore. 

In the time-dependent case, we observe two regimes. The first regime spans until $t=1100\, \mathrm{s}$, and has the parameter varying around and close to the parameter estimated in the constant case, thus hinting at free flow conditions. The second regime starts at $t=1100\, \mathrm{s}$ and has the parameter displaying sharp variations between very low and very large values: such behavior can  be interpreted as the model trying to accommodate the congested conditions by intermittently  stopping or letting all the vehicles go in order to create congested cells. In the space-dependent case, the estimated densities globally decrease across space: this can be seen as an attempt from the model to create congestion by having cells with higher transfer rates upstream (which  will tend to let vehicles flow easily), and then gradually decreasing these rates as we go down the road, so that vehicles can accumulate downstream. In both these cases however, the resulting RMSE of the estimated densities is lower but still of the same order as the one from the constant case.

In the space-time-dependent case, the estimated parameters at all locations show the same trend: they start close to the value estimated in the constant case and after some time globally decrease with time. Besides, this drop in parameter value occurs at increasing times as we go from the right-most cell to the left-most cell. Hence, the model seems to account for congestion by gradually reducing the transfer rates between the compartments, going from right to left. In this case, the resulting RMSE of the estimated densities is significantly reduced compared to the constant case and the estimated densities display a realistic profile, which also recreates the congestion observed in the data.  

Note that, following the link established between the reaction rates and the parameters of the continuous traffic flow models, the gradual decrease of reaction rates observed for Dataset~2 in the time-dependent and space-time-dependent cases can be interpreted as a gradual drop in road capacity. This observation is corroborated by looking at the actual trajectories corresponding to this dataset and shown in \Cref{fig:traj_both}. Indeed, overtaking between vehicles can be observed from trajectory crossings. These overtakings mechanically decrease the overall capacity of the road as less lanes are free. As on can see, these overtaking happen more and more frequently as time passes, and start to appear downhill on the road. The same observations were made when looking at the space-time dependent reaction rates.

\begin{figure}
\centering
\includegraphics[width=0.8\textwidth]{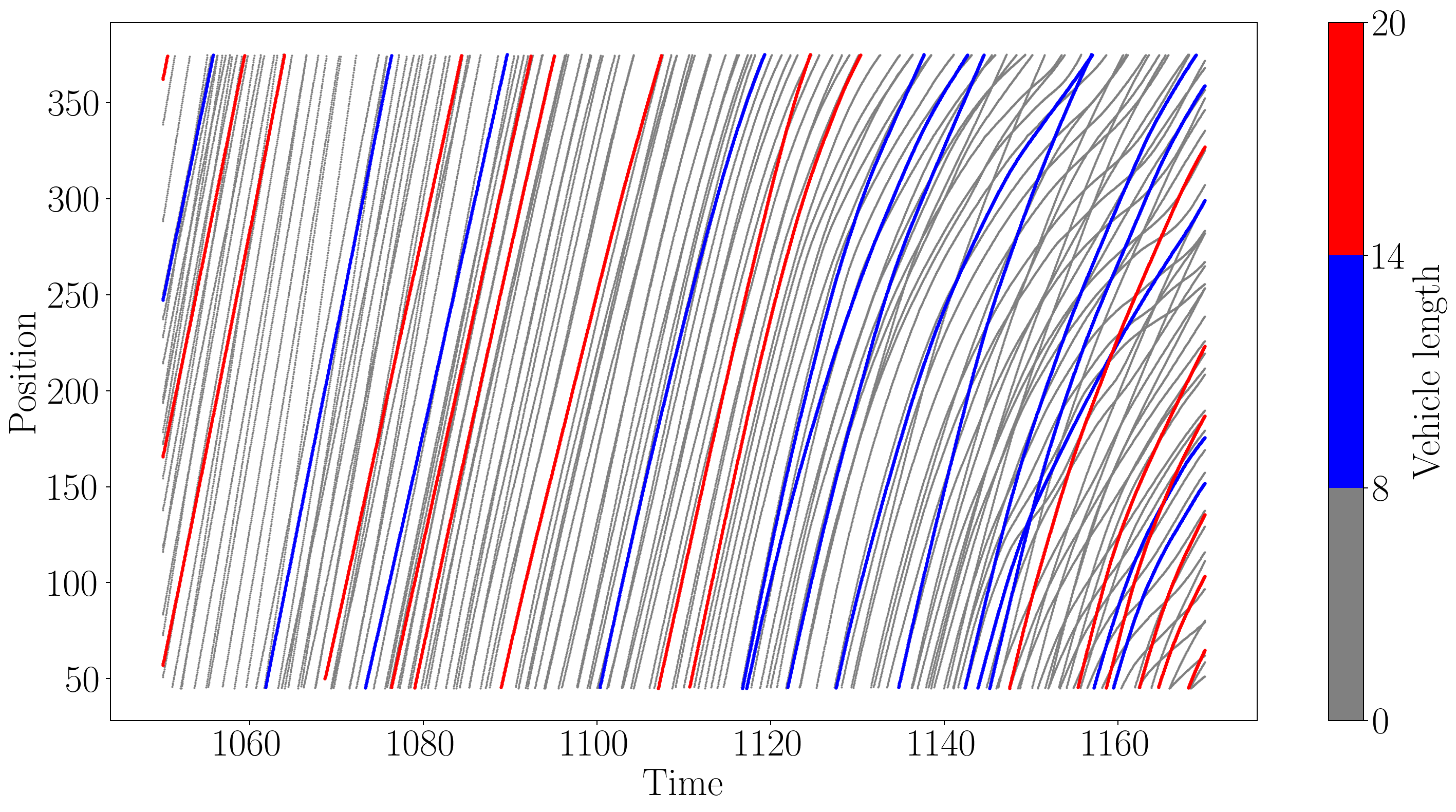}
\caption{Trajectories of Dataset~2. Each line corresponds to the trajectory of a given vehicle, the line color is linked to the length of the vehicle.\label{fig:traj_both}}
\end{figure}

Finally, we compare the three choices of parametrizations considered in this section in terms of their ability to recreate a fundamental diagram similar to the one associated with the density data. In particular, flux estimates can be derived from the density estimates by once again applying the quadratic flux-density relation~\eqref{eq:f}, but using now the varying parameter $v_m$: to compute the flux estimate of the $j$-th cell at time $t_n$, $v_m$ is taken as the average of the parameter estimates at both boundaries of the $j$-th cell, at time $t_n$. We obtain the fundamental diagrams shown in \Cref{fig:fd_var}. The time-dependent estimates show for both datasets a fundamental diagram which is more scattered than the one from the data, and yield higher RMSE than in the constant case. On the other end, the space-dependent estimates produce fundamental diagrams that are similar to the one obtained in the constant case, but with slightly more dispersion. In both cases however, the fundamental diagrams consist in superposition of quadratic functions and seem to fail to reproduce the scattering observed in Dataset~2 (and due to congested conditions). This goal is however achieved with the space-time-dependent estimates which yield fundamental diagrams that nicely overlap the ones from the data, and significantly lower RMSE compared to the constant case.

\begin{figure}
\centering
\begin{subfigure}{0.4\textwidth}
\includegraphics[width=\textwidth]{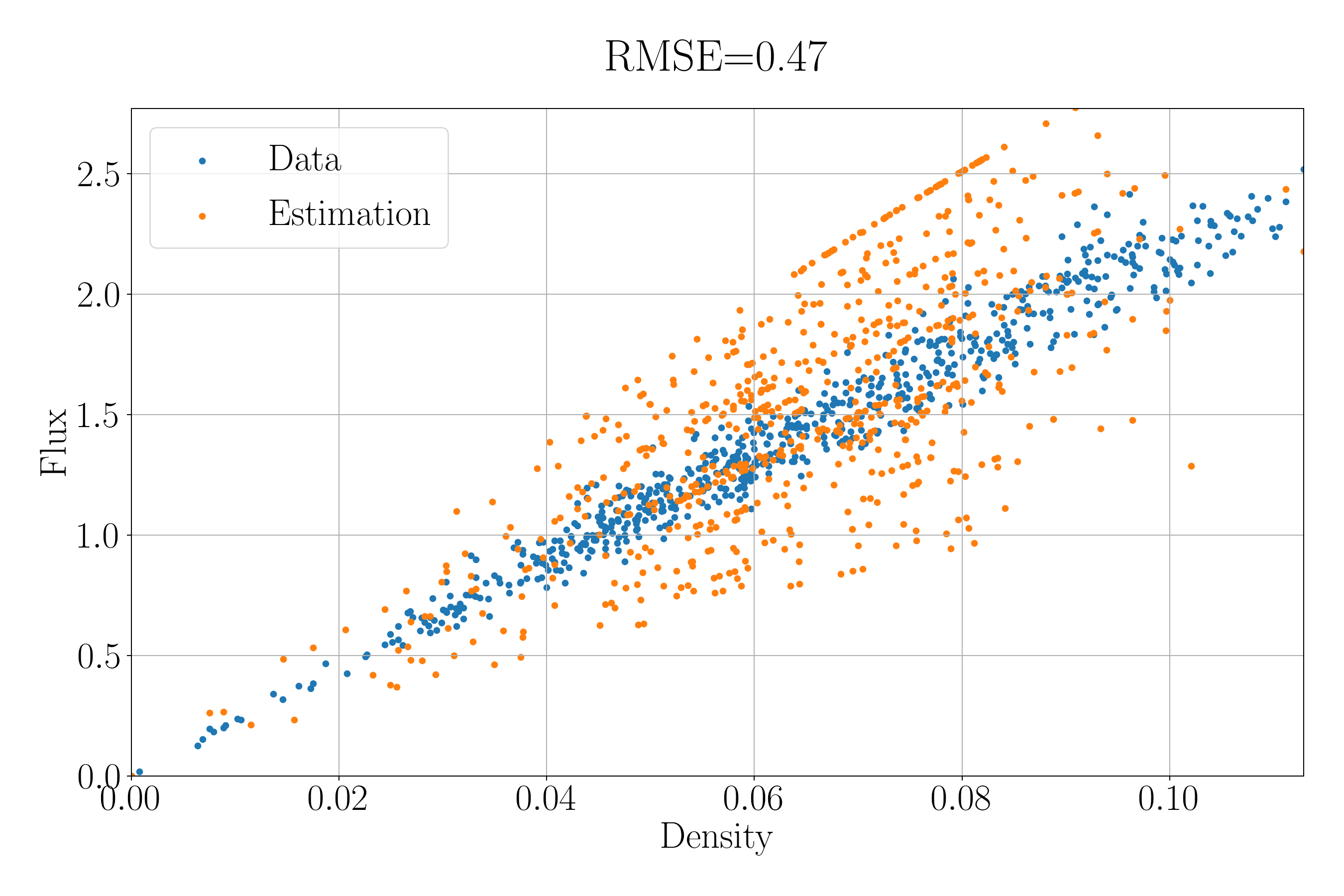}
\caption{Dataset~1: Time-dependent.}
\end{subfigure}%
\hspace{6ex}%
\begin{subfigure}{0.4\textwidth}
\includegraphics[width=\textwidth]{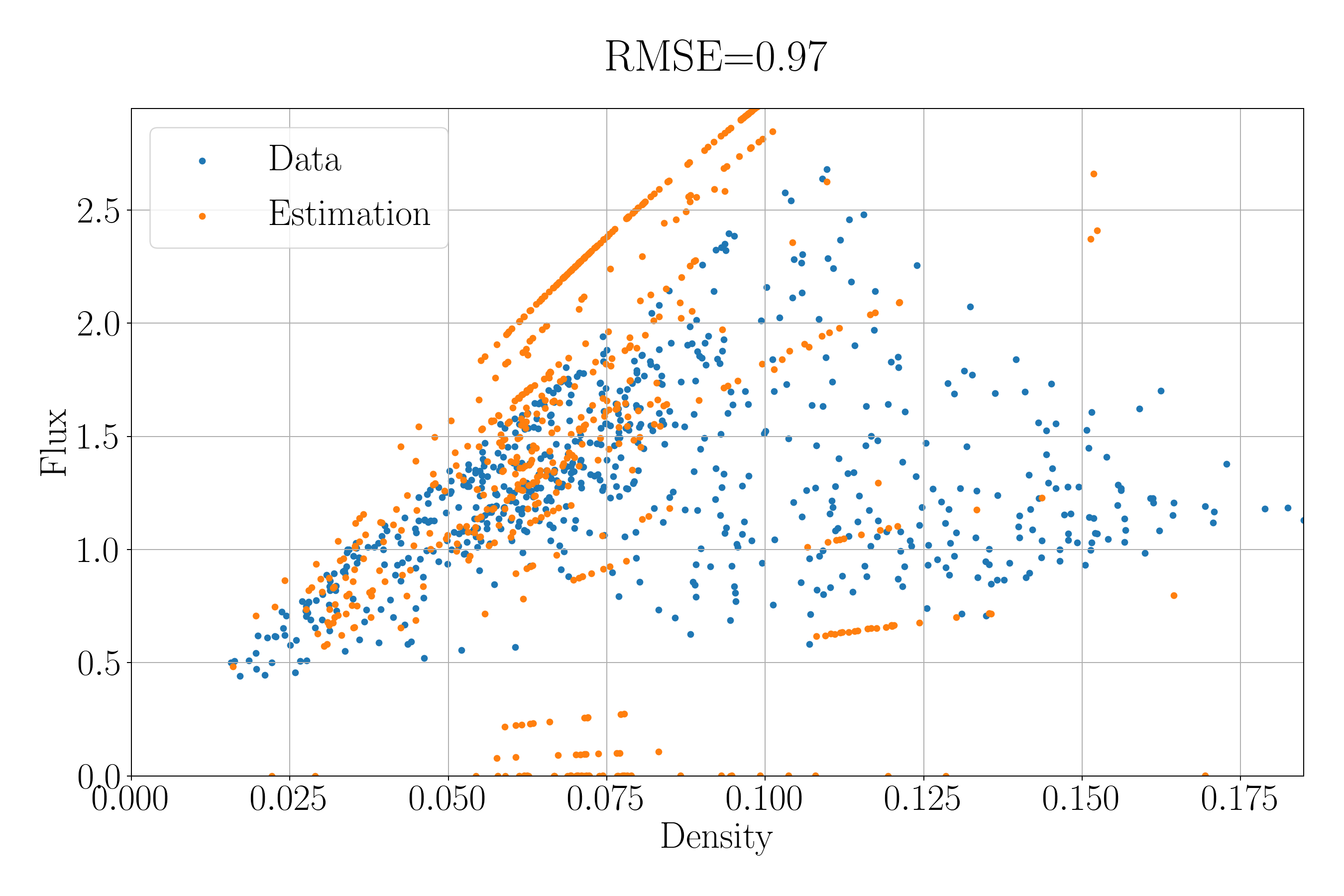}
\caption{Dataset~2: Time-dependent.}
\end{subfigure}%
 \\
\begin{subfigure}{0.4\textwidth}
\includegraphics[width=\textwidth]{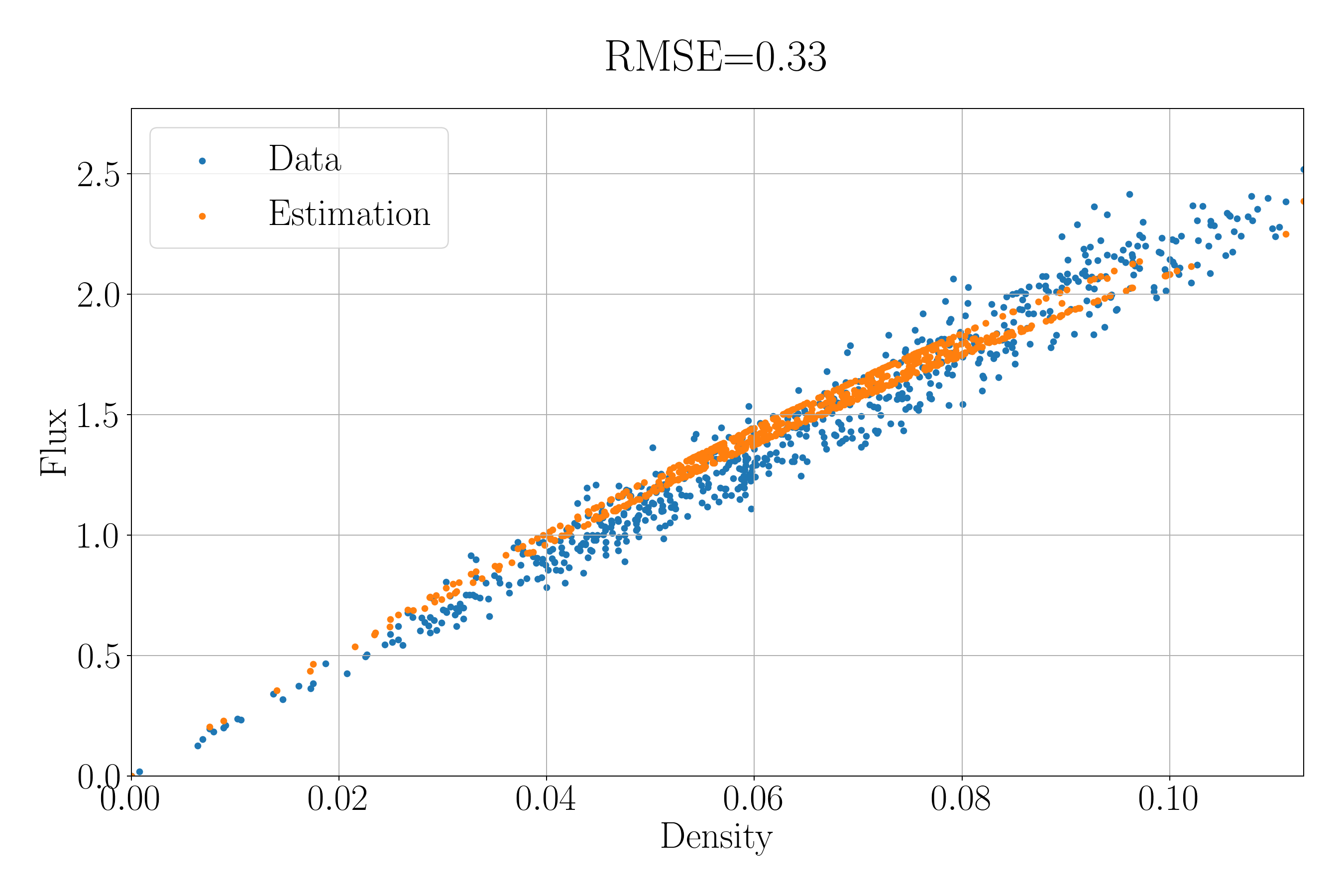}
\caption{Dataset~1: Space-dependent.}
\end{subfigure}%
\hspace{6ex}%
\begin{subfigure}{0.4\textwidth}
\includegraphics[width=\textwidth]{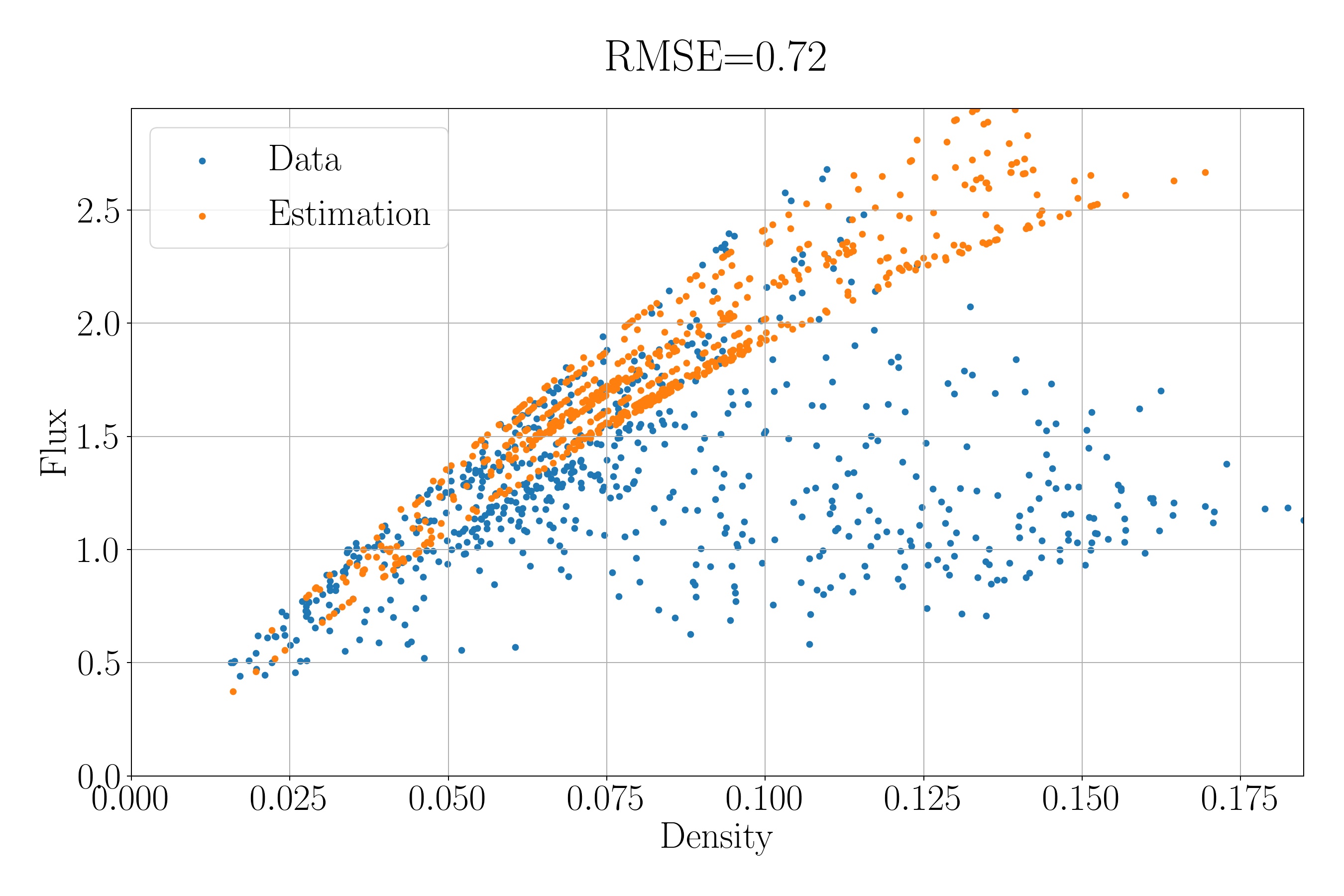}
\caption{Dataset~2: Space-dependent.}
\end{subfigure}%
\\
\begin{subfigure}{0.4\textwidth}
\includegraphics[width=\textwidth]{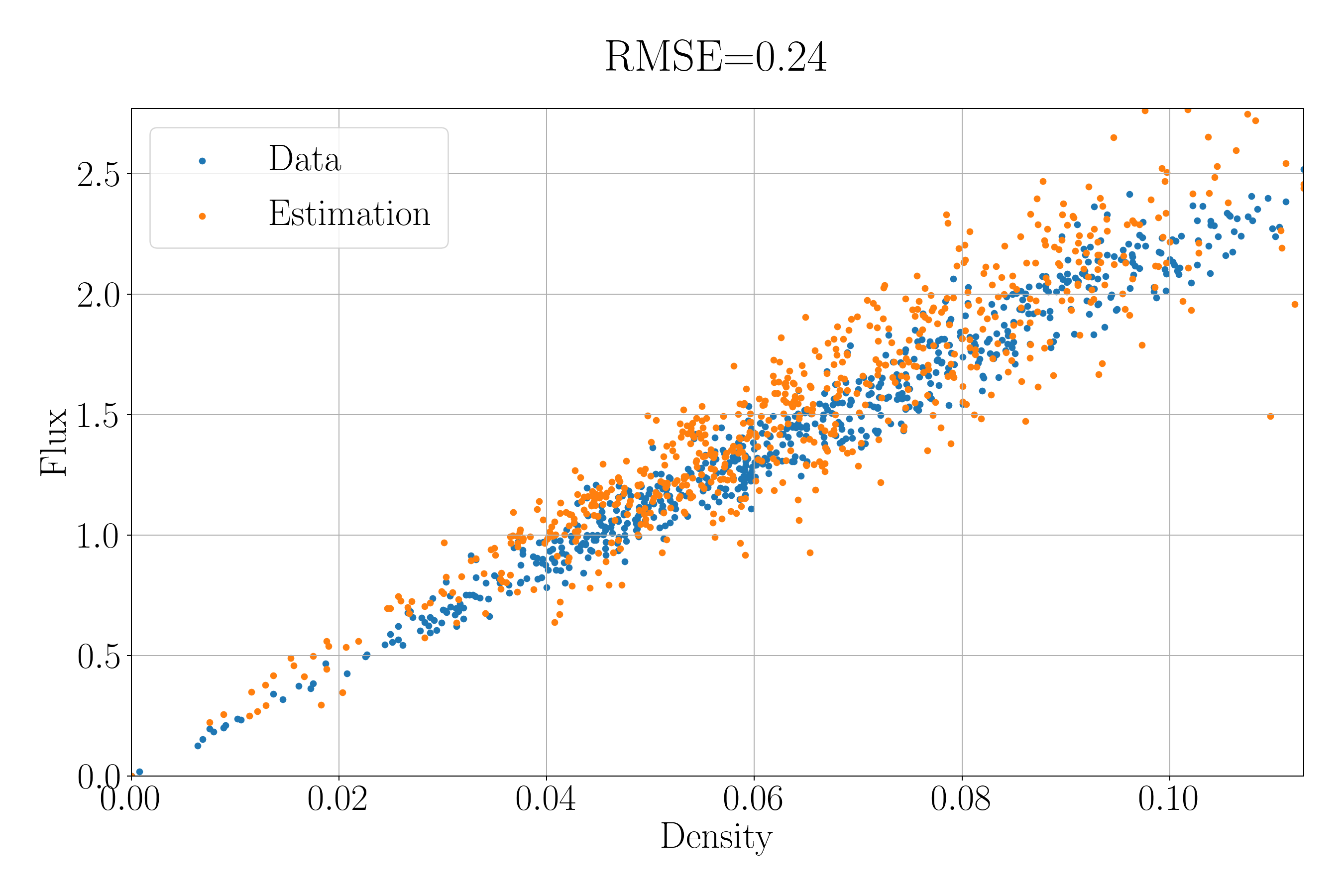}
\caption{Dataset~1: Space-time-dependent.}
\end{subfigure}%
\hspace{6ex}%
\begin{subfigure}{0.4\textwidth}
\includegraphics[width=\textwidth]{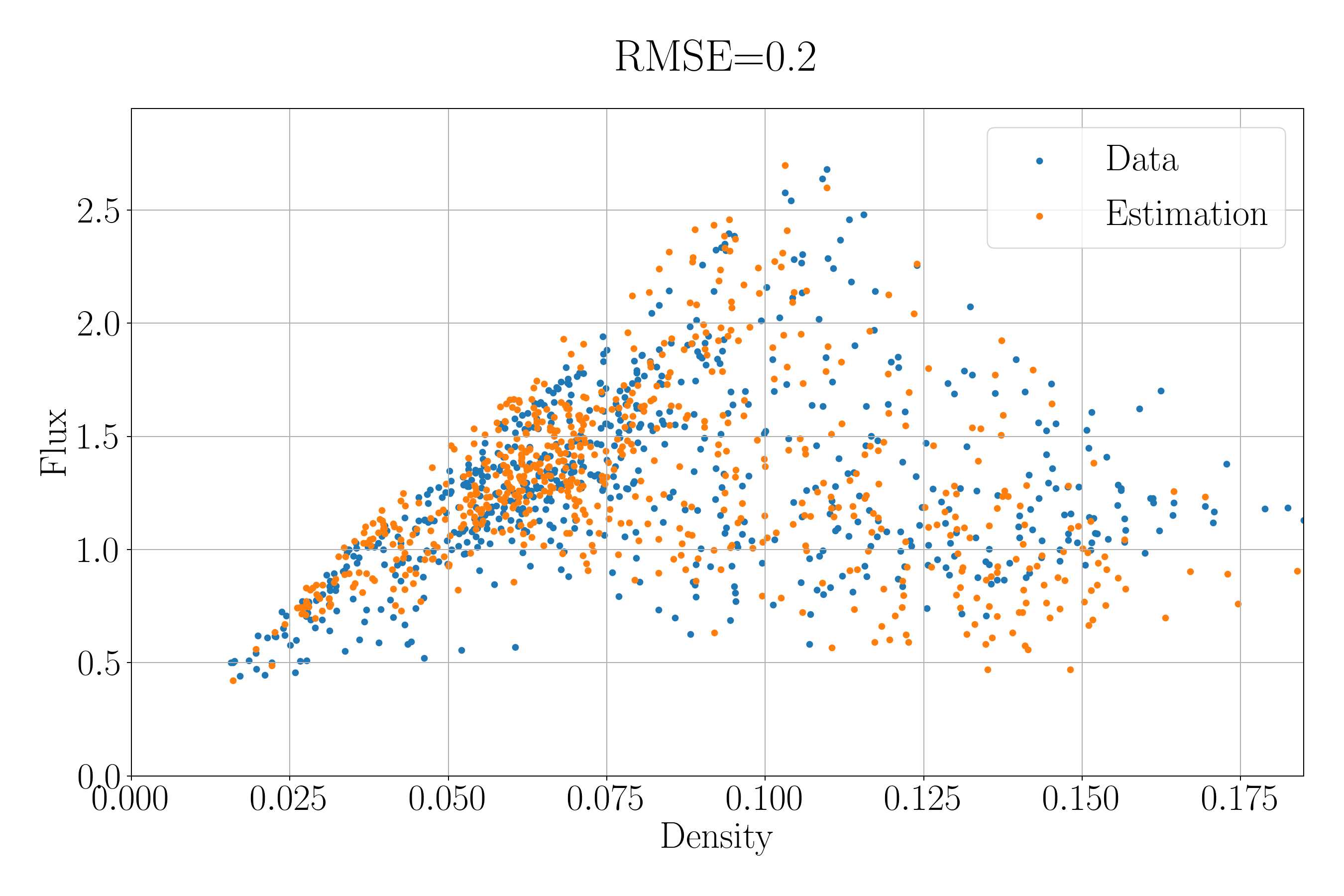}
\caption{Dataset~2: Space-time-dependent.}
\end{subfigure}%
\caption{Comparison between the fundamental diagram of the two datasets and their estimations through the TRM with time-dependent, space-dependent and space-time dependent parameters.}
\label{fig:fd_var}
\end{figure}

In conclusion, the recourse to space-time dependent parameters provides more flexibility to the LWR model in a physically sound manner, thus allowing it to recreate real-word density and fundamental diagrams:
\begin{itemize}
\item On the one hand, allowing the reaction rates between compartments in the TRM to vary in space and time locally creates conditions that give rise to congestion or sharp changes in the density.
\item On the other hand, realistic fundamental diagrams are obtained even though a quadratic relation between flux and density is assumed, by allowing the shape of the relation to change over space and time. Hence, the change of behavior in the fundamental diagram usually interpreted as a capacity drop now becomes a transfer rate drop. Besides, complex point patterns in the diagram can be recreated since in theory each point of the diagram belongs to its own quadratic function.
\end{itemize}


\section{Conclusion}

The main motivation of this work is to assess the validity of a LWR traffic flow model to model measurements obtained from trajectory data, and propose extensions of this model to improve it. We answer these questions by comparing continuous models and measurements using a discrete dynamical system defined from a particular discretization of the PDE of the continuous model. This discretization is formulated as a chemical reaction network where road cells are interpreted as compartments, the transfer of vehicles from one cell to the other is seen as a chemical reaction between adjacent compartment and the density of vehicles is seen as a concentration of reactant. Several degrees of flexibility on the parameters of this system, which basically consist of the reaction rates between the compartments, are considered: These rates are taken equal to the same constant value or allowed to depend on time and/or space. We then interpret generalized density measurements coming from trajectory data as observations of the states of the discrete dynamical system at consecutive times, and derive optimal reaction rates for the system by minimizing the discrepancy between the output of the system and the state measurements. 

The use of constant reaction rates proves to be enough to reproduce the patterns observed in the density and flux data in free flow conditions but not in mixed conditions where congestion appears. This motivates us to recommend the use of the more flexible models, and in particular the model with space-time dependent reaction rates. This last model proved to perform well both in free flow and mixed conditions as it mimicked the patterns observed in the density data as well as the fundamental diagrams.  
Recall that the discrete dynamical system can be seen as a particular finite volume discretization of the LWR model with the flux of vehicles depending quadratically on the density. The reaction rates of the system then simply set the shape of this relation (meaning here the maximal value of the flux function). Our numerical experiments hence showed that allowing the shape of this quadratic relation to vary through time, space or even better both, allowed the LWR model to better recreate specific patterns observed in real-world data, such as the appearance of congestion (compared to when a fixed shape is considered). 

Direct extensions of the approach presented in this paper are possible. On the one hand, working on networks would be straightforward since the proposed kinetic system can be generalized to this setting by simply dropping the assumptions that the compartments are ordered as chain (which makes sense for a single road)  and allowing them to be linked to more than 2 other compartments (thus mimicking the junctions of the network). On the other hand, the use of conventional detector data could be considered, since it would simply come down to the assumption that measurements are only available in some compartments (those where sensors are located), similarly as what was assumed in the robustness tests done in \Cref{sec:ident}. Finally, a link between the proposed discrete dynamical system and artificial neural networks (ANN) was not
exploited in this paper but paves the way to exciting outlooks


\bibliographystyle{abbrvnat}
\bibliography{biblio_clean}
\addcontentsline{toc}{section}{References}

\appendix

\begin{center}
\begin{LARGE}
\textbf{APPENDIX}

\end{LARGE}
\end{center}

\section{Some examples of finite volume schemes}
\label{sec:fvs}

The Lax--Friedrich (LxF) scheme is defined for the choice of numerical flux $F=F_L$ with
\begin{equation*}
F_L(u,v \pv v_m)=\frac{f(u\pv v_m)+f(v\pv v_m)}{2}+\frac{\Delta x}{2\Delta t}(u-v) \peq
\end{equation*}
The Godunov scheme is given by the choice $F=F_G$ with 
\begin{equation*}
F_G(u,v \pv v_m)=\begin{cases}
\min\limits_{w\in[u,v]} f(w\pv v_m) & \text{if } u\le v \\
\max\limits_{w\in[v,u]} f(w\pv v_m) & \text{if } v\le u
\end{cases} \peq
\end{equation*}

In the particular case where $f$ is defined by~\eqref{eq:f}, note that the recurrence relation of the Godunov scheme can be rewritten as 
\begin{equation}
\widehat{U}_j^{i+1}=\widehat{U}_{j}^i+\frac{\Delta t}{\Delta x}v_m\left[\tilde F_G(\widehat{U}_{j-1}^i, \widehat{U}_j^i)-\tilde F_G(\widehat{U}_{j}^i, \widehat{U}_{j+1}^i) \right], \quad j\in\mathbb{Z}, i\in\mathbb{N}
\veq
\label{eq:rec_fv_god}
\end{equation}
where $\tilde{F}_G$ is a normalized numerical flux (in the sense that it does not depend on the parameter $v_m$ anymore) given by
\begin{equation*}
\tilde F_G(u,v)=
\begin{cases}
\min\limits_{w\in[u,v]} w(1-w) & \text{if } u\le v \\
\max\limits_{w\in[v,u]} w(1-w) & \text{if } v\le u
\end{cases} \peq
\end{equation*}
Similarly, for the Lax--Friedrichs scheme, we can write
\begin{equation}
\widehat{U}_j^{i+1}=\frac{\widehat{U}_{j-1}^i+\widehat{U}_{j+1}^i}{2}+\frac{\Delta t}{\Delta x}v_m\left[\tilde F_L(\widehat{U}_{j-1}^i, \widehat{U}_j^i)-\tilde F_L(\widehat{U}_{j}^i, \widehat{U}_{j+1}^i ) \right], \quad j\in\mathbb{Z}, i\in\mathbb{N}
\label{eq:rec_fv_lxf}
\end{equation}
for the normalized numerical flux $\tilde F_L$  defined by
\begin{equation*}
\tilde F_L(u,v )= \frac{u(1-u)+v(1-v)}{2} \peq
\end{equation*}

\section{Minimization problems in the multilevel approach}
\label{sec:adpt_pbm}

On the one hand, in the constant parameter case, the minimization problem can be reformulated as 
\begin{equation*}
\theta^*=\argmin_{\theta \in\R} L_{[P_t,P_x]}(\theta) \veq
\end{equation*}
where the cost function $L_{[P_t,P_x]}$ is now defined by 
\begin{equation}
L_{[P_t,P_x]}(\theta)=\frac{1}{2}\sum_{i=1}^{N_t-1}\sum_{j=1}^{N_x-2} \left( \frac{1}{P_x}\sum_{k=0}^{P_x-1}\widehat{U}_{k+jP_x}^{iP_t}(C(\theta)) - U_j^i\right)^2, \quad \theta \in\R \veq
\label{eq:costf_b}
\end{equation}
with the same mapping $C : \R \rightarrow (0,1/2)$ defined by~\eqref{eq:corresp}. The optimal value $v_m^*$ of the parameter of PDE~\eqref{eq:pde_u} is obtained by
\begin{equation*}
v_m^*=\frac{\widehat{\Delta x}}{\widehat{\Delta t}}C(\theta^*)=\frac{P_t}{P_x}\frac{{\Delta x}}{{\Delta t}}C(\theta^*) \peq
\end{equation*}

On the other hand, in the varying parameter case, we adopt the following changes:
\begin{itemize}
\item the boundary and initial conditions are set in the same way;
\item the recurrence relation of the scheme, now defined on the subdivided grid, takes the form
\begin{equation*}
\widehat{U}_k^{m+1} =\widehat{U}_k^m+  \widehat C_k^m\widehat{U}_{k-1}^m\left(1-\widehat{U}_k^m\right)- \widehat C_{k+1}^m\widehat{U}_{k}^m\left(1-\widehat{U}_{k+1}^m\right), \quad k\in\bi P_x,\; (N_x-1)P_x-1\ei, \quad m \in \bi 0,\; P_t(N_t-1)\ei \veq
\end{equation*}
where the coefficients $\lbrace \widehat C_k^m : k\in \bi P_x, (N_x-1)P_x\ei, m \in\bi 0, P_t(N_t-1)\ei\rbrace$ are defined through a bilinear interpolation (in space and time) of the coefficients  $\bm C=\lbrace C_j^n : j\in \bi 0, N_x\ei, n \in\bi 0, N_t-1\ei\rbrace$ defined in the case where no subdivision is introduced. In particular, for $j\in \bi 0, N_x-1\ei$, $n \in\bi 0, N_t-2\ei$, we have:
\begin{equation*}
\widehat C_{q+jP_x}^{l+n P_t}=\begin{pmatrix}
1-l/P_t \\
l/P_t
\end{pmatrix}^T
\begin{pmatrix}
C_j^n & C_{j+1}^n \\
C_j^{n+1} & C_{j+1}^{n+1}
\end{pmatrix}
\begin{pmatrix}
1-q/P_x \\
q/P_x
\end{pmatrix}, \quad q\in\bi 0, P_x-1\ei, \quad l\in\bi 0, P_t\ei  \; ;
\end{equation*}
\item the coefficients $\bm C=\lbrace C_j^n : j\in \bi 0, N_x\ei, n \in\bi 0, N_t-1\ei\rbrace$ are determined by minimizing (without constraints)   a cost function $L_{[P_t,P_x]}$ given as the sum of a least-square cost and a regularization term $R(\bm C)$:
\begin{equation*}
L_{[P_t,P_x]}(\bm \theta)= 
\frac{1}{2}\sum_{i=1}^{N_t-1}\sum_{j_c\in I_c} \left( \frac{1}{P_x}\sum_{k=0}^{P_x-1}\widehat{U}_{k+j_cP_x}^{iP_t}(\bm C(\bm\theta)) - U_{j_c}^i\right)^2 
+\lambda R(\bm C(\bm\theta)), \quad \bm \theta \in \R^{(N_x+1)N_t} \veq
\end{equation*}
where  $\bm C(\bm\theta)\in (0,1/2)^{(N_x+1)N_t}$ is obtained by applying \eqref{eq:corresp} to each entry of $\bm \theta \in \R^{(N_x+1)N_t}$, $I_c$ denotes the set of observed columns in the density matrix (excluding the boundary columns) and $\lambda>0$ is a hyperparameter weighting the least-square minimization of the regularization. 
\end{itemize}

In both cases, the minimization can once gain be tackled using the conjugate gradient algorithm, since the same rules can be applied to derive an explicit expression of the gradient of the cost function (cf. \Cref{sec:grad_comp}).

\section{Gradient computation for cost minimization}
\label{sec:grad_comp}

Let $L$ denote the least-square cost function defined by~\eqref{eq:costf}.
Starting then from the recurrence relation~\eqref{eq:rec}, we see that any $i\in\bi 1, N_t\ei$, $j\in\bi 1, N_x-2\ei$, the finite volume approximation $\widehat{U}_j^i(C)$ can be expressed as a composition of the functions $\mathcal{H}_l^k$, for $k<i$ and $l\in \bi 1, N_x-2\ei$. Assuming that the functions $h$ and $\tilde{F}$ are smooth with respect to their arguments, the gradient $\nabla L$ of $L$ with respect to $\theta$ can actually be computed using the chain rule of derivation, giving the expression given in the next proposition.
\begin{prop}
Let $L$ be the cost function defined in~\eqref{eq:costf} and assume that the mappings $\mathcal{H}^i$, $i\ge 0$, defined through~\eqref{eq:transfo} and~\eqref{eq:rec}  are smooth with respect to their arguments. Then, the gradient of $L$ is given by
\begin{equation}
\frac{\partial L}{\partial \theta}(\theta)
=\frac{\partial C}{\partial \theta}(\theta)\cdot\sum_{i=1}^{N_t-1} \J_{\widehat{\bm U}^i}( C(\theta))^T \left(\widehat{\bm U}^i(C(\theta))-\bm U^i \right), \quad \theta \in\R \veq
\label{eq:gradl}
\end{equation}
where 
\begin{equation*}
\frac{\partial C}{\partial \theta}(\theta)=\frac{1}{2}\lgt(\theta)(1-\lgt(\theta)), \quad \theta \in\R \veq
\end{equation*}
and for $i\in\bi 1, N_t-1\ei$, and $ C\in (0,1/2)$, $\J_{\widehat{\bm U}^i}( C)\in\R^{N_x\times N_C}$ denotes the Jacobian matrix of the mapping $ C \mapsto \widehat{\bm U}^i(C)$, which can be computed through the recurrence relation
\begin{equation}
\left\lbrace\begin{aligned}
&\J_{\widehat{\bm U}^0}(C) = \bm 0 \\
&\J_{\widehat{\bm U}^{i+1}}(C)=
 \J_{\mathcal{H}^{i}}(\widehat{\bm U}^i)\cdot\J_{\widehat{\bm U}^i}(C)
 + \J_{\mathcal{H}^{i}}(C)
 , \quad i\ge 0
\end{aligned}\right.
\label{eq:rec_jac}
\end{equation}
with $\J_{\mathcal{H}^{i}}(\widehat{\bm U}^i)\in\R^{N_x\times N_x}$  being the Jacobian matrix of the mapping $\widehat{\bm U}^i \mapsto \mathcal{H}^i(\widehat{\bm U}^i, C \pv \bm U)$ and $ \J_{\mathcal{H}^{i}}(C)\in\R^{N_x\times N_C}$ being the Jacobian matrix of the mapping  $C\mapsto \mathcal{H}^i(\widehat{\bm U}^i, C \pv \bm U)$.
\label{prop:grad_expr}
\end{prop}
\begin{proof}
Applying the chain rule to~\eqref{eq:costf} yields exactly~\eqref{eq:gradl}. Then, applying the chain rule to~\eqref{eq:rec} gives the recurrence relation in~\eqref{eq:rec_jac}. The fact that $\J_{\bm U^0}(C) = \bm 0$ follows from the fact that $\bm U^0$ does not depend on $C$ (but is defined in~\eqref{eq:ic} using the data $\bm U$).
\end{proof}

The explicit expression of the Jacobian matrices $ \J_{\mathcal{H}^{i}}(\widehat{\bm U}^i)$ and $\J_{\mathcal{H}^{i}}(\bm C)$ in \Cref{prop:grad_expr} depends on the choice of the numerical scheme to compute the approximations in $\widehat{\bm U}$. This scheme should only involve smooth functions as assumed at the beginning of this section. This is in particular the case for the TRM and Lax--Friedrichs scheme, and the corresponding Jacobian matrices are given in \Cref{sec:jac}.

Using these expressions, \Cref{alg:grad_fp} provides a first way to compute the gradient vector~\eqref{eq:gradl}. This algorithm can be referred to as a Forward-Propagation algorithm: we visit each \q{time} $i$ sequentially from $0$ to the final time to compute the gradient. The finite volume approximations $\widehat{\bm U}^1, \dots, \widehat{\bm U}^{N_t-1}$ are computed on-the-fly, thus saving some storage space. On the other hand, note that each iteration requires matrix-matrix multiplications. Even though the Jacobian matrices involved in these products are sparse, the stored matrix $\bm G$ will fill up as $i$ grows, rendering the computational and storage costs of each iteration more and more expensive. This could become cumbersome in some applications where the size of this matrix, which is $N_t\times N_x$, is large. 

Inspired by the theory around the fitting of neural networks we propose a Back-Propagation algorithm which allows us to compute this same gradient while only requiring matrix-vector products, thus keeping the computational and storage costs in check. This algorithm is based on the next result.

\begin{corol}
The gradient defined in \Cref{prop:grad_expr} satisfies
\begin{equation}
\begin{aligned}
\frac{\partial L}{\partial  \theta}(\theta)=
\frac{\partial C}{\partial \theta}(\theta)\cdot\sum_{i=0}^{N_t-2} \J_{\mathcal{H}^{i}}( C(\theta))^T\bm\delta^{i+1}(C(\theta)) \veq
\end{aligned}
\label{eq:grad_bp}
\end{equation}
where for $C\in (0,1/2)$, the sequence $(\bm\delta^i(C))_{i\in\bi 1, N_t-1\ei}$ is defined by the recurrence
\begin{equation}
\left\lbrace
\begin{aligned}
&\bm\delta^{N_t-1}(C) = \left(\widehat{\bm U}^{N_t-1}(C)-{\bm U}^{N_t-1} \right)\\
&\bm\delta^{i}(C) = \left(\widehat{\bm U}^{i}(C)-{\bm U}^{i} \right) + \J_{\mathcal{H}^{i}}(\widehat{\bm U}^i(C))^T\bm \delta^{i+1}(C), \quad i\in\bi 0, N_t-2\ei \peq
\end{aligned}
\right.
\label{eq:rec_grad_bp}
\end{equation}
\label{prop:grad_bp}
\end{corol}

\begin{proof}
See \Cref{sec:bp_grad}.
\end{proof}

\Cref{prop:grad_bp} provides an alternative expression for computing the gradient function, which in turn yields \Cref{alg:grad_bp}. This last algorithm can be referred to as a Back-Propagation algorithm: we visit each \q{time} $i$ sequentially from the latest to the initial time to compute the gradient.  Consequently, it is no longer possible to compute the density vectors on-the-fly: they must be computed and stored beforehand. Once this is done, each iteration of \Cref{alg:grad_bp} requires the same computational and storage costs, those associated with products between some sparse matrices and vectors. Hence, these costs are much less influenced by the size of the problem, assuming that there is enough storage space for the density vectors.

\begin{minipage}[c]{0.46\textwidth}
  \vspace{0pt}  
\begin{algorithm}[H]
\SetAlgoLined
\KwIn{Parameter $\theta$, Density matrix $\bm U$.}
\KwOut{Gradient $\displaystyle g=\frac{\partial L}{\partial \theta}(\theta)$ of the cost function~\eqref{eq:costf}.}
\dotfill\\
\medskip
Set $\bm G= \bm 0$
and $g = 0$ \;
Set $\widehat{\bm U}^0=\bm U^0$ \;
\For{$i=0, \dots, N_t-2$}{
$\bm G =
 \J_{\mathcal{H}^{i}}(\widehat{\bm U}^i)\cdot \bm G 
 + \J_{\mathcal{H}^{i}}( C(\theta))$ \;
 $\widehat{\bm U}^{i+1}= \mathcal{H}^i(\widehat{\bm U}^i,  C(\theta) \pv \bm U)$ \; 
$g =  g+ \bm G^T \left(\widehat{\bm U}^{i+1}-{\bm U}^{i+1} \right)$ \;
}
$\displaystyle g=\frac{\partial C}{\partial \theta}(\theta)\cdot g$ \;
\KwRet{$g$.}
\caption{Compute the gradient of the cost function (Forward-Propagation).}
\label{alg:grad_fp}
\end{algorithm}
\end{minipage}\hfill
\begin{minipage}[c]{0.46\textwidth}
  \vspace{0pt}  
\begin{algorithm}[H]
\SetAlgoLined
\KwIn{Parameter $\theta$, Density matrix $\bm U$.}
\KwOut{Gradient $\displaystyle g=\frac{\partial L}{\partial \theta}(\theta)$ of the cost function~\eqref{eq:costf}.}
\dotfill\\
\medskip

Set $\widehat{\bm U}^0=\bm U^0$ \;
\For{$i=0, \dots, N_t-2$}{
$\widehat{\bm U}^{i+1}= \mathcal{H}^i(\widehat{\bm U}^i,  C(\theta) \pv \bm U)$ \; 
}
Set $\bm \delta=\left(\widehat{\bm U}^{N_t-1}-{\bm U}^{N_t-1} \right)$,  $g= \J_{\mathcal{H}^{N_t-2}}(\bm C)^T\bm \delta$\;
\For{$l=N_t-2, \dots, 1$}{
$\displaystyle \bm \delta = \left(\widehat{\bm U}^{l}-{\bm U}^{l} \right) + \J_{\mathcal{H}^{l}}(\widehat{\bm U}^{l})^T\bm \delta$\; 
$\displaystyle g = g+ \J_{\mathcal{H}^{l-1}}(C(\theta))^T\bm\delta $\; 
}
$\displaystyle g=\frac{\partial C}{\partial \theta}(\theta)\cdot g$ \;
\KwRet{$\bm g$}
\caption{Compute the gradient of the cost function (Back-Propagation).}
\label{alg:grad_bp}
\end{algorithm}
\end{minipage}

The minimization problem \eqref{eq:min_pbm2} can then be solved using \Cref{alg:opt}. In this algorithm, convergence is understood as the fulfillment of some numerical criterion based on the value of the gradient of the cost function or on the value of the cost function (or both), and chosen by the user. A typical choice is declaring convergence once the norm of the gradient vector is below some predefined threshold. Algorithms allowing to compute descent directions for various gradient descent algorithms can be found in \citep{ruder2016overview}. We can for instance cite the steepest gradient method for which the descent direction $d$ is computed from the gradient $g$ as 
\begin{equation*}
d=-\alpha \cdot g
\end{equation*}
for some fixed step size $\alpha>0$. The (Polak--Ribi\`{e}re) Conjugate gradient algorithm on the other hand uses, at the $t$-th iteration of the process, the  descent direction $d^{(t)}$ given by
\begin{equation*}
d^{(t)}=-g^{(t)} +\left(\frac{\left(g^{(t)}\right)^T\left(g^{(t)}-g^{(t-1)}\right)}{\left(g^{(t-1)}\right)^Tg^{(t-1)}}\right)d^{(t-1)} \veq
\end{equation*}
where $g^{(t)}$ denotes the gradient of the cost function at the $t$-th iteration \citep{nocedal2006numerical}.
This last algorithm is the one used in the numerical applications of this paper.

\begin{algorithm}
\SetAlgoLined
\KwIn{Density matrix $\bm U$ defined from a discretization with step sizes $\Delta t, \Delta x$\;  Initial value $\theta \in\R$\; A routine $\text{grad}$ to compute gradients (Algorithm \ref{alg:grad_fp} or \ref{alg:grad_bp})\; A routine $\text{dir}$ to compute descent directions (Steepest gradient method, conjugate gradient algorithm,...).}
\KwOut{Parameter $C^*$ of the minimization problem~\eqref{eq:min_pbm2}.}
\dotfill\\
\medskip
\While{Convergence is not achieved}{
Compute the gradient  $\displaystyle g=\frac{\partial L}{\partial \theta}(\theta)$ of the cost function~\eqref{eq:costf} with respect to the unrestricted parameters $\theta$:
\begin{equation*}
g=\text{grad}\left(\theta, \bm U\right) \; ;
\end{equation*}

Compute the descent direction from the gradient: $d=\text{dir}(g)$\;

Update the parameter: 
\begin{equation*}
\theta= \theta + d \; ;
\end{equation*}
}
\KwRet{$C^*=C(\theta)$.}
\caption{Find the optimal control parameter of the discrete dynamical system associated with a density matrix (constant case).}
\label{alg:opt}
\end{algorithm}

\begin{prop}
Let $L_{[P_t,P_x]}$ be the cost function defined in~\eqref{eq:costf_b} and assume that the mappings $\mathcal{H}^i$, $i\ge 0$, defined through~\eqref{eq:transfo} and~\eqref{eq:rec}  are smooth with respect to their arguments. Then, following the notations from \Cref{prop:grad_expr}, the gradient of $L$ is given by
\begin{equation}
\frac{\partial L_{[P_t,P_x]}}{\partial \theta}(\theta)
=\frac{\partial C}{\partial \theta}(\theta)\cdot\sum_{i=1}^{N_t-1} \J_{\widehat{\bm U}^{iP_t}}( C(\theta))^T \bm M_{P_x}^T\left(\bm M_{P_x}\widehat{\bm U}^i(C(\theta))-\bm U^i \right), \quad \theta \in\R \veq
\label{eq:gradl_b}
\end{equation}
where $\bm M_{P_x}\in\R^{N_x\times(P_xM_x)}$ is the averaging matrix defined by
\begin{equation*}
[\bm M_{P_x}]_{j,m}=
\begin{cases}
1/P_x & \text{if } m=l+jP_x \text{ with } l\in\bi 0, P_x-1\ei \\
0 & \text{otherwise}
\end{cases}
\end{equation*}
and the Jacobian matrices $\widehat{\bm U}^\iota(C(\theta))$, $\iota \in\bi 0, (N_t-1)P_t\ei$, are once gain obtained through the recurrence~\eqref{eq:rec_jac}.

Equivalently, this gradient can be obtained by
\begin{equation}
\begin{aligned}
\frac{\partial L_{[P_t,P_x]}}{\partial  \theta}(\theta)=
\frac{\partial C}{\partial \theta}(\theta)\cdot\sum_{\iota=0}^{(N_t-1)P_t-1} \J_{\mathcal{H}^{\iota}}( C(\theta))^T\bm\delta^{\iota+1}(C(\theta)) \veq
\end{aligned}
\label{eq:grad_bp_b}
\end{equation}
where for $C\in (0,1/2)$, the sequence $(\bm\delta^\iota(C))_{\iota\in\bi 0, (N_t-1)P_t\ei}$ is defined by the recurrence
\begin{equation}
\left\lbrace
\begin{aligned}
&\bm\delta^{(N_t-1)P_t}(C) = \bm M_{P_x}^T\left(\bm M_{P_x}\widehat{\bm U}^{N_t-1}(C)-{\bm U}^{N_t-1} \right)\\
&\bm\delta^{\iota}(C) = \ind_{(\iota/P_t)\in\bi 0, N_t-2\ei}\bm M_{P_x}^T\left(\bm M_{P_x}\widehat{\bm U}^{\iota/P_t}(C)-{\bm U}^{\iota/P_t} \right) + \J_{\mathcal{H}^{\iota}}(\widehat{\bm U}^\iota(C))^T\bm \delta^{\iota+1}(C), \quad \iota\in\bi 0, (N_t-1)P_t-1\ei
\end{aligned}
\right.
\label{eq:rec_grad_bp_b}
\end{equation}
with $\ind_{A}$ denoting the indicator function of a proposition $A$.

\label{prop:grad_b}
\end{prop}

\begin{proof}
This result is a direct consequence of \Cref{prop:grad_expr} and \Cref{prop:grad_bp} after noting that replacing, in the expression~\eqref{eq:costf} of the cost function $L$, the approximation matrix $\widehat{\bm U}$  by the averaged approximations $\bm M_{P_x}\widehat{\bm U}$ yields the expression~\eqref{eq:costf_b} of the cost function $L_{[P_t,P_x]}$. The chain rule then yields the result.
\end{proof}

Then, \Cref{alg:opt} can be used to solve the minimization problem, after adjusting the algorithms for computing gradients according to the previous proposition, thus yielding \Cref{alg:grad_fp_b,alg:grad_bp_b}.

\begin{minipage}[c]{0.46\textwidth}
  \vspace{0pt}  
\begin{algorithm}[H]
\SetAlgoLined
\KwIn{Parameter $\theta$, Density matrix $\bm U$, Subdivision parameters $P_t, P_x$.}
\KwOut{Gradient $\displaystyle g=\frac{\partial L}{\partial \theta}(\theta)$ of the cost function~\eqref{eq:costf_b}.}
\dotfill\\
\medskip
Set $\bm G= \bm 0$\;
Set $g = 0$ \;
Set $\widehat{\bm U}^0=P_x\bm M_{P_x}^T\bm U^0$ \;
\For{$i=0, \dots, (N_t-1)P_t-1$}{
$\bm G =
 \J_{\mathcal{H}^{i}}(\widehat{\bm U}^i)\cdot \bm G 
 + \J_{\mathcal{H}^{i}}( C(\theta))$ \;
 $\widehat{\bm U}^{i+1}= \mathcal{H}^i(\widehat{\bm U}^i,  C(\theta) \pv \bm U)$ \; 
$g =  g+ \bm G^T \bm M_{P_x}^T\left(\bm M_{P_x}\widehat{\bm U}^{i+1}-{\bm U}^{i+1} \right)$ \;
}
$\displaystyle g=\frac{\partial C}{\partial \theta}(\theta)\cdot g$ \;
\KwRet{$g$.}
\caption{Compute the gradient of the cost function (Forward-Propagation).}
\label{alg:grad_fp_b}
\end{algorithm}
\end{minipage}\hfill
\begin{minipage}[c]{0.46\textwidth}
  \vspace{0pt}  
\begin{algorithm}[H]
\SetAlgoLined
\KwIn{Parameter $\theta$, Density matrix $\bm U$, Subdivision parameters $P_t, P_x$.}
\KwOut{Gradient $\displaystyle g=\frac{\partial L}{\partial \theta}(\theta)$ of the cost function~\eqref{eq:costf_b}.}
\dotfill\\
\medskip
Set $\widehat{\bm U}^0=P_x\bm M_{P_x}^T\bm U^0$ \;
\For{$i=0, \dots, (N_t-1)P_t-1$}{
$\widehat{\bm U}^{i+1}= \mathcal{H}^i(\widehat{\bm U}^i,  C(\theta) \pv \bm U)$ \; 
}
Set $\bm \delta=\bm M_{P_x}^T\left(\bm M_{P_x}\widehat{\bm U}^{N_t-1}-{\bm U}^{N_t-1} \right)$ \;
Set $g= \J_{\mathcal{H}^{N_t-2}}(\bm C)^T\bm \delta$\;
\For{$l=(N_t-1)P_t-1, \dots, 1$}{
$\displaystyle \bm \delta =  \J_{\mathcal{H}^{l}}(\widehat{\bm U}^{l})^T\bm \delta$\; 
\If{$(l/P_t) \in\lbrace 0, \dots, N_t-2\rbrace$}{
$\displaystyle \bm \delta = \bm\delta + \bm M_{P_x}^T\left(\bm M_{P_x}\widehat{\bm U}^{l/P_t}-{\bm U}^{l/P_t} \right)$\; 
}
$\displaystyle g = g+ \J_{\mathcal{H}^{l-1}}(C(\theta))^T\bm\delta $\; 
}
$\displaystyle g=\frac{\partial C}{\partial \theta}(\theta)\cdot g$ \;
\KwRet{$\bm g$}
\caption{Compute the gradient of the cost function (Back-Propagation).}
\label{alg:grad_bp_b}
\end{algorithm}
\end{minipage}

\begin{corol}
Let $\tilde L_{[P_t,P_x]}$ be the cost function defined by~\eqref{eq:costf_t} and associated with a set $I_c\subset \bi 1, N_x-2\ei$ of observed columns. Assume that the mappings $\mathcal{H}^i$, $i\ge 0$, defined through~\eqref{eq:transfo} and~\eqref{eq:rec}  are smooth with respect to their arguments. 

Then, following the notations from \Cref{prop:grad_b}, the gradient of $\tilde L_{[P_t,P_x]}$ is given by \Cref{eq:gradl_b} (or equivalently by~\eqref{eq:grad_bp_b})
after replacing the matrix $\bm M_{P_x}$ by the matrix $\tilde{\bm M}_{P_x}$ defined by
\begin{equation*}
[\tilde{\bm M}_{P_x}]_{j,m}=
\begin{cases}
1/P_x & \text{if } j=j_c\in I_c \text{ and } m=l+j_cP_x   \text{ with } l\in\bi 0, P_x-1\ei  \\
0 & \text{otherwise}
\end{cases} \peq
\end{equation*}

\label{prop:grad_t}
\end{corol}

\begin{proof}
This result is a direct consequence of \Cref{prop:grad_expr} \Cref{prop:grad_b} after noting that replacing, in the expression~\eqref{eq:costf_b} of the cost function $L_{[P_t,P_x]}$, the matrix $\bm M_{P_x}$  by the matrix $\tilde{\bm M}_{P_x}\widehat{\bm U}$ yields an expression equal to the sum of the cost function $\tilde L_{[P_t,P_x]}$ (given in~\eqref{eq:costf_t}) and a term that does not depend on the parameter $\theta$ (but only on the entries of the density matrix $\bm U$). Hence, the gradient of this expression (with respect to the parameters) will be the same as the gradient of $\tilde L_{[P_t,P_x]}$, which gives the result.
\end{proof}

Consequently, the gradient of the cost function~\eqref{eq:costf_t} can be computed using either \Cref{alg:grad_fp_b} or \Cref{alg:grad_bp_b} and accounting for the modification described in \Cref{prop:grad_t}. Hence, gradient-based optimization can once again be considered to minimize this cost function.

\section{Jacobian matrices and gradient computations}

\subsection{Jacobian matrices for the TRM and LxF}
\label{sec:jac}

We derive here, for the TRM and LxF schemes, the expression of the Jacobian matrices needed to compute the gradient of the cost functions considered in this work.

Since the initial condition vector $\bm U^0$, given by~\eqref{eq:ic}, does not depend on $\bm C$, $\J_{\bm U^0}(\bm C)=0$.

For the TRM, \Cref{eq:transfo} is used to derive the expression of the remaining Jacobian matrices involved in the recurrence relation~\eqref{eq:rec_fv_trm}: They are sparse matrices, whose non-zero entries are given by
\begin{align}
\left[\J_{\mathcal{H}^{i}}(\bm U^i)\right]_{j,k}
&=\begin{cases}
C_j^i(1-U_j^i) & \text{if } k=j-1 \\
1-C_j^iU_{j-1}^i+C_{j+1}^i(1-U_{j+1}^i) & \text{if } k=j \\
C_{j+1}^iU_j^i & \text{if } k=j+1
\end{cases}, \quad j\in \bi 1, N_x-2\ei, \quad i\in\bi 0, N_t-2\ei \veq \\
\left[\J_{\mathcal{H}^{i}}(\bm C)\right]_{j,k}
&=\begin{cases}
U_{j-1}^i(1-U_j^i) & \text{if } k=j \\
-U_{j}^i(1-U_{j+1}^i) & \text{if } k=j+1 \\
\end{cases}, \quad j\in \bi 1, N_x-2\ei, \quad i\in\bi 0, N_t-2\ei \peq
\end{align}
Similarly, we get for the LxF scheme
\begin{align}
\left[\J_{\mathcal{H}^{i}}(\bm U^i)\right]_{j,k}
&=\begin{cases}
1/2+C_j^i(1-2U_{j-1}^i)/2 & \text{if } k=j-1 \\
C_j^i(1-2U_{j}^i)/2 -C_{j+1}^i(1-2U_{j}^i)/2 & \text{if } k=j \\
1/2-C_{j+1}^i(1-2U_{j+1}^i)/2 & \text{if } k=j+1
\end{cases}, \quad j\in \bi 1, N_x-2\ei, \quad i\in\bi 0, N_t-2\ei \veq \\
\left[\J_{\mathcal{H}^{i}}(\bm C)\right]_{j,k}
&=\begin{cases}
U_{j-1}^i(1-U_{j-1}^i)+U_{j}^i(1-U_{j}^i) & \text{if } k=j \\
U_{j}^i(1-U_{j}^i)+U_{j+1}^i(1-U_{j+1}^i) & \text{if } k=j+1 \\
\end{cases}, \quad j\in \bi 1, N_x-2\ei, \quad i\in\bi 0, N_t-2\ei \peq
\end{align}

\subsection{Back-propagated gradient}
\label{sec:bp_grad}

We present here the proof of \Cref{prop:grad_bp}.

\begin{proof}
Using \Cref{eq:rec_grad_bp}, we have
\begin{equation*}
\begin{aligned}
\sum_{l=0}^{N_t-2} \J_{\mathcal{H}^{l}}(\bm C)^T\bm\delta^{l+1}
&= 
\sum_{l=0}^{N_t-2} 
	\left(\J_{\bm U^{l+1}}(\bm C)-\J_{\mathcal{H}^{l}}(\bm U^l)\J_{\bm U^l}(\bm C)\right)^T\bm\delta^{l+1} \\
&= \sum_{l=0}^{N_t-2}\J_{\bm U^{l+1}}(\bm C)^T\bm\delta^{l+1}-\sum_{l=0}^{N_t-2}\J_{\bm U^l}(\bm C)^T\J_{\mathcal{H}^{l}}(\bm U^l)^T\bm\delta^{l+1} \\
&= \sum_{l=0}^{N_t-2}\J_{\bm U^{l+1}}(\bm C)^T\bm\delta^{l+1}-\sum_{l=1}^{N_t-2}\J_{\bm U^l}(\bm C)^T\J_{\mathcal{H}^{l}}(\bm U^l)^T\bm\delta^{l+1} \veq
\end{aligned}
\end{equation*}
since $\J_{\bm U^{0}}(\bm C)=\bm 0$.
The definition of the sequence $(\bm \delta^l)_{1\le l\le N_t-1}$ in \Cref{eq:rec_grad_bp} then gives
\begin{equation*}
\begin{aligned}
\sum_{l=0}^{N_t-2} \J_{\mathcal{H}^{l}}(\bm C)^T\bm\delta^{l+1}
&= \sum_{l=0}^{N_t-2}\J_{\bm U^{l+1}}(\bm C)^T\bm\delta^{l+1}-\sum_{l=1}^{N_t-2}\J_{\bm U^l}(\bm C)^T(\bm\delta^l-\bm d^l) \\
&= \J_{\bm U^{1}}(\bm C)^T\bm\delta^{1}+\sum_{l=1}^{N_t-2}\left(\J_{\bm U^{l+1}}(\bm C)^T\bm\delta^{l+1}-\J_{\bm U^l}(\bm C)^T\bm\delta^l\right)+\sum_{l=1}^{N_t-2}\J_{\bm U^l}(\bm C)^T\bm d^l \\
&= \J_{\bm U^{1}}(\bm C)^T\bm\delta^{1}+ \J_{\bm U^{N_t-1}}(\bm C)^T\bm\delta^{N_t-1}-\J_{\bm U^{1}}(\bm C)^T\bm\delta^{1} +\sum_{l=1}^{N_t-2}\J_{\bm U^l}(\bm C)^T\bm d^l \\
&= \J_{\bm U^{N_t-1}}(\bm C)^T\bm\delta^{N_t-1}+\sum_{l=1}^{N_t-2}\J_{\bm U^l}(\bm C)^T\bm d^l \peq
\end{aligned}
\end{equation*}
Finally, since from \Cref{eq:rec_grad_bp}, $\bm\delta^{N_t-1}=\bm d^{N_t-1}$, we get
 \begin{equation*}
\sum_{l=0}^{N_t-2} \J_{\mathcal{H}^{l}}(\bm C)^T\bm\delta^{l+1}
=\sum_{l=1}^{N_t-1} \J_{\bm U^l}(\bm C)^T\bm d^l=\frac{\partial L}{\partial \bm C}(\bm C)\veq
\end{equation*}
according to \Cref{eq:gradl}.
\end{proof}

\end{document}

%% file: fig/reaction_rates.tex
\begin{tikzpicture}[x=\textwidth,y=0.2\textwidth]

\draw[-{Latex[length=3mm,width=2mm]}] (0,1) -- (1,1) node[anchor=west]{$x$};
\draw (0.2,1.05) node[anchor=south]{$x_{j-1}$} -- (0.2,0.95) ;
\draw (0.5,1.05) node[anchor=south] (tot) {$x_{j}$} -- (0.5,0.95) ;
\draw (0.8,1.05) node[anchor=south]{$x_{j+1}$} -- (0.8,0.95) ;

\draw[thick] (0.01,0.75) -- (0.99,0.75);
\draw[thick](0.01,0.45) -- (0.99,0.45);

\draw[thick] (0.01,0.2) -- (0.99,0.2);
\draw[thick](0.01,-0.1) -- (0.99,-0.1);

\draw[dashed,red] (0.05,1)  -- (0.05,-0.3);
\draw[dashed,red] (0.35,1)    -- (0.35,-0.3) node[anchor=north] {{\large $\displaystyle \Phi_{j}+O_{j-1}\mathop{\rightarrow}\limits^{k} O_j + \Phi_{j-1}$}};
\draw[dashed,red] (0.65,1)  -- (0.65,-0.3) node[anchor=north] {{\large $\Phi_{j+1}+O_{j}\mathop{\rightarrow}\limits^{k} O_{j+1} + \Phi_{j}$}};
\draw[dashed,red] (0.95,1)  -- (0.95,-0.3);

\draw (0.2,0.325) node {$j-1$};
\draw (0.5,0.325) node {$j$};
\draw (0.8,0.325) node {$j+1$};

\draw (0.825,0.6) node {\includegraphics[scale=0.7]{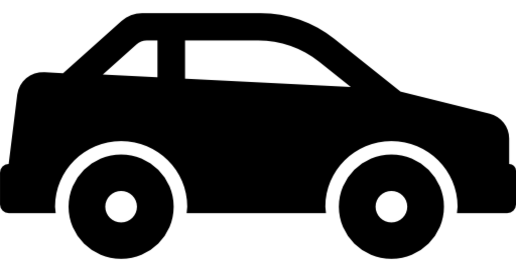}};
\draw (0.1,0.6) node {\includegraphics[scale=0.7]{fig/car.png}};
\draw (0.2,0.6) node {\includegraphics[scale=0.7]{fig/car.png}};
\draw (0.4,0.6) node {\includegraphics[scale=0.7]{fig/car.png}};
\draw (0.5,0.6) node {\includegraphics[scale=0.7]{fig/car.png}};
\draw (0.575,0.6) node {\includegraphics[scale=0.7]{fig/car.png}};

\draw (0.5,0) node {$O_j$};
\draw (0.45,0.08) node {$O_j$};
\draw (0.57,0.08) node {$O_j$};
\draw (0.41,-0.04) node {$\Phi_j$};
\draw (0.61,-0.03) node {$\Phi_j$};
\draw (0.38,0.12) node {$O_j$};
\draw (0.53,0.12) node {$O_j$};

\draw (0.08,0) node {$O_{j-1}$};
\draw (0.15,0.0) node {$\Phi_{j-1}$};
\draw (0.2,0.07) node {$O_{j-1}$};
\draw (0.3,-0.04) node {$\Phi_{j-1}$};
\draw (0.25,-0.03) node {$\Phi_{j-1}$};
\draw (0.1,0.12) node {$O_{j-1}$};
\draw (0.32,0.12) node {$\Phi_{j-1}$};

\draw (0.9,0.05) node {$\Phi_{j+1}$};
\draw (0.77,0.1) node {$\Phi_{j+1}$};
\draw (0.85,-0.02) node {$O_{j+1}$};
\draw (0.72,-0.04) node {$\Phi_{j+1}$};
\draw (0.68,0.1) node {$\Phi_{j+1}$};
\draw (0.82,0.12) node {$O_{j+1}$};


\end{tikzpicture}

%% file: fig/multilevel.tex
\begin{tikzpicture}

\draw[draw,myred] (12.5, 32) node[scale=8] {\underline{Data discretization grid}};

\draw[draw,myred] (92.5, 32) node[scale=8] {\underline{Scheme discretization grid}};

\foreach \x in {0,12,24}{
\pgfmathtruncatemacro{\y}{\x+2}
\draw[line width=8pt] (0,\x) -- (25,\x) -- (25,\y) -- (0,\y) -- (0,\x);
\foreach \l in {5,10,15,20}{
\draw[line width=8pt] (\l,\x) -- (\l,\y);
}
}

\draw[line width=10pt,myred] (0,-2) -- (25,-2) node[midway,below,scale=6.5] (m0) {$N_x=5$};
\draw[line width=10pt,myred] (-2,0) -- (-2,26) node[midway,above,scale=6.5,rotate=90] (m0) {$N_t=3$};


\draw[-latex,line width=15pt,gray] (28,13) -- (38,13) node[midway,below,scale=7.5] (md) {$P_x=4$};

\draw[-latex,line width=15pt,gray] (28,13) -- (38,13) node[midway,above,scale=6.5] (md) {\begin{tabular}{c} Space \\ subdivision\end{tabular} };

\foreach \x in {0,12,24}{
\pgfmathtruncatemacro{\y}{\x+2}
\draw[line width=8pt] (40,\x) -- (65,\x) -- (65,\y) -- (40,\y) -- (40,\x);
\foreach \l in {45,50,55,60}{
\draw[line width=8pt] (\l,\x) -- (\l,\y);
}
}

\foreach \x in {0,12,24}{
\pgfmathtruncatemacro{\y}{\x+2}
\foreach \l in {40,45,50,55,60}{
\foreach \m in {1,...,3}{
\pgfmathsetmacro{\n}{\l+1.25*\m}
\draw[line width=6pt,gray,dashed] (\n,\x) -- (\n,\y);
}
}
}


\draw[-latex,line width=15pt,gray] (68,13) -- (78,13) node[midway,below,scale=7.5] (md) {$P_t=3$};

\draw[-latex,line width=15pt,gray] (68,13) -- (78,13) node[midway,above,scale=6.5] (md) {\begin{tabular}{c} Time \\ subdivision\end{tabular} };

\foreach \x in {0,12,24}{
\pgfmathtruncatemacro{\y}{\x+2}
\draw[line width=8pt] (80,\x) -- (105,\x) -- (105,\y) -- (80,\y) -- (80,\x);
\foreach \l in {85,90,95,100}{
\draw[line width=8pt] (\l,\x) -- (\l,\y);
}
}

\foreach \x in {4,8,16,20}{
\pgfmathtruncatemacro{\y}{\x+2}
\draw[line width=6pt,gray,dashed] (80,\x) -- (105,\x) -- (105,\y) -- (80,\y) -- (80,\x);
\foreach \l in {85,90,95,100}{
\draw[line width=6pt,gray,dashed] (\l,\x) -- (\l,\y);
}
}

\foreach \x in {0,12,24,4,8,16,20}{
\pgfmathtruncatemacro{\y}{\x+2}
\foreach \l in {80,85,90,95,100}{
\foreach \m in {1,...,3}{
\pgfmathsetmacro{\n}{\l+1.25*\m}
\draw[line width=6pt,gray,dashed] (\n,\x) -- (\n,\y);
}
}
}

\draw[line width=10pt,myred] (80,-2) -- (105,-2) node[midway,below,scale=6.5] (m0) {$\widehat{N_x}=20$};
\draw[line width=10pt,myred] (107,0) -- (107,26) node[midway,above,scale=6.5,rotate=270] (m0) {$\widehat{N_t}=7$};

\end{tikzpicture}

%% file: fig/tab_param_all.tex
\begin{minipage}{0.95\textwidth}
\begin{subtable}{0.06\textwidth}
\caption*{}
\end{subtable}
\hspace{-1em}
\begin{subtable}{0.08\textwidth}
\caption*{}
\end{subtable}
\begin{subtable}{0.36\textwidth}
\input{fig/xn_discr.tex}
\caption*{}
\end{subtable}
\hspace{1em}
\begin{subtable}{0.36\textwidth}
\input{fig/xn_discr.tex}
\caption*{}
\end{subtable}\vspace{-4ex}\\

\begin{subtable}{0.06\textwidth}
\caption*{}
\end{subtable}
\hspace{-1em}
\begin{subtable}{0.08\textwidth}
\caption*{}
\end{subtable}
\begin{subtable}{0.36\textwidth}
\input{fig/x_discr.tex}
\caption*{}
\end{subtable}
\hspace{1em}
\begin{subtable}{0.36\textwidth}
\input{fig/x_discr.tex}
\caption*{}
\end{subtable}\vspace{-4ex}\\

\begin{subtable}{0.06\textwidth}
\input{fig/tn_discr.tex}
\caption*{}
\end{subtable}
\hspace{-1em}
\begin{subtable}{0.08\textwidth}
\input{fig/t_discr.tex}
\caption*{}
\end{subtable}
\begin{subtable}{0.36\textwidth}
\input{fig/trm_1_param0.tex}
\caption{TRM: No Space subdivision}
\end{subtable}
\hspace{1em}
\begin{subtable}{0.36\textwidth}
\input{fig/lxf_1_param0.tex}
\caption{LxF: No Space subdivision}
\end{subtable}\vspace{1ex}\\

\begin{subtable}{0.06\textwidth}
\input{fig/tn_discr.tex}
\caption*{}
\end{subtable}
\hspace{-1em}
\begin{subtable}{0.08\textwidth}
\input{fig/t_discr.tex}
\caption*{}
\end{subtable}
\begin{subtable}{0.36\textwidth}
\input{fig/trm_3_param0.tex}
\caption{TRM: $3$ Space subdivisions}
\end{subtable}
\hspace{1em}
\begin{subtable}{0.36\textwidth}
\input{fig/lxf_3_param0.tex}
\caption{LxF: $3$ Space subdivisions}
\end{subtable}\vspace{1ex}\\

\begin{subtable}{0.06\textwidth}
\input{fig/tn_discr.tex}
\caption*{}
\end{subtable}
\hspace{-1em}
\begin{subtable}{0.08\textwidth}
\input{fig/t_discr.tex}
\caption*{}
\end{subtable}
\begin{subtable}{0.36\textwidth}
\input{fig/trm_5_param0.tex}
\caption{TRM: $5$ Space subdivisions}
\end{subtable}
\hspace{1em}
\begin{subtable}{0.36\textwidth}
\input{fig/lxf_5_param0.tex}
\caption{LxF: $5$ Space subdivisions}
\end{subtable}

\end{minipage}%
\begin{minipage}{0.05\textwidth}
\centering
\hspace{-7em}\includegraphics[width=\textwidth]{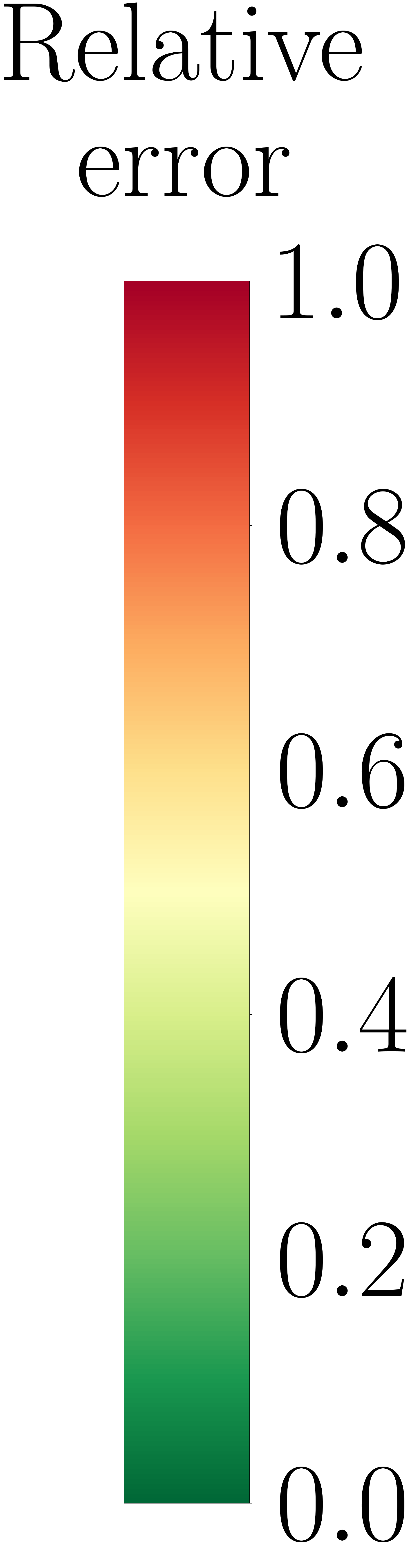}
\end{minipage}

%% file: fig/tab_err_all.tex
\begin{minipage}{0.95\textwidth}
\begin{subtable}{0.06\textwidth}
\caption*{}
\end{subtable}
\hspace{-1em}
\begin{subtable}{0.08\textwidth}
\caption*{}
\end{subtable}
\begin{subtable}{0.36\textwidth}
\input{fig/xn_discr.tex}
\caption*{}
\end{subtable}
\hspace{1em}
\begin{subtable}{0.36\textwidth}
\input{fig/xn_discr.tex}
\caption*{}
\end{subtable}\vspace{-4ex}\\

\begin{subtable}{0.06\textwidth}
\caption*{}
\end{subtable}
\hspace{-1em}
\begin{subtable}{0.08\textwidth}
\caption*{}
\end{subtable}
\begin{subtable}{0.36\textwidth}
\input{fig/x_discr.tex}
\caption*{}
\end{subtable}
\hspace{1em}
\begin{subtable}{0.36\textwidth}
\input{fig/x_discr.tex}
\caption*{}
\end{subtable}\vspace{-4ex}\\

\begin{subtable}{0.06\textwidth}
\input{fig/tn_discr.tex}
\caption*{}
\end{subtable}
\hspace{-1em}
\begin{subtable}{0.08\textwidth}
\input{fig/t_discr.tex}
\caption*{}
\end{subtable}
\begin{subtable}{0.36\textwidth}
\input{fig/trm_1_rmse0.tex}
\caption{TRM: No Space subdivision}
\end{subtable}
\hspace{1em}
\begin{subtable}{0.36\textwidth}
\input{fig/lxf_1_rmse0.tex}
\caption{LxF: No Space subdivision}
\end{subtable}\vspace{1ex}\\

\begin{subtable}{0.06\textwidth}
\input{fig/tn_discr.tex}
\caption*{}
\end{subtable}
\hspace{-1em}
\begin{subtable}{0.08\textwidth}
\input{fig/t_discr.tex}
\caption*{}
\end{subtable}
\begin{subtable}{0.36\textwidth}
\input{fig/trm_3_rmse0.tex}
\caption{TRM: $3$ Space subdivisions}
\end{subtable}
\hspace{1em}
\begin{subtable}{0.36\textwidth}
\input{fig/lxf_3_rmse0.tex}
\caption{LxF: $3$ Space subdivisions}
\end{subtable}\vspace{1ex}\\

\begin{subtable}{0.06\textwidth}
\input{fig/tn_discr.tex}
\caption*{}
\end{subtable}
\hspace{-1em}
\begin{subtable}{0.08\textwidth}
\input{fig/t_discr.tex}
\caption*{}
\end{subtable}
\begin{subtable}{0.36\textwidth}
\input{fig/trm_5_rmse0.tex}
\caption{TRM: $5$ Space subdivisions}
\end{subtable}
\hspace{1em}
\begin{subtable}{0.36\textwidth}
\input{fig/lxf_5_rmse0.tex}
\caption{LxF: $5$ Space subdivisions}
\end{subtable}

\end{minipage}%
\begin{minipage}{0.05\textwidth}
\centering
\hspace{-7em}\includegraphics[width=\textwidth]{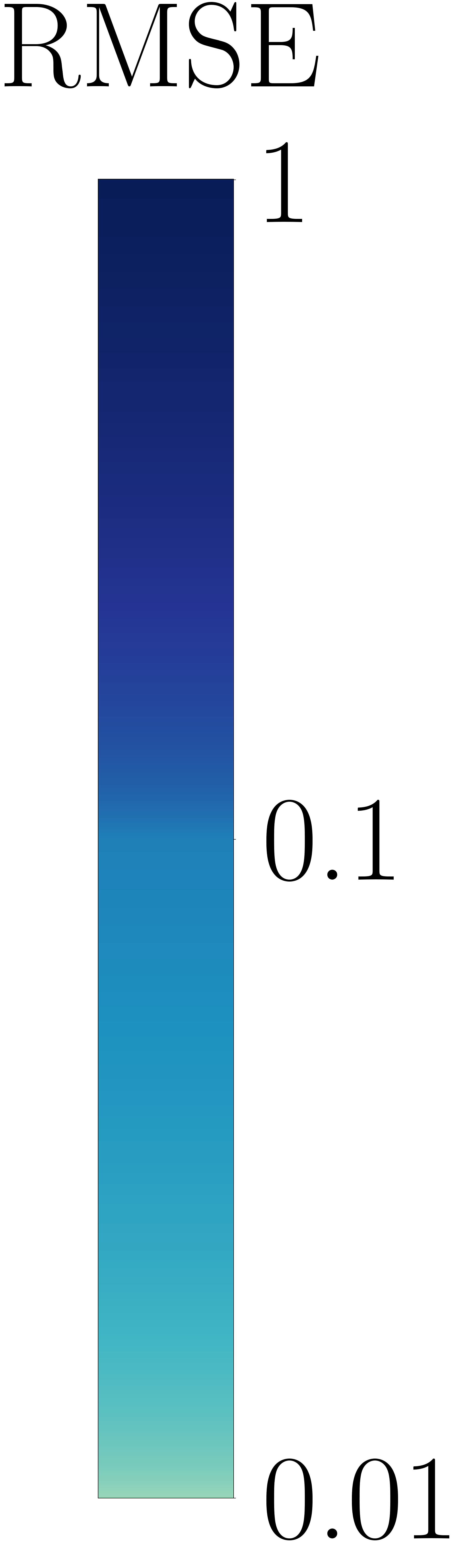}
\end{minipage}

%% file: fig/tab_param_mid.tex
\begin{minipage}{0.95\textwidth}
\begin{subtable}{0.06\textwidth}
\caption*{}
\end{subtable}
\hspace{-1em}
\begin{subtable}{0.08\textwidth}
\caption*{}
\end{subtable}
\begin{subtable}{0.36\textwidth}
\input{fig/xn_discr.tex}
\caption*{}
\end{subtable}
\hspace{1em}
\begin{subtable}{0.36\textwidth}
\input{fig/xn_discr.tex}
\caption*{}
\end{subtable}\vspace{-4ex}\\

\begin{subtable}{0.06\textwidth}
\caption*{}
\end{subtable}
\hspace{-1em}
\begin{subtable}{0.08\textwidth}
\caption*{}
\end{subtable}
\begin{subtable}{0.36\textwidth}
\input{fig/x_discr.tex}
\caption*{}
\end{subtable}
\hspace{1em}
\begin{subtable}{0.36\textwidth}
\input{fig/x_discr.tex}
\caption*{}
\end{subtable}\vspace{-4ex}\\

\begin{subtable}{0.06\textwidth}
\input{fig/tn_discr.tex}
\caption*{}
\end{subtable}
\hspace{-1em}
\begin{subtable}{0.08\textwidth}
\input{fig/t_discr.tex}
\caption*{}
\end{subtable}
\begin{subtable}{0.36\textwidth}
\input{fig/trm_1_param1.tex}
\caption{TRM: No Space subdivision}
\end{subtable}
\hspace{1em}
\begin{subtable}{0.36\textwidth}
\input{fig/lxf_1_param1.tex}
\caption{LxF: No Space subdivision}
\end{subtable}\vspace{1ex}\\

\begin{subtable}{0.06\textwidth}
\input{fig/tn_discr.tex}
\caption*{}
\end{subtable}
\hspace{-1em}
\begin{subtable}{0.08\textwidth}
\input{fig/t_discr.tex}
\caption*{}
\end{subtable}
\begin{subtable}{0.36\textwidth}
\input{fig/trm_3_param1.tex}
\caption{TRM: $3$ Space subdivisions}
\end{subtable}
\hspace{1em}
\begin{subtable}{0.36\textwidth}
\input{fig/lxf_3_param1.tex}
\caption{LxF: $3$ Space subdivisions}
\end{subtable}\vspace{1ex}\\

\begin{subtable}{0.06\textwidth}
\input{fig/tn_discr.tex}
\caption*{}
\end{subtable}
\hspace{-1em}
\begin{subtable}{0.08\textwidth}
\input{fig/t_discr.tex}
\caption*{}
\end{subtable}
\begin{subtable}{0.36\textwidth}
\input{fig/trm_5_param1.tex}
\caption{TRM: $5$ Space subdivisions}
\end{subtable}
\hspace{1em}
\begin{subtable}{0.36\textwidth}
\input{fig/lxf_5_param1.tex}
\caption{LxF: $5$ Space subdivisions}
\end{subtable}

\end{minipage}%
\begin{minipage}{0.05\textwidth}
\centering
\hspace{-7em}\includegraphics[width=\textwidth]{fig/param_clr.png}
\end{minipage}

%% file: fig/tab_err_mid.tex
\begin{minipage}{0.95\textwidth}
\begin{subtable}{0.06\textwidth}
\caption*{}
\end{subtable}
\hspace{-1em}
\begin{subtable}{0.08\textwidth}
\caption*{}
\end{subtable}
\begin{subtable}{0.36\textwidth}
\input{fig/xn_discr.tex}
\caption*{}
\end{subtable}
\hspace{1em}
\begin{subtable}{0.36\textwidth}
\input{fig/xn_discr.tex}
\caption*{}
\end{subtable}\vspace{-4ex}\\

\begin{subtable}{0.06\textwidth}
\caption*{}
\end{subtable}
\hspace{-1em}
\begin{subtable}{0.08\textwidth}
\caption*{}
\end{subtable}
\begin{subtable}{0.36\textwidth}
\input{fig/x_discr.tex}
\caption*{}
\end{subtable}
\hspace{1em}
\begin{subtable}{0.36\textwidth}
\input{fig/x_discr.tex}
\caption*{}
\end{subtable}\vspace{-4ex}\\

\begin{subtable}{0.06\textwidth}
\input{fig/tn_discr.tex}
\caption*{}
\end{subtable}
\hspace{-1em}
\begin{subtable}{0.08\textwidth}
\input{fig/t_discr.tex}
\caption*{}
\end{subtable}
\begin{subtable}{0.36\textwidth}
\input{fig/trm_1_rmse1.tex}
\caption{TRM: No Space subdivision}
\end{subtable}
\hspace{1em}
\begin{subtable}{0.36\textwidth}
\input{fig/lxf_1_rmse1.tex}
\caption{LxF: No Space subdivision}
\end{subtable}\vspace{1ex}\\

\begin{subtable}{0.06\textwidth}
\input{fig/tn_discr.tex}
\caption*{}
\end{subtable}
\hspace{-1em}
\begin{subtable}{0.08\textwidth}
\input{fig/t_discr.tex}
\caption*{}
\end{subtable}
\begin{subtable}{0.36\textwidth}
\input{fig/trm_3_rmse1.tex}
\caption{TRM: $3$ Space subdivisions}
\end{subtable}
\hspace{1em}
\begin{subtable}{0.36\textwidth}
\input{fig/lxf_3_rmse1.tex}
\caption{LxF: $3$ Space subdivisions}
\end{subtable}\vspace{1ex}\\

\begin{subtable}{0.06\textwidth}
\input{fig/tn_discr.tex}
\caption*{}
\end{subtable}
\hspace{-1em}
\begin{subtable}{0.08\textwidth}
\input{fig/t_discr.tex}
\caption*{}
\end{subtable}
\begin{subtable}{0.36\textwidth}
\input{fig/trm_5_rmse1.tex}
\caption{TRM: $5$ Space subdivisions}
\end{subtable}
\hspace{1em}
\begin{subtable}{0.36\textwidth}
\input{fig/lxf_5_rmse1.tex}
\caption{LxF: $5$ Space subdivisions}
\end{subtable}

\end{minipage}%
\begin{minipage}{0.05\textwidth}
\centering
\hspace{-7em}\includegraphics[width=\textwidth]{fig/param_clrb.png}
\end{minipage}